\documentclass[10pt]{article}
\usepackage[latin1]{inputenc}
\usepackage[english]{babel}
\usepackage{latexsym}
\usepackage{amsmath}
\usepackage{amsfonts}
\usepackage{amssymb}
\usepackage{amsthm}
\usepackage{amsthm}

\ifx\macrosloaded\relax\endinput\else\let\macrosloaded\relax\fi
\newtheorem{thm}{Theorem}[section]
\newtheorem{prop}{Proposition}[section]

\newtheorem{cor}{Corollary}[section]
\newtheorem{defn}{Definition}[section]

\newtheorem{remark}{Remark}[section]

\catcode`\@=11
\newcommand{\eqinsec}{\relax\@addtoreset{equation}{section}}
\renewcommand{\theequation}{\ifx\showlabels\iftrue\the\id\else\thesection.\arabic{equation}\fi}
\catcode`\@=13
\newcounter{supeq}
\newenvironment{subeq}
\stepcounter{equation} 
\def\theequation{\ifx\showlabels\iftrue\the\id\else\thesection.\arabic{equation}\fi}
\newtoks\id
\newcommand{\eqlabel}[1]{\label{#1}\global\id={(#1)}}
\newcommand{\medn}{\medskip\noindent}
\newcommand{\tr}{\mbox{tr}}
\newcommand{\be}{\begin{equation}}
\newcommand{\eeq}{\end{equation}}
\newcommand{\bea}{\begin{eqnarray}}
\newcommand{\eea}{\end{eqnarray}}
\newcommand{\beaa}{\begin{eqnarray*}}
\newcommand{\eeaa}{\end{eqnarray*}}
\newcommand{\bseq}{\begin{subeq}}
\newcommand{\eseq}{\end{subeq}}
\newcommand{\ba}{\begin{array}}
\newcommand{\ea}{\end{array}}
\newcommand{\eql}{\eqlabel}

\def \rectangle#1#2{\hbox{\vrule\vbox to #2
{\hrule\hbox to
#1{\hfil}\vfil\hrule}\vrule}}
\newcommand{\edd}{\end{document}}

\renewcommand{\c}{\cdot}
\newcommand{\NI}{\noindent}

\newcommand{\Si}{\Sigma}
\newcommand{\ga}{\gamma}
\newcommand{\Ga}{\Gamma}

\newcommand{\EEb}{\mbox{${\cal E}_1$}}
\newcommand{\EEbb}{\mbox{${\cal E}_2$}}
\newcommand{\rr}{\mbox{${\bf R}$}}

\newcommand{\dd}{\mbox{${\bf D}$}}

\newcommand{\nab}{\mbox{$\nabla$}}
\newcommand{\nabb}{\mbox{$\nabla \mkern-13mu /$\,}}

\newcommand{\ddb}{\mbox{$\dd \mkern-13mu /$\,}}

\newcommand{\pr}{\partial}

\newcommand{\hot}{\widehat{\otimes}}

\newcommand{\Lie}{\mbox{$\cal L$}}

\newcommand{\lie}{\hat{\Lie}}
\newcommand{\lieb}{\mbox{$\lie\mkern-11mu /$\,}}
\newcommand{\nn}{\nonumber}

\newcommand{\chib}{\underline{\chi}}

\newcommand{\de}{\delta}
\newcommand{\De}{\Delta}
\newcommand{\e}{\epsilon}

\newcommand{\chih}{\hat{\chi}}

\newcommand{\chibh}{\underline{\hat{\chi}}}

\newcommand{\Hb}{\und{H}}

\newcommand{\ab}{\und{a}}

\newcommand{\und}[1]{\underline{#1}}

\newcommand{\s}{\mbox{$\mkern 7mu$}}

\newcommand{\Cb}{\und{C}}

\newcommand{\ub}{{\und{u}}}
\renewcommand{\c}{\cdot}
\newcommand{\D}{{\cal D}}
\newcommand{\M}{{\cal M}}

\renewcommand{\aa}{\underline{\alpha}}
\newcommand{\bb}{\underline{\beta}}
\renewcommand{\a}{\alpha}
\renewcommand{\b}{\beta}
\newcommand{\dual}{\mbox{}^{\star}\!}
\newcommand{\rdual}{^{\star}}

\newcommand{\si}{\sigma}

\newcommand{\ro}{\rho}

\newcommand{\z}{\zeta}
\newcommand{\divv}{\mbox{div}\mkern-19mu /\,\,\,\,}

\newcommand{\pp}{\mbox{p}\mkern-8mu /}

\newcommand{\pih}{\hat{\pi}}

\newcommand{\om}{\omega}

\newcommand{\omb}{\underline{\omega}}

\newcommand{\etab}{\underline{\eta}}
\newcommand{\la}{\lambda}

\newcommand{\ep}{\epsilon}
\newcommand{\ddd}{\nabb}
\newcommand{\dddd}{{\bf D} \mkern-13mu /\,}

\newcommand{\QQ}{{\tilde{\cal Q}}}
\newcommand{\QQb}{\underline{\tilde{{\cal Q}}}}

\newcommand{\varep}{\varepsilon_0}
\newcommand{\ttau}{\tau_{-}}
\newcommand{\tttau}{\tau_{+}}

\renewcommand {\div}{\mbox{div$\mkern 1mu$}}
\newcommand{\Div}{\mbox{{\bf Div}}}

\newcommand{\acc}{\bar{K}}

\newcommand{\Th}{\Theta}

\newcommand{\ML}{\!\!\!\!\!\!\!\!\!}
\newcommand{\un}{\underline}
\newcommand{\p}{\partial}
\newcommand{\tW}{\tilde{W}}

\begin{document}
\author{Giulio Caciotta , Tiziana Raparelli }
\title{\LARGE{Asymptotic Behaviour of Zero Mass spin 2 Fields propagating in the external region of Kerr spacetime}}
\date{}
\maketitle
\begin{abstract} 
\NI{After describing the inhomogeneous equation suitable to describe $W_{\nu\mu\ro \si}$, a solution of the linearized version of the full quasilinear equation for the conformal part of the Riemann tensor connected to theperturbations of the Kerr spacetime far from the origin, we find the right decays we have to impose to the source term to obtain the peeling decays for this linearized  solution. We basically follow the ideas of the Christodoulou Klainerman approach, \cite{Ch-Kl:book} and \cite{Kl-Ni1}.
This result requires some new detailed estimates which could be considered a useful result by themselves.} 
\end{abstract}
\section{Introduction and results}\label{S1}
The problem we are studying in this paper concerns the asymptotic behavior of the solutions of a
particular class of equations, the spin 2 zero-rest mass  inhomogeneous equations, propagating in the Kerr spacetime.

\NI Let us state here a first, rough, version of the main theorem we want to prove.A precise version of it will be given at the end of the introduction:
\begin{thm}[Peeling Theorem]\label{peeling1}
Let $W$ be a Weyl field solution of the inhomogeneous massless spin 2 equations propagating in the Kerr spacetime \footnote{$D_{(0)}$ the covariant derivative associated to the Kerr metric } ,
\bea
&&D_{(0)}^{\nu}W_{\nu\mu\ro \si}=h_{\mu\ro\si}\ .
\eea
 Let us assume that $W=R_{Kerr}+\delta R$, with $R_{Kerr}$ the conformal part of the Riemann tensor of the Kerr spacetime and $\delta R$  small with respect to $R_{Kerr}$.\footnote{In suitable Sobolev norms.}  

\NI Let us assume that the inhomogeneous term $h_{\nu\ro\si}$ satisfies appropriate decay conditions in $r$,\footnote{ $r$ the radius in Boyer-Lyndquist coordinates}. 
\NI Let the null components of  $W$ on  the spacelike initial data surface $\Sigma_0$ \footnote{See section 2 for the precise definition of $\{\a,\aa,\b,\bb,\ro,\si$\} the null components of $W$.} satisfy suitable decays and smallness conditions.\footnote{ See theorem 1.3 and the definition of  the smallness condition $\cal{J}$ in 3.408 of  \cite{Ca-Ni}} Then the $W$ null components  satisfy the following decays in accordance with the peeling theorem:
\bea
&&\lim_{u\rightarrow\infty}r|u|^{(4+\ep)}\aa=C_0\nn\\
&&\lim_{u\rightarrow\infty}r^2|u|^{(3+\ep)}\bb=C_0\nn\\
&&\lim_{u\rightarrow\infty}r^3\ro=C_0\nn\\
&&\lim_{u\rightarrow\infty}r^3\si=C_0\eql{peel1}\\
&&\lim_{u\rightarrow\infty}r^4|u|^{(1+\ep)}\b=C_0\nn\\
&&\lim_{u\rightarrow\infty}r^{5}|u|^{\ep}|\a|\leq C_0\ .\nn
\eea
with $\e>0$ , $C_0$  a constant depending on the initial data and $u$ the null coordinate associated to the incoming cones of the foliation.
Moreover the decays we have to impose on the inomogeneous term $h_{\nu\ro\si}$ are compatible with the decays we expect for it in the case $h$ represents 
the perturbation term in the linearization of the Kerr spacetime.
\end{thm}

\NI Let us better explain this yet vague statement.

\NI This result is connected to the problem of the global stability of Kerr spacetime and to the, presumable, asymptotically simplicity of suitable nonlinear perturbations of it. 
To see this connection let us briefly recall what are the known facts and the main open problems.
\smallskip

\NI To prove the global stability for the Kerr spacetime is a very difficult and open problem. The more difficult issue is, of course, that of proving the existence of solutions of the vacuum Einstein equations with initial data ``near to Kerr" in the region 
near the event horizont up to the ergosphere, which are also  unknowns of the problem.\footnote{For small perturbations the horizon is expected to stay ``near" to the Kerr horizon for $r=m+\sqrt{m^2-a^2}$.}

\NI What is known, up to now, relatively to this region are some relevant uniform boundedness results for solutions to the wave equation in the Kerr spacetime as a background spacetime, see the paper of Dafermos-Rodnianski, \cite{Daf-Rod:Kerr} and references therein.\footnote{See also for the Schwartchild case, \cite{Blue} and references therein.}
\smallskip

\NI If we consider the existence in a region sufficiently far from the Kerr horizon this result is included (for angular momentum sufficiently small) in the version of Minkowski stability result proved by S.Klainerman and F.Nicol\`o with initial data near the flat ones (see \cite{Kl-Ni1}). Nevertheless the global solution proved in \cite{Kl-Ni1} does not satisfy the decay of the Riemann components suggested by the ``Peeling theorem", \cite{Wald}.
 
 \NI In fact in that result, the null asymptotic behaviour of some of the null components of the Riemann tensor, specifically the $\a$ and the $\b$ components, is different from the one expected from the ``Peeling Theorem" as their proved decay is slower. 
 
 \NI In a subsequent paper, \cite{Kl-Ni2}, S.Klainerman and F.Nicol\`o proved for these components an asymptotic behaviour consistent with the ``Peeling Theorem" under stronger asymptotic conditions for the initial data. Unfortunately these conditions exclude initial data ``near to Kerr", the main difficulty being connected to the fact that, as the angular moment $J$ of the Kerr spacetime is different from zero, the Kerr initial data have metric components decaying as $r^{-2}$, a decay not sufficient to guarantee the peeling in  \cite{Kl-Ni2}.

\NI Recently one of the present authors, G.C., together with F.Nicol\`o has improved a previous result of F.Nicol\`o, see \cite{NI:Kerrpeeling}proving that also with initial data  $W_0$ ``near" to Kerr, independently from the angular momentum, the global solution in the "external region" \footnote{the external region is defined as $\frac{M}{R^0}<\ep$ with $\ep$ sufficiently small} spacetime satisfies the decay expected from the ``Peeling Theorem" provided the perturbation is given in an suitable way, see \cite{Ca-Ni}, namely are allowed all the perturbations such that $\Lie_{\tilde{T}_0}W_0$ decay sufficiently fast, where $\tilde{T}_0$ is the quasi-Killing vector field corresponding to $T_0$ in Kerr. 



\smallskip

\NI To prove this result one has first to prove the global existence of the external region. The technical machinery used in these works is not  powerful enough to allow us to treat the region near the event horizont, nevertheless this result allows to disentagle the size of the external region from the value of $J$ or in other words we can prove the peeling in the ``external region" perturbing around the (external) Kerr spacetime, no matter which is the value of $J\leq M^2$. 

\NI This result is difficult to obtain for different reasons, all based on the fact that is more complicated to perturb around a curved non spherical symmetric spacetime than to with respect to the Minkowski one. The two main difficulties are: 
\medskip

\NI a) As the Kerr spacetime is only axially symmetric we cannot consider anymore all the generators of the rotation group as approximated Killing vector fields, a fact which does not allow to use all the  $\cal Q$ "energy"  norms defined in \cite{Kl-Ni1}, which are crucial in this approach.
\smallskip

\NI b) To perturb non linearly around a given spacetime one should look first to a linearization of the problem. In this case, as we are not expanding around the flat spacetime to define the ``natural" linearization is not  at all obvious.
\smallskip
 
\NI The first problem is overcame considering perturbations far from the origin, in such a way all the Poincare group of Killing vectorfields of Minkowski can be recostructed exploiting the corresponding "quasi" 
Killing vectorfieds. For what concern the second  problem, in \cite{Ca-Ni}, we have overcame it considering the $\mathcal{\tilde{Q}}$ norms associated  to ${\tilde R}\equiv \Lie_{T_0}R$ with $T_0$ the corresponding of the timelike Killing vector in Kerr. This trick allows
us to eliminate the Kerr contribution. This approach restrict our class of perturbations to initial data such that the $\mathcal{\tilde{Q}}$ are small and consequently give a bound on suitable Sobolev norms associated to $\Lie_{T_0}R$. This is not completely satisfactory as we would obtain estimates directly for $ \de R$, the perturbations of Kerr spacetime and not for the $\Lie_{T_0}$ derivatives of them.\footnote{see \cite{Ca-Ni}, for the definition of the corresponding $\mathcal{\tilde{Q}}$ norms.}

\medskip

\NI Anyway, if we want to follow the \cite{Ch-Kl:book} and  \cite{Kl-Ni1} approach, a required step\footnote{See, nevertheless the different approach to the global existence around the Minkowski spacetime due to Linblad and Rodnianski, \cite{Li-Ro}.} is to decouple the problem in two parts: First we consider the structure equation for the connection coefficients where the null Riemann components are considered assigned, second we consider the Bianchi equation for the null Riemann components where the connection coefficients are considered as an external source, see the next section for more details. 

\NI Hence, to possibly enlarge the class of perturbations  a first step is to consider only the ``decoupled"  part of the problem associated to the Bianchi equations. This would require to consider $W$, the conformal part of $R$ satisfying the Bianchi equations:
\bea
D_{g}^{\nu}W_{\nu\mu\ro\si}=0 \label{Bianchi1}
\eea
with $D_g$ the covariant derivative with respect to $g=g_{kerr}+\delta g$ , $\delta g$ a small perturbation with respect to $g$ \footnote{$\de g$ corresponding to the $\de W$ perturbation. Clearly the connection between these two quantities is very involved}.

\NI In this paper we consider a linearized version of this second part of the problem, which is still a delicate part to deal with. 
\NI More precisely we consider $W$ as an independent Weyl field and consider the covariant derivative $D_g$ relative to the metric of the spacetime we want to perturb. In this case to perturb around Kerr we proceed as follows:

\NI Let us denote with $R$ the Riemann tensor for a vacuum spacetime near to Kerr, we can write:
\bea
R_{\mu\nu\ro\si}\equiv {R_{(Kerr)}}_{\mu\nu\ro\si}+{\de R}_{\mu\nu\ro\si}\eql{Riempert}
\eea
We assume that ${\de R}_{\mu\nu\ro\si}$ is small, let us say of order $\ep$, with respect to ${R}_{(Kerr)}$. Let the associated metric be $g=g_{(Kerr)}\!+\!\de g$,  $\D_g$ the covariant derivative associated to $g$, $\D_{(0)}$ the covariant derivative associated to $g_{(Kerr)}$ and $\de g$ of order $\ep$ with respect of $g_{(Kerr)}$.

\NI Substituting it in \ref{Bianchi1} and keeping only the terms of first order in $\ep$ we do not obtain the homogeneous spin 2 equation in the background spacetime
\bea 
D_{(0)}^{\mu}\de R_{\mu\nu\ro\si}=0\ .\eql{Bianchi2}\nn
\eea
This is due to the fact that the Riemann tensor of the curved background spacetime does not vanish and give rise to an equation of the following kind, written in a very schematic way:
\bea
D_{(0)}^\mu {\de R}_{\mu\nu\ro\si}=\big({(\Ga_0-\Ga)R_{(Kerr)}}\big)_{\nu\ro\si}+O(\ep^2)\ ,\eql{Bianchi3}
\eea
where $(\Ga_0-\Ga)$ denotes the difference between the Christoffel symbols $\Ga$ of the metric $g=g_{(Kerr)}\!+\!\de g$ and those, $\Ga_0$, associated to the Kerr metric,  $g_{(Kerr)}$. The term  $(\Ga_0-\Ga)R_{(Kerr)}$ is a term of order $\ep$ and cannot be neglected. The difference with a linearization ``around zero", that is around the Minkowski spacetime, is that in that case the analogous of $R_{(Kerr)}$ is identically zero and, therefore, equation \ref{Bianchi2} would be the correct linearization. Moreover if we insist to study equation \ref{Bianchi2} a problem immediately arises as, due to the Buchdal constraint the ssolutions of \ref{Bianchi2} are static and cannot be interpreted as a linear perturbation of our problem, see for instance the detailed discussion in \cite{Blue}. 
\smallskip

\NI Our present goal is, therefore, to look for the solution of a linear inhomogeneous equation for a Weyl field $W$ which we want to interpret as mimicking equation \ref{Bianchi1}.  Therefore, as the previous argument suggests, we  have to study an equation of the following kind:
\bea
D_{(0)}^\mu W_{\mu\nu\ro\si}=h_{\nu\ro\si}(W)\ ,\eql{Bianchi4}
\eea
with $h$ corresponding to the term  $\big({(\Ga_0-\Ga)R_{(Kerr)}}\big)$ in order to mimick equation \ref{Bianchi3}. Nevertheless let us observe that  equation \ref{Bianchi3}, apart from the $O(\ep^2)$ terms, is a linearization in the metric $g$ and the dependance of $(\Ga_0-\Ga)$ on $\de R$ is a very indirect one. Therefore to choose appropriately  the tensor $h_{\nu\ro\si}$ we have to see in more detail how the Bianchi equations have been written in \cite{Kl-Ni1} and \cite{Ch-Kl:book} to obtain a global solution of \ref{Bianchi4} in the external region. 
\smallskip

\subsection{The general strategy}

\NI The proof of the global existence (for the external region) provided in \cite{Kl-Ni1} was based on the construction of a coordinate set $(u, \ub, \theta, \phi) $ adapted to a (double) foliation made by null hypersurfaces denoted by $\{C(u)\}$ and $\{\Cb(\ub)\}$ corresponding, basically to the outgoing and incoming cones of the Minkowski spacetime.  On these null hypersurfaces an adapted null frame $\{e_{\mu}\}$ ,$\mu=1,...4$ and the corresponding connection coefficients were defined, see \cite{Kl-Ni1}, Chapter 3, for all the details. Moreover we define \[S(u,v)=C(u)\cap\Cb(v)\ .\]
The connection coefficients we denote by $\cal{O}=\{M,\underline{M},H, \underline{H}\}$ satisfy the structure equations with respect to $\{e_{\mu}\}$ which are of  two types:

a) transport equations  on the incoming and outgoing cones,
\bea
&&\frac{\partial M}{\partial \ub}+\tr\chi{M}= H\c M+(1+M)\c R+[error]\label{eq:transportM}\\
&&\frac{\partial \underline{M}}{\partial u}+\tr\chib\underline{M}= \Hb\c \underline{M}+(1+ \underline{M})\c R+[error]\nn
\eea
b) elliptic Hodge systems on the two dimensional surfaces $S(u,v)$
\bea
{\ddb} H&=&R+M+[error]\label{eq:HodgeH}\\
{\ddb}\underline{H}&=& R+\underline{M}+[error]\nn
\eea
where ${\ddb}$ is the covariant derivative associated to the metric induced on $S(u,v)$ and $[error]$ are terms which, under some smallness assumptions, can be neglected. This picture is, obviously very schematic, all the details are given in \cite{Kl-Ni1}, Chapter 4.
What is important to point out here is that the strategy to obtain norm estimates for the connection coefficients from these equations is to supplement them with the Bianchi equations for $W$ the conformal part of  the Riemann tensor $R$.
\bea
D^{\mu}W_{\mu\nu\ro\si}=0 \eql{bianchi again}
\eea

\smallskip

\NI In \cite{Ch-Kl:book} and in  \cite{Kl-Ni1} the existence problem were faced in the following way: first one proves a local existence result, then considers $\cal {K}$ the largest possible (in principle finite) region where the norms of the connection coefficients and of the Riemann tensor satisfy some smallness conditions. Then one proves that this region can be extended which, to avoid a contradictionthis implies that the region is unbounded and the existence result is, therefore, obtained. To achieve this result the authors consider the Riemann tensor $R$ in the structure equations \ref{eq:transportM}, \ref{eq:HodgeH}  as an external source satisfying suitable conditions of smallness and decays  while in the Bianchi equations, see \ref{nullBian},  the connection coefficients are considered assigned and satisfying the appropriate bounds.

\NI This sort of decoupling between the connection coefficients and the Riemann components reminds of a linearization, even if, strictly speaking, it is not.\footnote{The Bianchi equations are from this point of view linear equations (it will be not the case for the ones obtained linearizing in the metric)  while the structure equations still have non linear terms.}
\medskip

\NI Nevertheless, as said in the previous section, this suggests that to control the solutions of a linearized version of the Bianchi equations should be relevant to treat the global problem.

\NI To  better define this problem, namely, to make some assumptions on the tensor $h(W)$ which are consistent with the interpretation of equations \ref{Bianchi4} as the linearized version of the Bianchi equations in Kerr spacetime  we have to show how the Bianchi equation have been written as transport equation for the null Riemann components in \cite{Ch-Kl:book} and in  \cite{Kl-Ni1}. Considered a null frame $\{e_3,e_4,e_{\theta},e_{\phi}\}$ they have the following form:

\bea
&&\dddd_4\aa+\frac{1}{2}\tr\chi\aa=-\nabb\hot\bb
+\left[4\om\aa-3(\hat{\chib}\ro-\dual\hat{\chib}\si)+(\zeta-4\etab)\hot\bb\right]\nn\\
&&\dddd_3\bb+2\tr\chib\s\bb=-\divv\aa-\left[2\omb \bb+(-2\zeta+ \eta)\c\aa\right]\nn\\
&&\dddd_4\bb+\tr\chi\bb=-\nabb \ro+\left[2\om\bb+2\hat{\chib}\c\b+\dual\nabb\si
-3(\etab\ro-\dual\etab\si)\right]\nn\\
&&\dd_3\ro+\frac{3}{2}\tr \chib\ro=-\divv\bb-\left[\frac{1}{2}\hat{\chi}\c\aa
-\zeta\c \bb+2\eta\c\bb\right]\nn\\
&&\dd_4\ro+\frac{3}{2}\tr\chi\ro=\divv\b-\left[\frac{1}{2}\hat{\chib}\c\a
 -\zeta\c\b-2\etab\c\b\right]\eql{nullBian}\\
&&\dd_3\si+\frac{3}{2}\tr \chib\si=-\divv\dual\bb+\left[\frac{1}{2}\hat{\chi}\c\dual\aa
-\zeta\c\dual\bb-2\eta\c\dual\bb\right]\nn\\
&&\dd_4\si+\frac{3}{2}\tr\chi\si=-\divv\dual\b+\left[\frac{1}{2}\hat{\chib}\c\dual\a
-\zeta\c\dual\b-2\etab\c\dual\b\right]\nn\\
&&\dddd_3\b+\tr\chib\b=\nabb\ro+\left[2\omb\b+\dual\nabb\si+2\hat{\chi}\c\bb+
3(\eta\ro+\dual\eta\si)\right]\nn\\
&&\dddd_4\b+2\tr\chi\b=\divv\a-\left[2\om \b-(2\zeta+\etab)\a\right]\nn\\
&&\dddd_3\a+\frac{1}{2}\tr\chib\a=\nabb\hot\b+\left[4\omb\a-3(\hat{\chi}\ro+
\dual\hat{\chi}\si)+(\zeta+4\eta)\hot\b\right]\nn
\eea
where $\hot$ denotes the ``traceless tensor product, $\dual$  denotes the Hodge dual and, with $X,Y$ vector fields tangent to $S(u,v)$,
\bea
&&\a(R)(X,Y)={R}(X,e_4,Y,e_4)\ ,\ \b(R)(X)=\frac{1}{2}{R}(X,e_4,e_3,e_4)\nn\\
&&\ro(R)=\frac{1}{4}{R}(e_3,e_4,e_3,e_4)\eql{3.1.19za}\ \ \ \ \ ,\ \ \ \ \si(R)
=\frac{1}{4}{^\star R}(e_3,e_4,e_3,e_4)\\
&&\bb(R)(X)=\frac{1}{2}{R}(X,e_3,e_3,e_4)\ ,\ \ \aa(R)(X,Y)=R(X,e_3,Y,e_3)\nn
\eea 
These equations are similar to the ones in the Minkowski spacetime. The main difference is given by the terms in square brackets, absent in the flat case, which are products between the $R$ 
null components and the connection coefficients.

\NI We can now calculate the Bianchi equations for 
\beaa
R_{\mu\nu\ro\si}\equiv {R_{(Kerr)}}_{\mu\nu\ro\si}+{\de R}_{\mu\nu\ro\si}
\eeaa
\NI with ${\de R}_{\mu\nu\ro\si}=O(\ep)$

\NI Denoting with $\{\de\a,\de\b,\de\ro,...\}$   the null components of the tensor $\de R$, $\de\a\equiv\a(\de R)$, and  with the index ${^{(0)}}$ the connection coefficients related to the Kerr spacetime
we obtain at order $\ep$, neglecting higher order terms:
\bea
&&\ML\ML\ML\dddd^{(0)}_4\de\aa+\frac{1}{2}\tr\chi\de\aa=-\nabb^{(0)}\hot\de\bb+\left[4\om^{(0)}\de\aa-3(\hat{\chib}^{(0)}\de\ro-\dual\hat{\chib}^{(0)}\de\si)+(\zeta^{(0)}-4\etab^{(0)})\hot\de\bb\right]+h(R)(e_a,e_3,e_b)\nn\\
&&\ML\ML\ML\dddd^{(0)}_3\de\bb+2\tr\chib\s\de\bb=-\divv^{(0)}\de\aa-\left[2\omb^{(0)}\de\bb+(-2\zeta^{(0)}+ \eta^{(0)})\c\de\aa\right]
+h(R)(e_3,e_3,e_b)\nn\\
&&\ML\ML\ML\dddd^{(0)}_4\de\bb+\tr\chi\de\bb=-\nabb^{(0)}\de\ro+\left[2\om^{(0)}\de\bb+2\hat{\chib}^{(0)}\c\de\b+\dual\nabb^{(0)}\de\si-3(\etab^{(0)}\de\ro-\dual\etab^{(0)}\de\si)\right]+h(R)(e_4,e_3,e_b)\nn\\
&&\ML\ML\ML\dd^{(0)}_3\de\ro+\frac{3}{2}\tr\chib\de\ro=-\divv^{(0)}\de\bb-\left[{2^{-1}}\hat{\chi}^{(0)}\c\de\aa-\zeta^{(0)}\c \de\bb+2\eta^{(0)}\c\de\bb\right]
+h(R)(e_3,e_4,e_3)\nn\\
&&\ML\ML\ML\dd^{(0)}_4\de\ro+\frac{3}{2}\tr\chi\de\ro=\divv^{(0)}\de\b-\left[{2^{-1}}\hat{\chib}^{(0)}\c\de\a-\zeta^{(0)}\c\de\b-2\etab^{(0)}\c\de\b\right]
+h(R)(e_4,e_3,e_4)\eql{nullBian1}\\
&&\ML\ML\ML\dd^{(0)}_3\de\si+\frac{3}{2}\tr\chib\de\si=-\divv^{(0)}\de\dual\bb+\left[{2^{-1}}\hat{\chi}^{(0)}\c\dual\de\aa
-\zeta^{(0)}\c\dual\de\bb-2\eta^{(0)}\c\dual\de\bb\right]+{^\star h(R)}(e_3,e_4,e_3)\nn\\
&&\ML\ML\ML\dd^{(0)}_4\de\si+\frac{3}{2}\tr\chi\de\si=-\divv^{(0)}\de\dual\b
+\left[{2^{-1}}\hat{\chib}^{(0)}\c\dual\de\a-\zeta^{(0)}\c\dual\de\b-2\etab^{(0)}\c\dual\de\b\right]+{^\star h(R)}(e_4,e_3,e_4)\nn\\
&&\ML\ML\ML\dddd^{(0)}_3\de\b+\tr\chib\de\b=\nabb^{(0)}\de\ro+\left[2\omb^{(0)}\de\b+\dual\nabb^{(0)}\de\si+2\hat{\chi}^{(0)}\c\de\bb
+3(\eta^{(0)}\de\ro+\dual\eta^{(0)}\de\si)\right]+h(R)(e_3,e_4,e_a)\nn\\
&&\ML\ML\ML\dddd^{(0)}_4\de\b+2\tr\chi\de\b=\divv^{(0)}\de\a-\left[2\om^{(0)}\de\b-(2\zeta^{(0)}+\etab^{(0)})\de\a\right]+h(R)(e_4,e_a,e_4)\nn\\
&&\ML\ML\ML\dddd^{(0)}_3\de\a+\frac{1}{2}\tr\chib\de\a=\nabb^{(0)}\hot\de\b+\left[4\omb^{(0)}\de\a-3(\hat{\chi}^{(0)}\de\ro
+\dual\hat{\chi}^{(0)}\de\si)+(\zeta^{(0)}+4\eta^{(0)})\hot\de\b\right]+h(R)(e_a,e_4,e_b)\nn
\eea

\NI Clearly we can think to these equations as the Bianchi equations for the Kerr spacetime where a source term $h$ has been added. If we consider $h$ depending on  $R^0=\{\a^0, \b^0...\}$ the null Riemann components in Kerr and on $\de O=\{\de\om,\de\omb,\de\hat{\chi},...\}$, the connection coefficients related to $\de R$, then $h(R^0,\de O)$ assume the form:
\bea\label{hdef1}
h(R)(e_a,e_3,e_b)\!&=&\!\!\left[4\de\om\aa^{(0)}-3(\de\hat{\chib}\ro^{(0)}-\dual\de\hat{\chib}\si^{(0)})+(\de\zeta-4\de\etab)\hot\bb^{(0)}\right]\nn\\
h(R)(e_3,e_3,e_b)\!&=&\!-\left[2\de\omb\bb^{(0)}+(-2\de\zeta+ \de\eta)\c\aa^{(0)}\right]\nn\\
h(R)(e_4,e_3,e_b)\!&=&\!\left[2\de\om\bb^{(0)}+2\de\hat{\chib}\c\b^{(0)}+(\dual\nabb-\dual\nabb^{(0)})\si^{(0)}-3(\de\etab\ro^{(0)}-\de\dual\etab\si^{(0)})\right]\nn\\
h(R)(e_3,e_4,e_3)\!&=&\!-\left[{2^{-1}}\de\hat{\chi}\c\aa^{(0)}-\de\zeta\c\bb^{(0)}+2\de\eta\c\bb^{(0)}\right]\nn\\
h(R)(e_4,e_3,e_4)\!&=&\!-\left[{2^{-1}}\de\hat{\chib}\c\a^{(0)}-\de\zeta\c\b^{(0)}-2\de\etab\c\b^{(0)}\right]\nn\\
{^\star h(R)}(e_3,e_4,e_3)\!&=&\!\left[{2^{-1}}\de\hat{\chi}\c\dual\aa^{(0)}-\de\zeta\c\dual\bb^{(0)}-2\de\eta\c\dual\bb^{(0)}\right]\eql{Hdef1}\\
{^\star h(R)}(e_4,e_3,e_4)\!&=&\!\left[{2^{-1}}\de\hat{\chib}\c\dual\a^{(0)}-\de\zeta\c\dual\b^{(0)}-2\de\etab\c\dual\b^{(0)}\right]\nn\\
h(R)(e_3,e_4,e_a)\!&=&\!\left[2\de\omb\b^{(0)}+(\dual\nabb-\dual\nabb^{(0)})\si^{(0)}+2\de\hat{\chi}\c\bb^{(0)}+3(\de\eta\ro^{(0)}+\de\dual\eta\si^{(0)})\right]\nn\\
h(R)(e_4,e_a,e_4)\!&=&\!-\left[2\de\om\b^{(0)}-(2\de\zeta+\de\etab)\a^{(0)}\right]\nn\\
h(R)(e_a,e_4,e_b)\!&=&\!\left[4\de\omb\a^{(0)}-3(\de\hat{\chi}\ro^{(0)}+\de\dual\hat{\chi}\si^{(0)})+(\de\zeta+4\de\eta)\hot\b^{(0)}\right]\nn
\eea
\smallskip

\NI Where, as said before the connection coefficients related to $\de R$ which are of order $O(\ep)$.

\NI Equations \ref{hdef1} show, as anticipated before, that there is not a linear equation in $\de R$ which can be thought as the natural linearization of equations. \ref{nullBian}. This is clear observing that $\{\de\om,\de\hat{\chib},\de\zeta,..\}$ depend on $\de R$ in a very indirect way, namely through the solutions of equations \ref{eq:transportM}, \ref{eq:HodgeH} 

\NI Therefore  we have not a natural choice for the tensor $h$ in \ref{Bianchi4}.

\medskip

\NI At this point we have two possibilities:

\medskip

\NI The first one is to consider $h_{\nu\ro\si}$ as a defined tensor field depending on $W$ assuming for it decays  we expect to be compatible with the expected decay for the connection coefficients related to a perturbation of Kerr spacetime. This implies we have a good candidate for  $\{ \de\hat{\chib}, \de\omega,\de\zeta,..\}$ and moreover that we are able to exploit the \cite{Kl-Ni1} approach in this linear case, adding the source term $h$.

\NI The second possibility is to impose $h_{\nu\ro\si}$ in such a way the norms $\cal{Q}$ we have to exploit to obtain the peeling decay mimyking the \cite{Kl-Ni1} method in the linear case, are bounded.
\smallskip

\NI Both approaches require a preliminary discussion on how these norms have been bounded in the \cite{Kl-Ni1} approach and how to adapt this result to our linear case in Kerr background. 





\smallskip

\subsection{The boundedness of the energy norms in the [Kl-Ni] results.}
The main ingredient to prove the global existence in the external region used in \cite{Ch-Kl:book} and \cite{Kl-Ni1} is the control of suitable energy type norms made with the Bel-Robinson tensor $W$. Once the boundedness of the norms for $W$ has been proved, from them it is possible to obtain the decay along the null directions of the various null components of $W$. As remarked before in \cite{Kl-Ni1} the decay of some of the Riemann components was slower than the one suggested by the ``Peeling Theorem". As the decay of the various null components is connected to the weights of the energy type norms, in \cite{Kl-Ni2} a better result was obtained showing that one can define a set of energy type norms for $\lie_{T_0}W$ with ``stronger" weights and from them with some work to obtain a decay in agreement with peeling. In the next subsection we define  the Bel-Robinson norms used in  \cite{Kl-Ni1} and the modified version for the Peeling in \cite{Kl-Ni1}.

\subsubsection{Notations, definitions and results}

\NI As already said, we can define an external region endowed with a double null canonical foliation for the details see \cite{Kl-Ni1}, Chapter 3) made by null hypersurfaces similar to the null outgoing and incoming cones of the Minlowski spacetime. They can be expressed at least locally as the level hypersurfaces of the solutions of the eikonal equation,
\[g^{\mu\nu}\pr_{\mu}w\pr_{\nu}w=0\ ,\]
 $u=u(p)$, $ \ub=\ub(p)$,  with initial data given on the external region of a spacelike hypersurface $\Si_0$ \footnote{for the sake of simplicity we do not discuss here the problem of the canonical foliation of the  ``final" incoming cones , see for example \cite{Ch-Kl:book} for details} .  In the non linear case the determination of the foliation is part of the problem one is solving,\footnote{and its existence is proved from the control of the connection coefficients obtained solving equations \ref{eq:transportM} and \ref{eq:HodgeH}.} in the linearized version we are considering in this paper the double null foliation is an appropriate foliation  of the Kerr spacetime, we will give explycitely later on. Once we have the null coordinates $u$ and $\ub$ we complete them with angular coordinates $\theta$, $\phi$ adapted to $S(u,\ub)$, see [Kl-Ni], cap 3. Then we introduce a null frame $\{e_{\mu}\}=\{ e_4, e_3, e_2,e_1\}$ with $e_4$ null tangent to the outgoing null cones and $e_3$ null tangent to the incoming null one, $g(e_3,e_4)=-2$, $e_{(1,2)}$ tangent to $S(u,\underline{u})$ and such that $g(e_{(1,2)},e_{(3,4)})=0$. 
 
 \NI Then we consider the connection coefficients $\cal{O}=\{\chi,
 \chib, \omega, \omb, \zeta\}$ and the components of $W$,and  $\cal{R}=\{\alpha, \underline{\alpha}, \beta, \underline{\beta}, \ro, \si\}$ the conformal part of the Riemann tensor, with respect to $\{e_{\mu}\}$, we call them null  Riemann components see [Kl-Ni], cap 3 for the definition.


\subsubsection{Energy type norms}
Let $W$ be the conformal part of the Riemann tensor of the vacuum spacetime or an arbitrary Weyl field, the Bel-Robinson tensor associated to it is the tensor field
\bea
Q_{\a\b\ga\de}[W]=
W_{\a\ro\ga\si}{{{W_\b}^\ro}_\de}^\si+{\rdual{W}_{\a\ro\ga\si}}\rdual{W}{{{_\b}^\ro}_\de}^\si
\eea
The  \cite{Kl-Ni1} energy type norms, denoted by $\cal Q, \underline{\cal Q}$ have the following expression: 
\bea\label{norme1}
&& {\cal Q}(u,\ub) = {\cal Q}_1(u,\ub)+{\cal Q}_2(u,\ub)\nn\\
&& \underline{\cal Q}(u,\ub) = \underline{\cal Q}_1(u,\ub)+\underline{\cal Q}_2(u,\ub)\	,
\eea
where, denoting
$V(u,\ub)=J^-(S(u,\ub))$
\bea\label{norms1}
\ML{\cal Q}_1(u,\ub) &\equiv & \int_{C(u)\cap V(u,\ub)}Q(\lie_T{W})(\bar{K},\bar{K},\bar{K},e_4)\nn\\
&+&\int_{C(u)\cap V(u,\ub)}Q(\lie_O W)(\bar{K},\bar{K},T,e_4)\nn\\
\ML{\cal Q}_2(u,\ub) &\equiv &\int_{C(u)\cap V(u,\ub)}Q(\lie_O\lie_T W)(\bar{K},\bar{K},\bar{K},e_4)\nn\\
&+&\int_{C(u)\cap V(u,\ub)}Q(\lie^2_O W)(\bar{K},\bar{K},T,e_4)\label{Q_2}\nn\\
&+&\int_{C(u)\cap V(u,\ub)}Q(\lie_S\lie_T W)(\bar{K},\bar{K},\bar{K},e_4)\nn
\eea
\bea\label{norms2}
\ML\underline{\cal Q}_1(u,\ub)&\equiv&
\int_{\un{C}(\ub)\cap V(u,\ub)}Q(\lie_T W)(\bar{K},\bar{K},\bar{K},e_3)\nn\\
&+&\int_{\un{C}(\ub)\cap V(u,\ub)}Q(\lie_O W)(\bar{K},\bar{K},T,e_3).\nonumber
\eea
\bea
\underline{\cal Q}_2(u,\ub) &\equiv&\int_{\un{C}(\ub)\cap V(u,\ub)}
Q(\lie_O\lie_T W)(\bar{K},\bar{K},\bar{K},e_3)\nn\\
&+&\int_{\un{C}(u)\cap V(u,\ub)}Q(\lie^2_O W)(\bar{K},\bar{K},T,e_3)\nn\\
&+&\int_{\un{C}(\ub)\cap V(u,\ub)}Q(\lie_S\lie_T W)(\bar{K},\bar{K},\bar{K},e_3)\label{Qb_2}
\eea
The norms associated to the initial data hypersurface are:
\bea
\ML\underline{\cal Q}_{1_{\Si_0\cap V(u,\ub)}}\!\!&\equiv&\!\int_{\Si_0\cap V(u,\ub)}Q(\lie_T W)(\bar{K},\bar{K},\bar{K},T)
+\int_{\Si_0\cap V(u,\ub)}Q(\lie_O W)(\bar{K},\bar{K},T,T)\nn\\
\ML\underline{\cal Q}_{2_{\Si_0\cap V(u,\ub)}}\!\!&\equiv&\!\int_{\Si_0\cap V(u,\ub)}Q(\lie_O\lie_T W)(\bar{K},\bar{K},\bar{K},T)
+\int_{\Si_0\cap V(u,\ub)}Q(\lie_O^2 W)(\bar{K},\bar{K},T,T)\nonumber\\
\!&+&\!\int_{\Si_0\cap V(u,\ub)}Q(\lie_S\lie_T W)(\bar{K},\bar{K},\bar{K},T).
\eea
With $\ ^{(j)}O$ the rotation vectorfields, \footnote{The precise definition of the $O$ vectorfields for a perturbed spacetime is given in [Ch-Kl], chapter 16}
\begin{eqnarray}\label{K_0}
& & T_0= \frac{\p}{\p_t}\nn\\
& &  S=\frac 1 2 (ue_e+\ub e_4)\\
& & K_0 = \frac 1 2 (u^2 e_3+\ub^2 e_4)\nn\\
& &\bar{K}=\frac 1 2 (\ttau^2e_3+\tttau^2e_4)\nn\\
& & T=\frac 1 2  (e_3+e_4)\nn\	.
\end{eqnarray}
and
$\lie$ the modified Lie derivative, see \cite{Kl-Ni1}, defined as \footnote{$\lie$ is such that if $W$ is a Weyl field $\lie_X W$ is also a Weyl field}
\bea
\hat{\mathcal{L}}_X W =\Lie_X W-\frac 1 2 {^{(X)}}[W]+\frac 3 8(\tr{^{(X)}}\pi)W,\eql{Liemod}
\eea
where
\bea
{^{(X)}[W]}_{\a\b\ga\de}={^{(X)}\pi}^{\la}_{\a}W_{\la\b\ga\de}+
{^{(X)}\pi}^{\la}_{\b}W_{\a\la\ga\de}+{^{(X)}\pi}^{\la}_{\ga}W_{\a\b\la\de}
+{^{(X)}\pi}^{\la}_{\de}W_{\a\b\ga\la}\nn
\eea
Where with $\Lie_OW$ we mean the Lie derivative with respect to the rotation group generator, see \cite{Kl-Ni1} for details.
\smallskip

\NI Moreover given $\cal{K}$ an open region of the Kerr spacetime we define the following quantities
\begin{equation}
{\cal Q}_{\cal{K}}\equiv
\sup_{\{u,\ub|S(u,\ub)\subseteq\cal{K}\}}\{{\cal Q}(u,\ub)+\underline{\cal Q}(u,\ub)\}
\end{equation}
and, on the initial spacelike hypersurface $\Si_0$ 
\begin{equation}
{\cal Q}_{\Si_0\cap\cal{K}}=\sup_{\{u,\ub|S(u,\ub)\subseteq\cal{K}\}}
\{{\cal Q}_{1_{\Si_0\cap V(u,\ub)}}+\underline{\cal Q}_{2_{\Si_0\cap V(u,\ub)}}\}.
\end{equation}
In \cite{Kl-Ni1} $\cal{K}$ is a finite spacetime region where, with a bootstrap mechanism, one proves that these norms have good estimates which allow to ``extend" it globally. In the linear case we are considering now,  $\cal{K}$ will be from the beginning an unbounded region describing the so called ``external region" of the Kerr spacetime, where we want to prove that the null components of the vector field $W$ decay in agreement with the peeling. 

\NI In the non linear case one of the main steps to obtain global existence is to prove that ${\cal Q}_{\cal{K}}$ can be bounded by
${\cal Q}_{\Si_0\cap\cal{K}}$ if the initial data are small\footnote{see [Kl-Ni], cap 2 for smallness condition}. In the linearized case an analogous estimate is the main technical result. Nevertheless, as already said, the norms \ref{norms1}, \ref{norms2}, which were used in \cite{Kl-Ni1} do not provide the correct asymptotic behaviour. This was cured in \cite{Kl-Ni2} defining a different set of $\cal Q$ norms with weights modified by a factor $|u|^{\ga}$ with $\ga>0$ appropriately chosen, see \cite{Kl-Ni2}, equations (2.6), (2.7), (2.8), (2.9), section 2.2. We recall the first few of them,
\bea
\widetilde{Q}_1(u,\ub)&\equiv&\int_{C(u)\cap
V(u,\ub)}|u|^{2\ga}Q(\lie_{T}R)(\acc,\acc,\acc,e_4)\nn\\
&&+\int_{C(\la)\cap V(u,\ub)}|\la|^{2\ga}Q(\lie_{O}R)(\acc,\acc,T,e_4)\nn\\
\widetilde{Q}_2(u,\ub)&\equiv&\int_{C(u)\cap V(u,\ub)}|u|^{2\ga}Q(\lie_{O}\lie_{T}\rr)(\acc,\acc,\acc,e_4)\nn\\
&&+\int_{C(u)\cap V(u,\ub)}|u|^{2\ga}Q(\lie^2_{O}R)(\acc,\acc,T,e_4)\eql{1.1xx}\\
&&+\int_{C(u)\cap V(u,\ub)}|u|^{2\ga}Q(\lie_{S}\lie_{T}R)(\acc,\acc,\acc,e_4)\nn
\eql{QQ12}
\eea
\bea
\widetilde{\underline Q}_1(u,\ub)&\equiv&\sup_{V(u,\ub)\cap{\Si}_0}|r^3{\overline\ro}|^2+
\int_{\Cb(\ub)\cap V(u,\ub)}|u|^{2\ga}Q(\lie_{T}R)(\acc,\acc,\acc,e_3)\nn\\
&&+\int_{\Cb(\ub)\cap V(u,\ub)}|u|^{2\ga}Q(\lie_{O}R)(\acc,\acc,T,e_3)\nn\\
\widetilde{\underline Q}_2(u,\ub)&\equiv&\int_{\Cb(\ub)\cap
V(u,\ub)}|u|^{2\ga}Q(\lie_{O}\lie_{T}R)(\acc,\acc,\acc,e_3)\nn\\
&&+\int_{\Cb(\ub)\cap V(\la,\nu)}|u|^{2\ga}Q(\lie^2_{O}R)(\acc,\acc,T,e_3)\eql{QQb12}\\
&&+\int_{\Cb(\ub)\cap V(u,\ub)}|u|^{2\ga}Q(\lie_{S}\lie_{T}R)(\acc,\acc,\acc,e_3)\ .\nn
\eea
 Clearly we can define the corresponding  norms \[\ML\underline{\cal Q}_{1,2_{\Si_0\cap V(u,\ub)}}\!\! \ ,\ \  \underline{\cal Q}_{\cal{K}}\mbox{ and } \underline{\cal Q}_{\Si_0\cap V}\]

\NI In \cite{Kl-Ni2} it was proved that also these norms can be globally bounded in terms of initial data with suitable decays, and that their extra weight $|u|^{\ga}$ improved the asymptotic behaviour adding to the various components of $W$ a decay factor in the $u$ variable. More precisely
This was, nevertheless, only an intermediate step as, moving on a null hypersurface toward the null infinity, the $u$ variable is constant so that $\a$ and $\b$ still do not have the expected decay.
At this point we find the following decays:
\bea\label{intes}
&&\lim_{u\rightarrow\infty}r^{\frac{7}{2}}|u|^{(\frac{5}{2}+\ep')}\a=C_0\nn\\
&&\lim_{u\rightarrow\infty}r^{\frac{7}{2}}|u|^{(\frac{5}{2}+\ep')}\b=C_0\nn\\
&&\lim_{u\rightarrow\infty}r^{3}|u|^{(3+\ep')}|\ro-\overline{\ro}|=C_0\nn\\
&&\lim_{u\rightarrow\infty}r^3|u|^{(3+\ep')}|\sigma-\overline{\sigma}|=C_0\eql{peel1}\\
&&\lim_{u\rightarrow\infty}r^{2}|u|^{(4+\ep')}|\bb|=C_0\nn\\
&&\lim_{u\rightarrow\infty}r |u|^{(5+\ep')}|\aa|\leq C_0\ .\nn
\eea
\NI The second technical point which allowed to obtain the result was to look at those  Bianchi equations which can be expressed as transport equations along the incoming cones $\Cb(\ub)$. The extra decay factor in the $u$ variable allowed to get better decays in $r$ for the various components of the Weyl field $W$ in terms of their values on $\Cb(\ub)\cap\Si_0$.This implied a decay in this variable which was in agreement with the ``Peeling theorem" provided  that the decay on $\Si_0$ was sufficiently fast.\footnote{We do not give more details here as the same argument will be used and discussed in detail during the proof of the present result.}
Let  us recall a shortened version of the main result in \cite{Kl-Ni2}:
\begin{thm}\label{TT1.2}
Let assume that on $\Si_0/B$, contained in the external region, the metric and the second fundamental form have the following asymptotic behaviour
\footnote{Here $f=O_q(r^{-a})$ means that $f$ asymptotically
behaves as $O(r^{-a})$ and its partial derivatives $\partial^kf$, up to order $q$ behave as $O(r^{-a-k})$. Here with $g_{ij}$ we mean the components written in Cartesian coordinates. }
\bea
&&g_{ij}={g_S}_{ij}+O_{q+1}(r^{-(\frac{3}{2}+\gamma)})\nn\\
&&{k}_{ij}=O_{q}(r^{-(\frac{5}{2}+\ga)})\eql{1.1b}
\eea
where ${g_S}$ denotes the restriction of the
Schwarzschild metric on the initial hypersurface:
\[g_S=(1-\frac{2M}{r})^{-1}dr^2+r^2(d\theta^2+\sin\theta^2d\phi^2)\ .\] 
and $\ga=\frac{3}{2}+\ep$ and $\ep>0$, then along the outgoing null hypersurfaces $C(u)$ the following limits hold, with $\ep'<\ep$:
\bea
&&\lim_{C(u);v\rightarrow\infty}r(1+|u|)^{(4+\ep')}\aa=C_0\nn\\
&&\lim_{C(u);v\rightarrow\infty}r^2(1+|u|)^{(3+\ep')}\bb=C_0\nn\\
&&\lim_{C(u);v\rightarrow\infty}r^3\ro=C_0\nn\\
&&\lim_{C(u);v\rightarrow\infty}r^3\si=C_0\eql{peel1}\\
&&\lim_{C(u);v\rightarrow\infty}r^4(1+|u|)^{(1+\ep')}\b=C_0\nn\\
&&\sup_{(u,v)\in \M}r^{5}(1+|u|)^{\ep'}|\a|\leq C_0\ .\nn
\eea
\end{thm}

\smallskip

\NI Looking at the assumptions of the theorem a problem appears immediately  in trying to extend it to spacetimes near to Kerr. In fact the Kerr metric in the Boyer-Linquist coordinates $\{t,r,\theta,\phi\}$ is:
\begin{eqnarray}\label{BL}
ds^2 =-\frac{\De-a^2\sin^2\theta}{\Si}dt^2+\frac{\Sigma}{\Delta}dr^2+\Sigma
d\theta^2-\frac{4Mar\sin^2\theta}{\Sigma}d\phi dt+R^2\sin^2\theta d\phi^2 \nn
\end{eqnarray}
and its restriction to $\Si_0$:
\bea
ds^2\!&=&\!\frac{\Sigma}{\Delta}dr^2+\Si d\theta^2+R^2\sin^2\theta d\phi^2\\
\!&=&\!g_S+\left(\frac{a^2\sin^2\theta}{r^2}+O\left(\frac{a^2m}{r^3}\right)+O\left(\frac{a^4}{r^4}\right)\right)\!dr^2+\left(\frac{a^2\cos^2\theta}{r^2}\right)\!r^2d\theta^2\nn\\
\!&+&\!\left(\frac{a^2}{r^2}+O\left(\frac{a^2m}{r^3}\right)+O\left(\frac{a^4}{r^4}\right)\right)\!r^2\sin^2\theta d\phi^2\ .\nn
\eea
It follows that the components of the correction to the $g_S$ metric have terms  of order $O({a^2}/{r^2})$ which do not satisfy the assumptions of Theorem \ref{TT1.2}.


\NI The idea to overcome this difficulty is based like in [Kl-Ni1], on the fact that the Kerr spacetime is static and, therefore, denoting with $W$ the Riemann tensor for a vacuum spacetime near Kerr, $W=W_{(Kerr)}+\de W$, if we consider $\Lie_TW$, basically we subtract the Kerr part and $\Lie_TW=\Lie_T\de W$. This suggests that we can try to obtain the correct asymptotic behaviour for $\Lie_TW$ whose initial data can be chosen to decay arbitrarily fast. Once we have a control of the asymptotic behaviour of $\Lie_TW$ we can recover the one of $W$ by an integration along the null outgoing directions.
\smallskip

\NI This has been done also in [Ca-Ni] for Kerr spacetime, where we obtained the right decays condition on the initial data in terms of the metric and its first outgoing derivatives on the external region, let us state the final result of that work,see [Ca-Ni], section 4:
 
\begin{thm}\label{final version}
Assume that initial data  are given on $\Si_0$ such that, outside of a ball centered in the origin of radius $R_0$, they are different from the ``Kerr initial data of a Kerr spacetime with mass $M$ satisfying 
\[\frac{M}{R_0}<<1\ \ ,\ \ J\leq M^2 \mbox{\ \ \ \ \ (external region)}\]
for some metric corrections decaying faster than $r^{-3}$ toward spacelike infinity together with its derivatives up to an order $q\geq 4$, namely
\footnote{The components of the metric tensor written in dimensional coordinates.}
\bea
g_{ij}=g^{(Kerr)}_{ij}+o_{q+1}(r^{-(3+\frac{\ga}{2})})\ \ ,\ \ k_{ij}=k^{(Kerr)}_{ij}+o_{q}(r^{-(4+\frac{\ga}{2})})
\eea
where $\ga>0$. Let us assume that the metric correction $\de g_{ij}$, the second fundamental form correction $\de k_{ij}$ are sufficiently small, namely the function ${\cal J}$ equation \ref{3.421aa} made by $L^2$ norms on $\Si_0$ of these quantities is small, \footnote{This will also imply a slightly stronger condition on the decay of the metric and second fundamental form components, basically that $\int_{R_0}^{\infty}dr r^{5+\ga}|\de g_{ij}|^2<\infty$,  $\int_{R_0}^{\infty}dr r^{7+\ga}|\de k_{ij}|^2<\infty$.}
\bea
{\cal J}(\Si_0,R_0; \de{^{(3)}\!}{\bf g}, \de{^{(3)}\!}{\bf k})\leq \varepsilon\ ,
\eea
then this initial data set has a unique development, ${\widetilde{\M}}$, defined outside the domain of influence of $B_{R_0}$ with the following properties:

\NI
i)\,\,\, ${\widetilde{\M}}={\widetilde{\M}}^{+}\cup{\widetilde{\M}}^{-}$ where ${\widetilde{\M}}^{+}$ consists of the part of
${\widetilde{\M}}$ which is in the future of $\Si/B_{R_0}$, ${\widetilde{\M}}^{-}$ the one to the past. 

\medn
\NI
ii) \,\,\,$({\widetilde{\M}}^{+},g)$ can be foliated by a canonical double null foliation
$\{C(u),\Cb(\ub)\}$  whose outgoing leaves $C(u)$ are complete
\footnote{By this we mean that the null geodesics generating $C(u)$ can be indefinitely
extended toward the future.}
for all $|\la|\geq |u_0|=R_0$. The boundary of $B_{R_0}$ can be chosen to be the intersection of $C(u_0)$ with $\Si_0$.
\medn

\NI
iii) \,\,\, The various null components of the Riemann tensor relative to the null frame associated to the double null canonical foliation, decay along the outgoing ``cones" in agreement with the ``Peeling Theorem".
\end{thm}

\begin{remark}

\NI The condition on the metric $g$ and on the outgoing first derivative $k$ assure us that every null component of $\lie_{T_0} W$ on $\Sigma_0$ 
 \footnote{$T_0=\frac {\p}{\p t}$ and $\lie_{T_0}$ is the modified Lie derivative introduced in \cite{Ch-Kl:book}).} decay like $r^{-(6+\e)}$. More precisely, the $\lie_{T_0} W$ null components satisfy the following decays on $\Si_0$,  with $\ep'<\ep$:
\begin{eqnarray}
&&\sup_{\mathcal{K}}r^{5}|u|^{(1+\e')}|\a(\lie_{T_0}W)| \leq C_0,\quad\sup_{\mathcal{K}}r^{4}|u|^{2+\e'} |\b(\lie_{T_0}W)|\leq C_0\nn\\
&&\sup_{\mathcal{K}}r^3|u|^{3+\e}|\ro(\lie_{T_0}W)| \leq C_0,\quad \sup r^3|u|^{3+\e'}|\si(\lie_{T_0}W)|\leq C_0\ \ \eql{LTWdec}\\
&& \sup_{\mathcal{K}}r^2|u|^{4+\e'}\bb(\lie_{T_0}W)\leq C_0,\quad\sup_{\mathcal{K}}r|u|^{5+\e' }|\aa(\lie_{T_0}W)|\leq C_0\nn\ \ .
\end{eqnarray}

\smallskip

\NI In the linear case this result can be obtained in a easier way without introducing the $\lie_{T_0}$ derivative of the Weyl tensor as in this case the subtraction of a ``$W_{(Kerr)}$" part with a slow asymptotic decay can be done without any problem. The strategy we use is, nevertheless, appropriate to treat in a very similar way  also the non linear perturbations of the  Kerr spacetime and  could be also adapted to study the analogous of equation \ref{Bianchi4} when its right hand side is a non linear term satisfying appropriate conditions.
\end{remark}

\NI Now that we have defined the $\cal{Q}$ norms and their modified version, we can discuss how they can by bounded in the nonlinear case and how to obtain the same result  for equations 1.8 imposing for $h_{\mu\nu\ro}$ the right decays.

\subsection{The estimates of the $\cal{Q}$ norms in [Kl-Ni]}

\NI The core of the proof of the boundedness of the $\cal{Q}_{\cal{K}}$ norms is the estimate of the error term $\mathcal{E}=\mathcal{E}_1+\mathcal{E}_2\leq c\ep_0\cal{Q}_{\cal{K}}$, with $c$ constant and $\ep_0$ sufficiently small and with $\varep1$ , $\varep2$ defined as, see [kl-Ni], eq, 6.0.6:

\bea
&&\zeta_1(u,\ub)=\int_{V(u,\ub)}Div Q(\lie_T W)_{\b,\ga,\de}(\bar{K}^{\b}\bar{K}^{\ga}\bar{K}^{\de})\ \nn\\
&+&\int_{V(u,\ub)}Div Q(\lie_O W)_{\b,\ga,\de}(\bar{K}^{\b}\bar{K}^{\ga}T^{\de})\ \nn\\
&+&\frac 3 2 \int_{V(u,\ub)} Q(\lie_T W)_{\a,\b,\ga,\de}(\ ^{(\bar{K})}\pi^{\a,\b}{\b}\bar{K}^{\ga}\bar{K}^{\de})\ \nn\\
&+& \int_{V(u,\ub)} Q(\lie_O W)_{\a,\b,\ga,\de}(\ ^{(\bar{K}})\pi^{\a,\b}{\b}\bar{K}^{\ga}T^{\de})\ \nn\\
&+&\frac 1 2 \int_{V(u,\ub)} Q(\lie_O W)_{\a,\b,\ga,\de}(\ ^{(T)}\pi^{\a,\b}{\b}\bar{K}^{\ga}\bar{K}^{\de})\ \nn
\eea
\bea
&&\zeta_2(u,\ub)=\int_{V(u,\ub)}Div Q(\lie^2_O W)_{\b,\ga,\de}(\bar{K}^{\b}\bar{K}^{\ga}T^{\de})\ \nn\\
&+&\int_{V(u,\ub)}Div Q(\lie_O  \lie_T W)_{\b,\ga,\de}(\bar{K}^{\b}\bar{K}^{\ga}\bar{K}^{\de})\ \nn\\
&+&\int_{V(u,\ub)}Div Q(\lie_S \lie_T W)_{\b,\ga,\de}(\bar{K}^{\b}\bar{K}^{\ga}\bar{K}^{\de})\ \nn\\
&+& \int_{V(u,\ub)} Q(\lie^2_O W)_{\a,\b,\ga,\de}(\ ^{(\bar{K})}\pi^{\a,\b}{\b}\bar{K}^{\ga}T^{\de})\ \nn\\
&+&\frac 1 2 \int_{V(u,\ub)} Q(\lie^2_O W)_{\a,\b,\ga,\de}(\ ^{(T)}\pi^{\a,\b}{\b}\bar{K}^{\ga}\bar{K}^{\de})\ \nn\\
&+&\frac 3 2 \int_{V(u,\ub)} Q(\lie_O \lie_T W)_{\a,\b,\ga,\de}(\ ^{(\bar{K})}\pi^{\a,\b}{\b}\bar{K}^{\ga}\bar{K}^{\de})\ \nn\\
&+&\frac 3 2 \int_{V(u,\ub)} Q(\lie_S \lie_T W)_{\a,\b,\ga,\de}(\ ^{(\bar{K})}\pi^{\a,\b}{\b}\bar{K}^{\ga}\bar{K}^{\de})\ \nn
\eea
Let us restrict our attemption tu $\mathcal{E}_1$. We can divide it in two parts, the first one made of the two terms containing $Div\cal{Q}$ and the second one containing only $\cal{Q}$.
The second part is conpceptually easier to treat, in fact using the bootstrap assumptions on the connection coefficients we can estimate the deformation tensor $\ ^{(X)}\pi^{\a,\b}$, for $X= K,T$ and consequently,
using the fact that the $\cal{Q}$ norms can be written in terms of the null Riemann coefficients, see \cite{Kl-Ni1}, section 3.5.1, all the terms with a sum of integrals of 
the form:

\bea
\int_{V(u,\ub)}\ttau^{\ga}\cal{O}\cal{R}^2\nn
\eea

\NI With $\ga$ an integer to be specified, $\cal{O}$  a combination of connection coefficients depending on $\ ^{(X)}\pi$ and $\cal{R}$ a null component of the conformal part of the Riemann tensor depending on the particular term we are considering, see [Kl-Ni] cap 6. By the Schwartz inequality and the Poincare inequalities, see [Kl-Ni], cap 5, we obtain the right estimate.

\NI To estimate the first part we express $Div Q(\lie_X W)_{\b,\ga,\de}(\bar{K}^{\b}K^{\ga}T^{\de})$  as:
\bea
&& Div Q (\lie_X W)_{\b,\ga,\de}(\bar{K}^{\b}\bar{K}^{\ga}X^{\de})=\lie_X Div Q(W)_{\b,\ga,\de}(\bar{K}^{\b}\bar{K}^{\ga}X^{\de})+\nn\\
&& [\lie_X,Div] Q(W)_{\b,\ga,\de}(\bar{K}^{\b}\bar{K}^{\ga}X^{\de})
\eea
We pose 

\bea
&& \lie_X Div Q(W)_{\b,\ga,\de}(\bar{K}^{\b}\bar{K}^{\ga}X^{\de})=J^0\ \  and\ \  [\lie_X,Div] Q(W)_{\b,\ga,\de}(\bar{K}^{\b}\bar{K}^{\ga}X^{\de})=J^i\nn\\
&&  i=1,2,3\nn
\eea

\NI Clearly in the case of $\mathcal{E}_1$ we have $J^0=0$ due to the Bianchi equations for $W$.

\NI The explicit expression of $J^i$ is given in [Kl-Ni] eq. 6.1.6. We write them in symbolical way:

\bea
&& J^1=\ ^{(x)}\pi \cdot D_{\mu}W^{\mu}_{\nu\si\ro}\nn\\
&&J^2=p\cdot W\nn\\
&&J^3=q\cdot W \nn
\eea

\NI Where  $p$ depends on $\ ^{(x)}\pi$ and $q$ depends on $D \ ^{(x)}\pi $. With the same method we can estimate the second part  we can prove the boundedness for these terms.
For what concern  the error term $\mathcal{E}_2$ the situation is analogous with the only difference that for the terms like $ \int_{V(u,\ub)}Div Q(\lie_S \lie_T W)_{\a,\b,\ga,\de}(\bar{K}^{\b}\bar{K}^{\ga}\bar{K}^{\de})$ the $J^0$ term is not zero due to the fact that $T$ is not exactly Killing. Nevertheless by the bootstrap assumption we can write $J^0=\int_{V(u,\ub)}\ttau^{\ga}\cal{O}\dot\cal{R}^2$ with the $\cal{O} $ terms
decaying sufficiently fast to obtain the right estimates and treat the other $J$ terms in the same way of $\mathcal{E}_1$. 

\subsection{The \cite{Kl-Ni1} approach in the linear case}

\NI Trying to  transport this technique to the linear case we have some differences.

\medskip

\NI First we do not have to make any bootstrap assumption on the connection coefficients but we have to calculate them in Kerr.

\medskip

\NI Second, as remarked before we have to choose the right null frame associated to the null cones in Kerr spacetime, namely the one introduced in \cite{Is-Pr} and calculate all the connection coefficients and the null Riemann components related to this frame. 

\NI Let us anticipate the result of that section and write the decays we have obtained for the null Riemann components .

\NI Denoted as $\{e_{\mu}\}$ the null frame associated to the null cones of \cite{Is-Pr} in  Kerr spacetime, the following decays for the null Riemann components hold:

\bea \label{Riemcomp}
&&\a_{(Kerr)},\underline{\a}_{(Kerr)}\simeq\frac {M^3}{ r^5}\nn\\
&&\b_{Kerr},\underline{\b}_{(Kerr)}\simeq\frac {M^2}{ r^4}\\
&&\ro_{(Kerr)}\simeq\frac {M }{r^3},  \si_{(Kerr)}\simeq\frac {M }{r^4}\nn
\eea

\medskip

\NI The third, crucial, point is that we have to choose the right conditions for the hynomogeneous term $h$. This correspond, considering equation 1.14 and \ref{Riemcomp},  to assign the decays for the perturbed connection coefficients $\de O$. 

\NI There is an arbytrariety in this choice as the only requirement is that they satisfy the peeling assumptions.

\NI Nevertheless, as we want $h$ to be compatible with perturbed Kerr spacetime, a good choice can be the one obtained for the connection coefficients far small perturbation of Kerr spacetime in the "very external region", see [Ni]. They  are:\footnote{We can obtain a better decay for $\de \hat{\chi}$, namely $\de \hat{\chi}=o(r^{-3}u^{-1})$}
\begin {prop}
We assume the following decays for the connection coefficients associated to  $\de R$:  
\bea\label{q}
&&\de \hat{\chi}\simeq r^{-2}u^{-2}, \de \hat{\chib}\simeq r^{-1}u^{-3}\nn\\
&&\de \hat{\eta}\simeq r^{-2}u^{-2}, \de \hat{\etab}\simeq r^{-2}u^{-2}\\
&&\de \hat{\om}\simeq r^{-2}u^{-2}, \de \hat{\omb}\simeq r^{-1}u^{-3}\nn
\eea
\end{prop}

\NI Once we have assigned these decays, to construct a good candidate for we make the following ansatz:  

\NI Let us suppose $h_{\nu\ro\si}=D^{\mu} W_{\mu\nu\ro\si}$ can be written as $\tilde{p}^{\mu} W_{\mu\nu\ro\si}$ analogously to what happen for the $J^2$ current. Then by equation 1.14 we obtain:

\bea\label{p}
\tilde{p}^{\mu} W_{\mu}( e_a,e_3,e_b)\!&=&\!\!\left[4\de\om\aa^{(0)}-3(\de\hat{\chib}\ro^{(0)}-\dual\de\hat{\chib}\si^{(0)})+(\de\zeta-4\de\etab)\hot\bb^{(0)}\right]\nn\\
\tilde{p}^{\mu} W_{\mu}( e_3,e_3,e_b)\!&=&\!-\left[2\de\omb\bb^{(0)}+(-2\de\zeta+ \de\eta)\c\aa^{(0)}\right]\nn\\
\tilde{p}^{\mu} W_{\mu }(e_4,e_3,e_b)\!&=&\!\left[2\de\om\bb^{(0)}+2\de\hat{\chib}\c\b^{(0)}+(\dual\nabb-\dual\nabb^{(0)})\si^{(0)}-3(\de\etab\ro^{(0)}-\de\dual\etab\si^{(0)})\right]\nn\\
\tilde{p}^{\mu} W_{\mu }(e_3,e_4,e_3)\!&=&\!-\left[{2^{-1}}\de\hat{\chi}\c\aa^{(0)}-\de\zeta\c\bb^{(0)}+2\de\eta\c\bb^{(0)}\right]\nn\\
\tilde{p}^{\mu} W_{\mu }(e_4,e_3,e_4)\!&=&\!-\left[{2^{-1}}\de\hat{\chib}\c\a^{(0)}-\de\zeta\c\b^{(0)}-2\de\etab\c\b^{(0)}\right]\nn\\
\tilde{p}^{\mu} W_{\mu }(e_3,e_4,e_3)\!&=&\!\left[{2^{-1}}\de\hat{\chi}\c\dual\aa^{(0)}-\de\zeta\c\dual\bb^{(0)}-2\de\eta\c\dual\bb^{(0)}\right]\eql{Hdef1}\\
\tilde{p}^{\mu} W_{\mu }(e_4,e_3,e_4)\!&=&\!\left[{2^{-1}}\de\hat{\chib}\c\dual\a^{(0)}-\de\zeta\c\dual\b^{(0)}-2\de\etab\c\dual\b^{(0)}\right]\nn\\
\tilde{p}^{\mu} W_{\mu }(e_3,e_4,e_a)\!&=&\!\left[2\de\omb\b^{(0)}+(\dual\nabb-\dual\nabb^{(0)})\si^{(0)}+2\de\hat{\chi}\c\bb^{(0)}+3(\de\eta\ro^{(0)}+\de\dual\eta\si^{(0)})\right]\nn\\
\tilde{p}^{\mu} W_{\mu }(e_4,e_a,e_4)\!&=&\!-\left[2\de\om\b^{(0)}-(2\de\zeta+\de\etab)\a^{(0)}\right]\nn\\
\tilde{p}^{\mu} W_{\mu }(e_a,e_4,e_b)\!&=&\!\left[4\de\omb\a^{(0)}-3(\de\hat{\chi}\ro^{(0)}+\de\dual\hat{\chi}\si^{(0)})+(\de\zeta+4\de\eta)\hot\b^{(0)}\right]\nn
\eea

\NI We can now calculate   the left hand side of \ref{p} , $\tilde{p}^{\mu} W_{\mu\nu\ro\si}=h( e_{\nu},e_{\ro},e_{\si})$.  The worst decays are those involving $\ro^{(0)}$ and $\si^{(0)}$, the other one give better decays for $\tilde{p}^{\mu}$:

\begin{prop}
Assumed the decays of proposition 1.1 for the null Riemann components in Kerr, and the decays of prop 1.2, by equation 1.14, the following decays hold for $h$

\bea\label{hdecay}
&&h_{a3b}\simeq 3\de\hat{\chib}\ro^{(0)}\simeq \frac{1}{r^4}\frac{1}{u^3}\nn\\
&&h_{33b}\simeq\frac{1}{r^5}\frac{1}{u^3}\nn\\
&&h_{43b}\simeq  (\ ^*\nabb-\ ^*\nabb_0)\si^{(0)}-3(\de\etab\ro^{(0)}-\de^*\etab\si^{(0)})\simeq \frac{1}{r^5}\frac{1}{u^2}\nn\\
&&h_{343}\ , \ h_{434}\simeq\frac{1}{r^6}\frac{1}{u^2}\nn\\
&&\ ^*h_{343}\ ,\ \ ^*h_{434} \ , \ \ ^*h_{343}\simeq\frac{1}{r^6}\frac{1}{u^2}\\
&&h_{34a}\simeq  (\ ^*\nabb-\ ^*\nabb_0)\si^{(0)}-3(\de\eta\ro^{(0)}-\de^*\eta\si^{(0)})\simeq \frac{1}{r^5}\frac{1}{u^2}\nn\\
&&h_{4a4}\simeq\frac{1}{r^5}\frac{1}{u^3}\nn\\
&&h_{a4b}\simeq 3 \delta \hat{\chi} \ro^{(0))}\simeq \frac{1}{r^5}\frac{1}{u^2}\nn
\eea
\end {prop}
\NI Imposing these decays in \ref{p} we can easily calculate the decays of $\tilde{p}^{\mu}$, we only consider the terms corresponding to the worst decays, namely $h_{a3b},\ h_{33b}, \ h_{34a} \ h_{a4b}$.
For example the first term implies:

\bea
&&h_{a3b}\simeq \frac{1}{r^4}\frac{1}{u^3}=\tilde{p}^{\mu} W_{\mu e_a e_3 e_b}=\tilde{p}^{3} W_{4a4b}+\tilde{p}^{4} W_{3a4b}+\tilde{p}^{c} W_{ca4b}\nn\\
\eea 

\NI which at its turn give for $\tilde{p}^{\mu}$:

\bea\label{fra}
&&\tilde{p}^3\simeq\frac 1 r \frac 1 u \ , \ \tilde{p}^4\simeq\frac {1}{ r^3} \cdot u \ , \ \tilde{p}^a\simeq\frac {1}{r^2} 
\eea

\NI It is easy to see that the other potentially harmfull terms give the same decays for $\tilde{p}^{\mu}$.

\NI Is this result satisfying? In other words, are the decays of $\tilde{p}^{\mu}$ sufficient to prove the boundedness of the $\cal{Q}$ norms?
In order to answer this question let us notice that the only difference in the estimates for $\cal{Q}$ in this case is that a new term appear, we will call it the $J^0$ current,  
for the $Div \cal{Q}$ terms in $\mathcal{E}_1$. If we want to calculate for example $Div {\cal{Q}}(\lie_O \de W (\ \ \ ))$, the $J^0$ current associate is:

\bea
&&\lie_O Div(\de W)_{\cdot\cdot\cdot}=\lie_O h\simeq h\simeq \tilde{p}^{\nu}W_{\nu\cdot\cdot\cdot}\nn
\eea

with the $\tilde{p}$  coefficients associated to this term decaying as in \ref{fra}:

\bea
&&\tilde{p}\simeq \frac 1 r \frac 1 u \ ,\ \tilde{p}\simeq u\cdot \frac{1}{r^3}\ ,\ \tilde{p}\simeq \frac{1}{r^2}\nn
\eea

\NI The key observation is that these decays are in perfect accordance with another current, namely, the decays of the $p^{\mu}$ terms of the $J^2$ current associated to the $\ ^{O}\pi$ deformation tensor necessary to estimate
$\lie_O Div ( W)$ , see [Kl-Ni] eq. 6.1.55, hence we can estimate this term exactly as in the nonlinear case, as the added $J^0$ current behave like this $J^2$ current.
The same happen for the other term $\lie_T Div( W)$ in $\mathcal{E}_1$ with the only difference that we gain a decay in $u$ due to the $\lie_T$ derivative. Also in this case we are in prefect accordance 
with the $J^2$ current associated to $\ ^T\pi$, see [Kl-Ni], 6.1.46.

\medskip

\NI For what concern $\mathcal{E}_2$ let us notice that also in this case we have to add to the divergence terms, see \cite{Kl-Ni1} eq. (6.0.6) a term like 

\bea
&&\lie_X\lie_O Div(\de W)_{\cdot\cdot\cdot}=\lie_X\lie_O h\nn
\eea

\NI or

\bea
&&\lie_X\lie_T Div(\de W)_{\cdot\cdot\cdot}=\lie_X\lie_T h\nn
\eea

\NI with $X=\{O,S\}$

\NI It is easy to see that all these terms behaves as well as or better than the corresponding terms analyzed for  $\mathcal{E}_1$.

\medskip

\NI By this crucial observation we can assure that the source term $h$ can be inserted in the estimate of the error and it will generate new terms which can be easily treated as the other ones already estimated in \cite{Kl-Ni1}.

\smallskip
\begin{remark}

\NI The other possible approach would be to directly assign $\tilde{p}$ in such a way the error term associated to the $\cal{Q}$ norms can be bounded, this approach would require a deep investigations of all the terms involved 
in the $J$ currents involved in the deformation tensors $\pi$. It is not clear in principle if the decays which can be obtained in this case are in accordance with the ones already obtained in equation \ref{hdecay}.
We will not investigate this method.

\end{remark}

\smallskip

\NI In the next  section we introduce the right smallness conditions  for the initial data then we can state the main theorem.
In the third section we prove the result emphasizing the steps we need to perform in order to mimick the [Kl-Ni] result.

\newpage

\subsection{Smallness conditions on the initial data}
\NI In this subsection we obtain the right smallness conditions we have to require on $\Si_0$ to bound the $\QQ_{\Si_0\cap V(u,\ub)}$ norms by a suitable constant $\ep<<1$.
To do this we adapt the conditions already obtained in [Ch-Kl] and modified in [Kl-Ni]  and [Ca-Ni].

\NI Let us recall the smallness condition obtained in [Kl-Ni], chap. 2:

\NI Given an initial data hypersurface $\Si_0$ and a compact set $\cal{B}$ on it such that $\Si_0/ \cal{B}$ is diffeomorphic to the complement of the unit ball in ${\cal{R}}^3$ , and $(g,k)$ initial data on $\Si_0/ {\cal{B}}$ , we define 
${\cal{J}}_B$ as follows:

\NI i) Let us denote with $\cal{G}$ the set of all smooth extensions $(\tilde{g},\tilde{k})$ to the whole spacetime of $\Si_0$ of $(g,k)$, with $\tilde{g}$ Riemannian and $\tilde{k}$ a symmetric 2-tensor;

\NI ii) Let us denote with $\tilde{d}_0$ the geodesic distance from a fixed point $O$ in $\cal{B}$ relative to the metric $\tilde{g}$;

\NI iii) We denote

\bea\label{3.421aa}
{\cal J}_B({\Si_0},{g} ,k)= \inf_{\cal{G}} {\cal J}_0(\Si_0, \tilde{g} ,\tilde{k})\nn\\
\eea

\NI with 

\bea\label{3.412a}
&&{\cal J}_0(\tilde{g},{\tilde{k}})=\sup_{\Si_0}[(d^0+1)^2|Ric|^2]\\ 
&&+\int_{\Si_0}\sum_{l=0}^3(1+d_0^2)^{(1+l)+\frac{3}{2}+\de}|\nabb^l k|^2\\
&&+\int_{\Si_0}\sum_{l=0}^1(1+d_0^2)^{(3+l)+\frac{3}{2}+\de}|\nab lB| 2^{\frac{1}{2}}\ .\nn
\eea

\NI With $\nabb$ the covariant derivative related to $\tilde{g}$ $B=\ep^{ab}_j\nabb(R_{ib}-\frac{1}{4}\tilde{g}_{ib}R)$ the Bauch tensor, see the introduction of [Ch-Kl], and $Ric$ the Ricci tensor relative to the metric $\tilde{g}$.

\NI Before passing to the linear case we remark that, roughly speaking the definition of $J_0$ allows us to bound the${\cal{Q}}_{\Si_0\cap\cal{K}}$ for any $\cal{K}$, i.e. globally on the initial data, it can be shown
that if we want to bound $\cal {Q}$ norms with different weights we have to add these weights in the integral parts of $J_0$, this is exactly what has be done in [Kl-Ni1] and [Ca-Ni] where factors respectively $\frac{5}{2}$
and $\frac {3}{2}$ in $d_0^2$ have been added. Notice that, in fact, the integrals in $k$ and $B$ and their derivatives as well as the supremum for $Ric$ all are needed to estimate the null Riemann components on $\Si_0$
in a suitable way to bound the ${\cal{Q}}_{\Si_0\cap\cal{K}}$ norms.

\NI In our linear case the main difference is that in principle $W$ have no relations with $g_{Kerr}$ and so  $J_0$ as defined above cannot be used to bound the null components of  $W$. Hence we have to define a smallness condition directly related to the smallness of the ${\cal{\tilde{Q}}}_{\Si_0\cap\cal{K}}$ norms. It can be easily shown from the relation between the $\cal{Q}$ norms and the null $W$ components see [Kl-Ni] equations 3.5.1, 3.5.2, that the requirement ${\cal{Q}}_{\Si_0\cap\cal{K}}\leq C\ep_0$ amounts to ask that $r^{5+\ep} {\cal{R}}\leq C_1\ep_0$, with $C$ and $C_1$ constant.
We can now state our smallness condition for the linear case 

\NI Given an initial data hypersurface $\Si_0$ and a compact set $\cal{B}$ on it such that $\Si_0/ \cal{B}$ is diffeomorphic to the complement of the unit ball in $R^3$ , and $W$ a Weil field on $\Si_0/ \cal{B]}$ , denoting as $\{\a, \ab, \b,...\}$ its null components, we define 
${\cal{J}}_{lin}:$

\bea\label{defj}
{\cal{J}}_{lin}(W):=r^5(|\a|+|\ab|+|\b|+|\bb|+|\ro|+|\si|)
\eea

\NI With this definition we can finally state the main theorem.

\subsection {the main theorem}

\begin{thm}\label{maintheorem}

\NI Given the Kerr spacetime  let us  consider $\Si_0$ the hypersurface corresponding to $t=0$ in Boyer-Lindquist coordinates and ${\cal{ K}}$ a compact set such that
$\Si_0/ { \cal{ K}}$ is contained in the external region. Moreover let us consider the  null coordinates $\{ u, \ub, \theta, \phi \}$ related to the \cite{Is-Pr} null cones foliations, and the associated null frame $\{e_{\mu}\}$.
Let us assign on $\Si_0/ { \cal {K}}$  a Weyl field $W$ in such a way that  ${\cal{J}}_{lin}(W)\leq C\ep_0$ with $C$ and $\ep_0$ constant, $\ep_0<<1$ and  $J_{Lin}$ defined in \ref{defj}, moreover let us assume $W$  satisfies the solution of the massless spin 2 equation:
\bea
&&D_{(0)}^{\nu}W_{\nu\mu\ro \si}=h_{\mu\ro\si}
\eea

\NI Let us assume that the inhomogeneous term $h_{\nu\ro\si}$ decay, with respect to the null frame in the following way:

\bea
&&h_{a3b}\simeq \frac{1}{r^4}\frac{1}{u^3}\ ,\ h_{33b}\simeq\frac{1}{r^5}\frac{1}{u^3}\ ,\  h_{43b}\simeq \frac{1}{r^5}\frac{1}{u^2}\nn\\
&&h_{343}\simeq\frac{1}{r^6}\frac{1}{u^2}\ ,\  h_{434}\simeq\frac{1}{r^6}\frac{1}{u^2}\ ,\ \ ^*h_{343}\simeq\frac{1}{r^6}\frac{1}{u^2}\nn\\
&&\ ^*h_{434}\simeq\frac{1}{r^6}\frac{1}{u^2}\ ,\ \ ^*h_{343}\simeq\frac{1}{r^6}\frac{1}{u^2}\, \ h_{34a}\simeq \frac{1}{r^5}\frac{1}{u^2}\nn\\
&&h_{4a4}\simeq\frac{1}{r^5}\frac{1}{u^3}\ , \ h_{a4b}\simeq \frac{1}{r^5}\frac{1}{u^2}\nn
\eea

\NI With $r$ the radial coordinate in Boyer-Lindquist coordinates.

\NI Then the null components of $W$ satisfies the peeling theorem decays, more precisely they decay in the following way:

\bea
&& \a= O(\frac 1{ r^5} \frac {1}{ u^{\ep}})\ , \  \ab= O(\frac 1 r \frac {1}{ u^{4+\ep}})\nn\\
&& \b= O(\frac 1 {r^4} \frac {1 }{u^{1+\ep}})\ , \  \bb= O(\frac 1 {r^2} \frac{ }{1 u^{3+\ep}})\nn\\
&& \ro= O(\frac 1 {r^3} \frac{ 1 }{u^{2+\ep}})\ , \  \si= O(\frac 1 {r^3} \frac {1}{ u^{2+\ep}})\nn
\eea

\end{thm}

\begin {remark}
Clearly in the main theorem $W$  is only the $\de R$ part of $R+\de R$ the peeling theorem for The full perturbation of Kerr follows by the linerarity of $D_0$ and the decays of the null Riemann components of Kerr
in the \cite{Is-Pr} foliation, which satisfy the peeling.
\end{remark}

\subsection{Proof of the results}

\NI The central technical part of this work is to show that the right $\cal Q$ norms  we will introduce for for $ W$ are bounded. This is very long to prove, Substantially it is a repetition of what has been done in Chapter 6 of \cite{Kl-Ni1} with two main differences.

\medskip

\NI first, as said before, we have to prove that the source term we have imposed a priori $h$ allow us to estimate the error terms. This point has been yet discussed in section 1.4

\medskip

\NI Second we have to estimate the error term. This will require a long calculation to show the connection coefficients related to the \cite{Is-Pr} foliation in Kerr have the right decays to obtain the same estimate of \cite{Kl-Ni1}.
 We remark as this calculation gives better decays with respect to the ones which can be obtained by dimentional argument, see for example \cite{Ca-Ni}, and can be considered as a good result itself potentially useful for other applications.

\NI Once these computations and the boundedness of the error term has been achieved, we obtain the following decays for the null Riemann components associated to $W$: 

\begin{eqnarray}
&&\sup_{\mathcal{K}}r^{\frac{7}{2}}|u|^{(\frac{5}{2}+\e')}|\a(W)| \leq C_0,
\quad\sup_{\mathcal{K}}r^{\frac{7}{2}}|u|^{\frac{5}{2}+\e'} |\b(W)|\leq C_0\nn\\
&&\sup_{\mathcal{K}}r^3|u|^{3+\e}|\ro W)| \leq C_0,\quad \sup r^3|u|^{3+\e'}|\sigma(W)|\leq C_0\eql{LTWdec1}\\
&& \sup_{\mathcal{K}}r^2|u|^{4+\e'}|\bb(W)|\leq C_0,\quad\sup_{\mathcal{K}}r|u|^{5+\e' }|\aa(W)|\leq C_0\ ."\nn
\end{eqnarray}

\NI  The decays obtained in \ref{LTWdec1} are not yet satisfying for $\a(\lie_{T_0}W)$ and $\b(\lie_{T_0}W)$, compare with \ref{LTWdec}. We have to bargain the $u$-decay factor with an $r$-decay factor. This is done, following the \cite{Kl-Ni2} approach, using the transport equations for the Weyl field along the null incoming hypersurfaces. This will allows to go from the inequalities  \ref{LTWdec1} to inequalities \ref{LTWdec}. We will not 
show this step as it is identical to  \cite{Kl-Ni2}.
\medskip

\newpage

\section{Preliminary definitions and properties}

\NI We want to describe in more detail the techniques we have to exploit to bound the generalized energy norms
and to give a complete picture of the analytic tools necessary to prove our result, paying attention to
explaining the logic which we tackle the problem with.\\
As we shall see, the crucial difference treating the linearized version of the problem is that we already know the
background spacetime, specifically the  Kerr spacetime, someone of them have been already introduced, we repeat them here for the sake of completness.
Let us first introduce some concepts and definitions:

\NI Let us consider the Kerr spacetime $(M,g)$ in the Boyer Lindquist coordinates, we consider
on it :

\medskip

\NI 1) The initial data hypersurface $\Si_0+\{p\in M | t(p)+0\}$

\medskip

\NI 2) The double null foliation, consisting in the double family of null hypersurfaces
     $C(u)=\{p\in M | u(p)=u\}$ and $\underline{C}(\underline{u})=\{p \in M | \underline(u)(p)=\underline{u}\}$ where $u$ and $\underline{u}$ are the solution of the eikonal equation
     $g^{\a\b}\partial\a u\partial\b u=0$      ($g^{\a\b}\partial\a\underline{u}\partial\b\underline{u}=0$) and passing trough the sphere $\{\Si_0\cap r=u\}$ (  $\{\Si_0\cap r=\underline{u}\} )$ \footnote {these hypersurfaces are the analogous of the null incoming and outgoing null cones in Minkowski}

\medskip

\NI 3) The sphere foliation $S(u, \underline{u})=C(u)\cap C(\underline{u})$

\medskip

\NI 4) The sphere foliation $\tilde{S}(t,r)=\{ p | t(p)=t r(p)=r \}$

\medskip

\NI Let us notice that we can foliate the initial data hypersurface $\Si_0$ with both the foliations in $S$ and
$\tilde{S}$.  In the case of Kerr spacetime every $S$ is also an $\tilde{S}$

\medskip

5) A null hortonormal frame $\{ e_4 , e_3 ,  e_2 , e_1 \} $  adapted to the double null foliation such that:
    
 $<e_4,e_4>=0$ , $<e_3,e_3>=0$  , $<e_3,e_4>=-2$ , $ <e_a,e_3>=<e_a,e_4>=0$, $<e_a,e_b>=\delta_{ab}$  with $a,b>0$ and $e_a \in TS(u,\underline{u}$

\medskip

\NI In order to prove the expected results about the asymptotic behavior of spin 2 zero-rest mass fields, we also introduce the concept and the main properties of the Weyl fields: 
\begin{defn}
\NI Given a spacetime $(\mathcal{M},g)$, a Weyl field is a tensor
field $W$ which satisfies the following properties
\begin{eqnarray}
& &W_{\a\b\ga\de}=
W_{\ga\de\a\b}=-W_{\b\a\ga\de}=-W_{\a\b\de\ga}\nonumber\\
& &W_{\a\b\ga\de}+W_{\a\ga\de\b}+W_{\a\de\b\ga}=0\\
& &g^{\a\ga}W_{\a\b\ga\de} = 0.\nonumber
\end{eqnarray}
\end{defn}
\begin{defn}
\NI A Weyl tensor field $W$ is a solution of the 2-spin and zero-rest
mass field equations (or Bianchi equations) in
$(\mathcal{M},g)$ if, relative to the Levi-Civita connection of
$g$, it satisfies
$$
\dd^\mu W_{\mu\nu\ro\si}=0.
$$
\end{defn}

\begin{defn}
\NI Let $X$ be a vector field, then the deformation tensor of $X$ is
defined in the following way:
$$
{^{(X)}\pi}_{\mu\nu}=\mathcal{L}_X g_{\mu\nu}=D_\mu X_\nu+D_\nu
X_\mu
$$
\NI and its traceless part is
$$
{^{(X)}\hat{\pi}}_{\mu\nu}={^{(X)}\pi}_{\mu\nu}-\frac 1 4
g_{\mu\nu}\tr{^{(X)}}\pi.
$$
\end{defn}
\NI Then, if $X$ is a Killing vector field, it follows
$$
{^{(X)}\pi}=0.
$$

\NI Given a null frame $\{e_3,e_4,e_a\}_{a=1,2}$, let's decompose
the $X$ deformation tensor  with respect to it:
\begin{eqnarray}\label{pih}
{^{(X)}}\pi_{ab} &=& g(D_{e_a}X,e_b)+g(D_{e_b}X,e_a)\nn\\
{^{(X)}}\pi_{a4} &=& g(D_{e_a}X,e_4)+g(D_{e_4}X,e_a)\nn\\
{^{(X)}}\pi_{a3} &=& g(D_{e_a}X,e_3)+g(D_{e_3}X,e_a)\nn\\
{^{(X)}}\pi_{34} &=& g(D_{e_3}X,e_4)+g(D_{e_4}X,e_3)\\
{^{(X)}}\pi_{44} &=& 2g(D_{e_4}X,e_4)\nn\\
{^{(X)}}\pi_{33} &=& 2g(D_{e_3}X,e_3)\nn\ 	.
\end{eqnarray}
\NI Now let's introduce the following notation for their tracelees part :
\begin{eqnarray}
&&{^{(X)}}i_{ab}  = {^{(X)}}\pih_{ab} = {^{(X)}}\pi_{ab}-\frac 1 4
\de_{ab}\tr{^{(X)}}\pi\nn\\
&&{^{(X)}}m_a= {^{(X)}}\pih_{a4} = {^{(X)}}\pi_{a4}\nn\\
&&{^{(X)}}\un{m}_a = {^{(X)}}\pih_{a3} = {^{(X)}}\pi_{a3}\\
&&{^{(X)}}j  = {^{(X)}}\pih_{34} = {^{(X)}}\pi_{34}+\frac 1 2 \tr{^{(X)}}\pi\nn\\
&&{^{(X)}}n = {^{(X)}}\pih_{44} = {^{(X)}}\pi_{44\nn}\\
&&{^{(X)}}\un{n} = {^{(X)}}\pih_{33} = {^{(X)}}\pi_{33}\nn\	.
\end{eqnarray}
\NI We call them the null components of the ${^{(X)}}\pi$.
\begin{defn}
\NI Given a Weyl tensor field $W$ and a vector field $X$, we define
the modified Lie derivative relative to $X$ by
\begin{equation}
\hat{\mathcal{L}}_X W = L_X W-\frac 1 2 {^{(X)}}[W]+\frac 3 8
\tr{^{(X)}}\pi W,
\end{equation}
\NI where
$$
{^{(X)}}[W]_{\a\b\gamma\delta}={^{(X)}}\pi_\a^\mu
W_{\mu\b\gamma\delta}+{^{(X)}}\pi_\b^\mu
W_{\a\mu\gamma\delta}+{^{(X)}}\pi_\gamma^\mu
W_{\a\b\mu\delta}+{^{(X)}}\pi_\delta^\mu W_{\a\b\gamma\mu}.
$$
\end{defn}
\begin{remark}
\NI The modified Lie derivative of a Weyl field is a Weyl field too but in general the modified derivative of a Weyl field does not satisfies the Bianchi equations anymore.Nevertheless if $X$ is Killing or conformal Killing  then it  also satisfies the Bianchi equations.
\end{remark}
\NI Next we introduce the null decomposition of a Weyl tensor, i.e. we
express $W$ in terms of the null frame $\{e_4,e_3,e_1,e_2\}$ in the
following way:
\begin{defn}
\NI For every point $p\in\cal{M}$ we define the following tensors on the tangent
space to the sphere $S(u,\un{u})$ passing through $p$ (null Riemann components):
\begin{eqnarray}\label{nulldec}
& &\a(W)(X,Y) = W(X,e_4,Y,e_4),\qquad \un{\a}(W)(X,Y) =
W(X,e_3,Y,e_3)\nonumber\\
& &\b (W)(X) = \frac 1 2 W(X,e_4,e_3,e_4),\qquad  \un{\b}(W)(X) =
\frac 1 2
W(X,e_3,e_3,e_4)\nn\\
& &\ro (W) = \frac 1 4 W(e_3,e_4,e_3,e_4), \qquad \si(W)=\frac 1 4
\rdual {W}(e_3,e_4,e_3,e_4)\	,
\end{eqnarray}
\NI where $\rdual{W}_{\a\b\ga\de}$ is the left Hodge dual of $W$,
defined in the following way:
$$
\rdual{W}_{\a\b\ga\de}= \frac 1 2
\e_{\a\b\phi\psi}{W^{\phi\psi}}_{\ga\de}.
$$
\end{defn}
\begin{prop}[Bianchi Equations]
\NI Expressed relatively to an adapted null frame, the Bianchi
equations take the following form
\begin{eqnarray}\label{bianchi}
\un{\a}_4 &\equiv & \dddd_4\un{\a}+\frac 1 2
\tr\chi\un{\a}=-\nabb\hat{\otimes}\un{\b}+4
\omega\un{\a}-3(\un{\hat{\chi}}\rho-{^\star}\hat{\un{\chi}}\si)+(\zeta-4\un{\eta}\hat{\otimes})\un{\b}\nonumber\\
\un{\b}_3 &\equiv & \dddd_3\un{\b}+2\tr\un{\chi}\un{\b} =
-\divv\un{\a}-2\un{\omega}\un{\b}+(2\zeta-\eta)\cdot\un{\a}\nonumber\\
\un{\b}_4 &\equiv & \dddd_4\un{\b}+\tr\chi\un{\b} =
-\nabb\rho+2\omega\un{\b}+2\hat{\un{\chi}}\cdot\b+{^\star}\nabb\si-3(\un\eta\rho-{^\star}\un{\eta}\si)\nonumber\\
\rho_3 &\equiv & \dd_3\rho+\frac 3 2 \tr\un{\chi}\rho
=-\divv\un{\b}-\frac 1 2
\hat{\chi}\cdot\un\a+\zeta\cdot\un{\b}-2\eta\cdot\un{\b}\nonumber\\
\rho_4 &\equiv & \dd_4\rho+\frac 3 2 \tr\chi\rho =\divv\b-\frac 1
2 \hat{\un{\chi}}\cdot\a+\zeta\cdot\b+2\un{\eta}\cdot\b\\
\si_3 &\equiv & \dd_3\si+\frac 3 2 \tr\un{\chi}\si
=-\divv\rdual{\bb}+\frac 1 2
\hat{\chi}\c\rdual{\aa}-(\zeta+2\eta)\c\rdual{\bb}\nonumber\\
\si_4 &\equiv & \dd_4\si+\frac 3 2 \tr\chi\si
=-\divv\rdual{\b}+\frac 1 2
\hat{\un{\chi}}\c\rdual{\a}-(\zeta+2\un{\eta})\c\rdual{\b}\nonumber\\
\b_3 &\equiv &\dddd_3\b+\tr\un{\chi}\b =
\ddd\ro+\rdual{\ddd}\si+2\un{\omega}\b+2\hat{\chi}\c\bb+3(\eta\ro+\rdual{\eta}\si)\nonumber\\
\b_4 &\equiv &\dddd_4\b+2\tr\chi\b =
\divv\a-2\omega\b+(2\zeta+\un{\eta})\a\nonumber\\
\a_3 &\equiv& \dddd_3\a+\frac 1 2 \tr\un{\chi}\a =
\ddd\hat{\otimes}\b+4\un{\omega}\a-3(\hat{\chi}\ro+\rdual\hat{\chi}\si)+(\z+4\eta)\hat{\otimes}\b,\nonumber
\end{eqnarray}
\NI where, here, $\dddd_4$ and $\dddd_3$ are the projections on the
tangent space to $S(u,\un{u})$ of the covariant derivatives along
$e_3,e_4$, $\divv$ and $\ddd$ are the projections on the tangent
space to $S(u,\un{u})$ of the divergence and the covariant
derivative relative to $\Si_t$, and $\hat{\otimes}$ denotes twice
the traceless part of the symmetric tensor product. The Hodge
operator $\rdual$ indicates the dual of the tensor fields relative
to the tangent space of $S(u,\un{u})$, in particular
\begin{defn}\label{rdual}
\NI Given the 1-form $\psi$ defined on $S(u,\un{u})$, we define its
Hodge dual:
$$
\rdual{\psi}_a=\e_{ab}\psi_b,
$$
\NI where $\e_{ab}$ are the components of the area element of
$S(u,\un{u})$ relative to an orthonormal frame $(e_a)_{a=1,2}$.\\
If $\psi$ is a symmetric traceless 2-tensor, we define the
following left, $\rdual{\psi}$, and right,$\psi^\star$, Hodge
duals:
$$
\rdual{\psi}_{ab}=\e_{ac}{\psi^c}_b,
\psi^\star_{ab}={\psi_a}^c\e_{cb}.
$$
\end{defn}
\NI see [12] , prop 3.2.4 pag. 77
\end{prop}
\NI Once we have introduced a Weyl field which satisfies the Bianchi
equations, we are able to define the Bel-Robinson tensor
associated to it, in the following way:
\begin{defn}
\NI The Bel-Robinson  tensor field associated to the Weyl tensor $W$
is the 4-covariant tensor field
\begin{eqnarray*}
Q_{\a\b\ga\de}[W]&=&
W_{\a\ro\ga\si}{{{W_\b}^\ro}_\de}^\si+{\rdual{W}_{\a\ro\ga\si}}\rdual{W}{{{_\b}^\ro}_\de}^\si\\
&=&
W_{\a\ro\ga\si}{{{W_\b}^\ro}_\de}_\si+W_{\a\ro\de\si}{{{W_\b}^\ro}_\ga}_\si-\frac
1 8 g_{\a\b}g_{\ga\de}W_{\ro\si\mu\nu}W^{\ro\si\mu\nu}.
\end{eqnarray*}
\end{defn}
\NI The Bel-Robinson tensor satisfies the following important
\begin{prop}\label{BelRob}\ \ \ \ \ \ \ \ \ \ \ \ \ \ \ \ \ \ \ \ \ \ \ \ \ \ \ \ \ \ \ \ \ \ \ \ \ \ \ \ \ \ \ \ \ \ \ \ \ \ \ \ \ \ \ \ \ \ \ \ \ \ \ \ \ \ \
\smallskip

\NI i) $Q$ is symmetric and traceless relative to all pairs of
indices.\\

\NI ii) $Q$ satisfies the following positivity condition: given any
timelike vector fields $X_\mu$, for $\mu=1,...,4$
$$
Q(X_1,X_2,X_3,X_4)> 0
$$
\NI unlike $W=0$.\\

\NI iii) If $W$ is a solution of the Bianchi equations, it follows
$$
D^\a Q_{\a\b\ga\de}=0.
$$
\end{prop}
\NI For the proof, see \cite{Ch-Kl2}.
\begin{prop}
\NI Let $Q(W)$ be the Bel-Robinson tensor of a Weyl field $W$ and
$X,Y,Z$ a triplet of vector fields in$\mathcal{M}$. We define the
1-form $P$ associated at the triplet as
\begin{equation}\label{P1}
P_\a=Q_{\a\b\ga\de}X^\b Y^\ga Z^\de.
\end{equation}
\NI Using all the symmetry properties of $Q$, we have:
\begin{eqnarray}\label{P}
\Div P &=& \Div Q_{\b\ga\de}X^\b Y^\ga Z^\de\\
&+& \frac 1 2 Q_{\a\b\ga\de}\bigl({^{(X)}}\pi^{\a\b}Y^\ga
Z^\de+{^{(Y)}}\pi^{\a\ga}X^\b Z^\de+{^{(Z)}}\pi^{\a\de}X^\b Y^\ga
\bigr).\nonumber
\end{eqnarray}
\end{prop}
\begin{remark}
\NI When $X,Y,Z$ are Killing or conformal Killing vector fields and
$W$ satisfies Bianchi equations, it follows
$$
\Div P=0
$$
\NI i.e. $P$ is a conserved quantity.
\end{remark}

\subsection{$\tilde{\mathcal{Q}}$ integral norms}
\NI  Now we  show which are the suitable integral norms to introduce on
$C(u),\un{C}(\ub)$ and on the initial hypersurface $\Si_0$  in order to find the peeling decays to estimate 
these $L_2$ norms. 

\NI From now on, let  us indicate by  $\tW$ a Weyl tensor field that satisfies Bianchi equations; we shall construct the energy norms starting from the Bel-Robinson tensor associated to it. 
We denote these norms with $\QQ[\tW]$. Using the vector fields $\bar{K},S,T$ and ${^{(i)}O}$ and denoting
$$
V(u,\ub)=J^-(S(u,\ub)),
$$
\NI we define the following energy-type norms:

\begin{eqnarray}\label{norme1}
&& \QQ(u,\ub) = \QQ_1(u,\ub)+\QQ_2(u,\ub)\nn\\
&& \QQb(u,\ub) = \QQb_1(u,\ub)+\QQb_2(u,\ub)\	,
\end{eqnarray}
\NI where
\begin{eqnarray}\label{norms1}
\QQ_1(u,\ub) &\equiv & \int_{C(u)\cap V(u,\ub)}\ttau^{5+\e}Q(\lie_T
{\tW})(\bar{K},\bar{K},\bar{K},e_4)\nonumber\\
& &+\int_{C(u)\cap V(u,\ub)}\ttau^{5+\e}Q(\lie_O
\tW)(\bar{K},\bar{K},T,e_4)\nonumber\\
\QQ_2(u,\ub) &\equiv &\int_{C(u)\cap V(u,\ub)}\ttau^{5+\e}Q(\lie_O\lie_T
\tW)(\bar{K},\bar{K},\bar{K},e_4)\nonumber\\
& &+\int_{C(u)\cap V(u,\ub)}\ttau^{5+\e}Q(\lie^2_O
\tW)(\bar{K},\bar{K},T,e_4)\label{Q_2}\\
& & +\int_{C(u)\cap V(u,\ub)}\ttau^{5+\e}Q(\lie_S\lie_T
\tW)(\bar{K},\bar{K},\bar{K},e_4)\nonumber
\end{eqnarray}
\begin{eqnarray}\label{norms2}
\QQb_1(u,\ub)& \equiv &
\int_{\un{C}(\ub)\cap
V(u,\ub)}\ttau^{5+\e}Q(\lie_T \tW)(\bar{K},\bar{K},\bar{K},e_3)\nonumber\\
& &+\int_{\un{C}(\ub)\cap V(u,\ub)}\ttau^{5+\e}Q(\lie_O
\tW)(\bar{K},\bar{K},T,e_3).\nonumber
\end{eqnarray}
\begin{eqnarray}
\QQb_2(u,\ub) &\equiv &\int_{\un{C}(\ub)\cap
V(u,\ub)}\ttau^{5+\e}Q(\lie_O\lie_T
\tW)(\bar{K},\bar{K},\bar{K},e_3)\nonumber\\
& &+\int_{\un{C}(u)\cap V(u,\ub)}\ttau^{5+\e}Q(\lie^2_O
\tW)(\bar{K},\bar{K},T,e_3)\label{Qb_2}\\
& & +\int_{\un{C}(\ub)\cap V(u,\ub)}\ttau^{5+\e}Q(\lie_S\lie_T
\tW)(\bar{K},\bar{K},\bar{K},e_3)\nonumber
\end{eqnarray}
\NI and
\begin{eqnarray}
\QQ_{1_{\Si_0\cap V(u,\ub)}} &\equiv & \int_{\Si_0\cap
V(u,\ub)}\ttau^{5+\e}Q(\lie_T \tW)(\bar{K},\bar{K},\bar{K},T)\nonumber\\
& & +\int_{\Si_0\cap V(u,\ub)}\ttau^{5+\e}Q(\lie_O
\tW)(\bar{K},\bar{K},T,T)\\
\QQ_{2_{\Si_0\cap V(u,\ub)}} &\equiv & \int_{\Si_0\cap
V(u,\ub)}\ttau^{5+\e}Q(\lie_O\lie_T \tW)(\bar{K},\bar{K},\bar{K},T)\nonumber\\
& & +\int_{\Si_0\cap V(u,\ub)}\ttau^{5+\e}Q(\lie_O^2
\tW)(\bar{K},\bar{K},T,T)\nonumber\\
& &+ \int_{\Si_0\cap V(u,\ub)}\ttau^{5+\e}Q(\lie_S\lie_T
\tW)(\bar{K},\bar{K},\bar{K},T).
\end{eqnarray}
\NI We introduce also the following quantity
\begin{equation}
\QQ_{\cal{K}}\equiv
\sup_{\{u,\ub|S(u,\ub)\subseteq\cal{K}\}}\{\QQ(u,\ub)+\QQb(u,\ub)\}.
\end{equation}
Moreover, on the initial spacelike hypersurface $\Si_0$ we define
\begin{equation}
\QQ_{\Si_0\cap\cal{K}}=\sup_{\{u,\ub|S(u,\ub)\subseteq\cal{K}\}}
\{\QQ_{1_{\Si_0\cap V(u,\ub)}}+\QQ_{2_{\Si_0\cap V(u,\ub)}}\}.
\end{equation}

\newpage

\subsection{Kerr Spacetime}

\NI From now on we focus our attention on the Kerr spacetime. 

\NI Kerr metric in the Boyer-Lindquist coordinates $\{t,r,\theta,\phi\}$ has the following form:
\begin{eqnarray}\label{BL}
ds^2 &=&
-\frac{\De-a^2\sin^2\theta}{\Si}dt^2+\frac{\Sigma}{\Delta}dr^2+\Sigma
d\theta^2
-\frac{4Mar\sin^2\theta}{\Sigma}d\phi dt\nonumber\\
&+& R^2\sin^2\theta d\phi^2\	,
\end{eqnarray}
\NI where:
\begin{eqnarray*}
\Delta &=& r^2+a^2-2Mr\\
\Sigma &=& r^2+a^2\cos^2\theta\\
R^2 &=& \frac{1}{\Sigma}\bigl((r^2+a^2)^2-\Delta
a^2\sin^2\theta\bigr),
\end{eqnarray*}
Let us note the useful identities:
\begin{eqnarray*}
\Sigma R^2 = {(r^2+a^2)}^2 &-& \Delta a^2\sin^2\theta,\\
g_{\phi\phi}g_{tt}-g_{\phi t}^2= &-&\Delta\sin^2\theta.
\end{eqnarray*}
\NI We remember that his metric is stationary, and axisymmetric, and 
asymptotically flat. Now we define the null foliation in the Kerr case: the level
hypersurfaces of the optical functions $u,\un{u}$:
\begin{eqnarray}\label{eikonal}
C(u) &=& \{p\in \mathcal{M}|u(p)=u\}\nonumber\\
\un{C}(\un{u})&=&\{p\in\mathcal{M}|\un{u}(p)=\un{u}\}.
\end{eqnarray}
\NI As it is shown in \cite{Is-Pr}, $u,\ub$ have the following form:
\begin{eqnarray*}
u &=& t-\ro\\ 
\ub &=& t+\ro
\end{eqnarray*}
\NI where $\ro=\ro(r,\theta)$ is the radial parameter of Kerr metric defined in \cite{Is-Pr}, equation (15).\\
Now let's define the null frame associated to this double null foliation:
\begin{eqnarray}\label{e_a}
e_4 &=&
\frac{\sqrt\Delta}{R}\{\frac{1}{\Delta\sin^2\theta}[g_{\phi\phi}\p_t
-g_{t\phi}\p_\phi]+
\frac 1 \Sigma[Q\p_r+P\p_\theta]\}\nn\\
e_3&=&
\frac{\sqrt\Delta}{R}\{\frac{1}{\Delta\sin^2\theta}[g_{\phi\phi}\p_t
-g_{t\phi}\p_\phi]- \frac 1 \Sigma[Q\p_r+P\p_\theta]\}\nn\\
e_\theta &=& \frac{1}{\Sigma R}(Q\p_\theta-\Delta P \p_r)
\end{eqnarray}
\NI where:
\begin{eqnarray*}
P^2(\theta, \lambda) &=& a^2(\lambda-\sin^2\theta)\\
Q^2(r,\lambda,M) &=& {(r^2+a^2)}^2-a^2\lambda\Delta\\
K^2(r) &=& r^2+a^2,
\end{eqnarray*}
\NI and $\la$ is a function of $\theta ,r$ defined implicitly as a function that at
spatial infinity is $\sin^2\theta$ (see \cite{Is-Pr}, sections 2 and 5).

\subsection{decay of $P$ with respect to $r$ and $\theta$}

\NI First let us estimate The decay of $P$ as $r$ tends to infinity.

\NI Let us make the following ansatz $\theta_*-\theta=O(\frac{1}{r^{\a}})$. We want to find the value of $\a$. By the classical decomposition:
\[sin^2\theta_*-sin^2\theta=4cos(\frac{\theta_*+\theta}{2})sin(\frac{\theta_*-\theta}{2})cos(\frac{\theta_*-\theta}{2})sin(\frac{\theta_*-\theta}{2})\]
The quantity $sin(\frac{\theta_*-\theta}{2})$ is the only one which tend to $0$ for every $\theta\neq 0$. Since at the first order $sin(\theta_*-\theta)=\theta_*-\theta$
it follows 
\[\sqrt{sin^2\theta_*-sin^2\theta}=O\sqrt{\sin(\frac{\theta_*-\theta}{2})}=O(\frac{1}{r^{\frac{\a}{2}}})\]
\NI Then $P\rightarrow0$ as $r^{-\frac{\a}{2}}$. Then by the definition of $F$ in \cite{Is-Pr}   $F=0$ along a null geodesic therefore

\[\int_r^{\infty}\frac{dr'}{Q(r',\la)}\simeq\int_{\theta}^{\theta_*}\frac{d\theta'}{P(\theta',\la)}\]

\NI As for large $r$ 

\[\int_r^{\infty}\frac{dr'}{Q(r',\la)}=O(\frac 1 r)\]

\NI It have to be  $\a=2$, hence $\lim_{r\rightarrow 0}P=O(\frac 1 r)$ for every fixed $\theta$.

\medskip













\NI For what concern the dependence from $\theta$, by the definition of $\la=\sin^2\theta^*$ and equation 22 of \cite{Is-Pr}, for $\theta<1$:

\bea
\tan\theta^*\leq\frac{\sqrt{a^2+r^2}}{r}\tan\chi+C\sqrt{2M}\sin\theta^*r^{\frac 5 2}\leq\frac{\sqrt{a^2+r^2}}{r}\tan\chi+ C'\sin\theta^*r^{\frac 5 2}\nn
\eea

\NI Hence for $\theta$ tending to zero and fixed $r$, recalling that $\lim_{\theta\rightarrow 0}\chi=O(\theta)$  we have:

\bea
&&\theta^*(1+C'{r^{\frac 5 2}})\leq\theta(1+\frac{a}{ r} )\nn\\
&&\theta^*\leq(1+{C}{r^{\frac{ 5}{ 2}}})^{-1}(1+\frac{a }{r})\theta\leq C''\theta
\eea

\NI and from this $\lim_{\theta\rightarrow 0}P=C'''\theta\ \ \ \ $ for any fixed $r$.

\subsection{The connection coefficients of Kerr spacetime} \label{simbchrist}
 
\NI In this section we calculate how the connection coefficients depend on the null frame defined above 
in relation with their decays in $r$, we report their
expression:
\begin{prop}\label{24}
\NI The connection coefficients related to the null frame associated to the \cite{Is-Pr} null cone foliation have the following explicit form: 
\begin{eqnarray}
\chi_{\theta\theta} &=& \frac 1 r-\frac{M}{r^2}+O(r^3) \nonumber\\
\chi_{\theta\phi} &=& 0\nonumber\\
\chi_{\phi\phi} &=& \frac 1 r -\frac{M}{r^2}-\frac{P}{r^2}\cot\theta\nonumber+O(r^3)\\
\hat\chi_{\theta\theta} &=& P\cot\theta\frac 1 {r^2}+O(r^3)\nonumber\\
\hat\chi_{\phi\phi} &=& -\hat\chi_{\theta\theta}\nonumber\\
\hat\chi_{\theta\phi} &=& 0\nonumber\\
\tr\chi &=& \frac 2 r-\frac{(2M+P\cot\theta)}{r^2}+O(r^3),
\end{eqnarray}
\NI where $\hat{\chi},\hat{\chib}$ are respectively
the traceless parts of $\chi,\chib$.
\NI Then:
\begin{eqnarray}
\zeta_{\theta} &=& -\frac {1} {2r^3}(a^2\sin\theta\cos\theta+MP+P\p_{\theta}P)+O(r^4)\nonumber\\
\zeta_{\phi} &=&\frac{3Ma\sin\theta}{r^3}+O(r^4)\nonumber\\
\eta_{\theta} &=& -\frac{1}{2r^3}(a^2\sin\theta\cos\theta+MP)+O(r^4)\nonumber\\
\eta_{\phi} &=& \frac{3Ma\sin\theta}{r^3}+O(r^4)\nonumber\\
\omega &=& -\frac {M}{2r^2}+O(r^3)\nonumber\\
\un{\omega} &=& \frac {M}{2r^2}+O(r^3).
\end{eqnarray}
\NI Moreover the following relations hold:
\begin{eqnarray}
\un{\chi}_{\theta\theta}&=&-\chi_{\theta\theta}\nonumber\\
\un{\chi}_{\theta\phi}&=& \chi_{\theta\phi}=0\nonumber\\
\un{\chi}_{\phi\phi}&=& -\chi_{\phi\phi}\nonumber\\
\hat{\un{\chi}}_{\theta\theta} &=& -\chi_{\theta\theta}\nonumber\\
\hat{\un{\chi}}_{\phi\phi} &=& -\chi_{\phi\phi}\\
\un{\eta}_{\theta} &=& -\eta_{\theta}\nonumber\\
\un{\eta}_{\phi} &=& -\eta_{\phi}\nonumber\\
\om &=& -\un{\om}\nonumber
\end{eqnarray}
\end{prop}

\medskip

\NI Proof: 

\NI As already said in the introduction the calculations for these quantities are very long,
We report as example the calculus just for one of these coefficients, the
remaining ones are computed in a similar way; specifically we
compute $\chi_{\theta\theta}$ in the appendix, see section 5.1. 
\begin{remark}
\NI Notice that the term $P\cot\theta$ do not blow up for $\theta\rightarrow 0$ due to the fact that $lim_{\theta\rightarrow 0} P(\theta)=O(\theta)$
\end{remark}

\NI As direct consequence of this calculations we have the following:

\begin{prop}

\NI Denoted as $\{e_{\mu}\}$ the null frame associated to the null cones of \cite{Is-Pr} in  Kerr spacetime, the following decays for the null Riemann components hold:

\bea 
&&\a\ \ ,\ \ \underline{\a}\ \ =O(\frac {M^3}{ r^5})\nn\\
&&\b\ \ ,\ \ \underline{\b}\ \ \ =O(\frac {M^2}{ r^4})\\
&&\ro\ \ =O(\frac {M }{r^3})\nn\\ 
&&\si\ \ =O(\frac {M }{r^4})\nn
\eea

\end{prop}

\NI The proof of this proposition follows in a  long but straightforward way by the the structure equations, see \cite{Ch-Kl:book} eq. (3.1.46), the decays in $r$ of the connection coefficients calculated in proposition \ref{24} and observing that $\dddd_4$  of a connection coefficient gain a decay factor $r$ while $\dddd_3$ gain a factor $u$.

\subsection{Asymptotic behavior of the deformation tensors}

\NI Once we have computed the decays of the deformation tensors, we can provide the asymptotic behavior of the  null
components and the first derivatives of the deformation tensors
relative to $O$, the rotation vector fields and to $S, \bar{K}, T$
since they will be very useful in the error estimate.

\NI Let us define$\{^{X} i_{\mu\nu},^{X} j,^{X} m_{\mu},^{X} \underline{m}_{\mu},^{X}n,^{X} \underline{n}\}$ the null components of the deformation tensor, see [Kl-Ni], chapter 4,

\begin{prop}\label{omega}
\NI As ${^{^{(3)}(O)}}$ is a Killing vector field, the deformation
tensor of ${^{(3)}}O$ is null . The null components of the two other rotation deformation tensors:
\begin{eqnarray}
&{^{^{(2)}(O)}}i_{\theta\theta} &= O\biggl(-\frac{a}{r^2}(2a-3M)\sin\theta\cos\theta\cos\phi\biggr)\nonumber\\
&{^{^{(2)}(O)}}i_{\phi\phi} &= O\biggl(-\frac{3Ma}{r^2}\sin\theta\cos\theta\cos\phi\biggr)\nonumber\\
&{^{^{(2)}(O)}}i_{\theta\phi} &=O\biggl(-\frac{a^2}{r^2}\sin\theta\cos\theta\sin\phi\biggr)\nonumber\\
&{^{^{(2)}(O)}}j &=O\biggl(\frac{a}{r^2}(a-3M){\sin\theta\cos\theta\cos\phi}\biggr)\nonumber\\
&{^{^{(2)}(O)}}m_{\theta} &= \frac{2Ma\sin\phi(1-\sin^4\theta)}{r^2}\\
&{^{^{(2)}(O)}}m_{\phi} &= 0\nonumber\\
&{^{^{(2)}(O)}}\un{m}_{\theta} &= O\biggl(\frac{2Ma\sin\phi(1-\sin^4\theta)}{r^2}\biggr)\nonumber\\
&{^{^{(2)}(O)}}\un{m}_{\phi} &=0\nonumber\\
&{^{^{(2)}(O)}}n &= O\biggl(\frac{2a\sin\theta\cos\theta}{r^2}[a\cos\phi-M\sin\phi]\biggr) \nonumber\\
&{^{^{(2)}(O)}}\un{n} &=
O\biggl(\frac{2a\sin\theta\cos\theta}{r^2}[a\cos\phi+M\sin\phi]\biggr)\ .\nonumber
\end{eqnarray}
\NI The behavior of ${{^{(1)}(O)}}\pi_{\mu\nu}$ is very similar,
\begin{eqnarray}
&{^{^{(1)}(O)}}i_{\theta\theta} &=O\biggl(-\frac{a}{r^2}(2a-3M)\sin\theta\cos\theta\sin\phi\biggr)\nonumber\\
&{^{^{(1)}(O)}}i_{\phi\phi} &=O\biggl(-\frac{3Ma}{r^2}\sin\theta\cos\theta\sin\phi\biggr)\nonumber\\
&{^{^{(1)}(O)}}i_{\theta\phi} &=O\biggl(\frac{a^2+M^2}{r^2}\sin\theta\cos\theta\cos\phi\biggr)\nonumber\\
&{^{^{(1)}(O)}}j
&=O\biggl(\frac{a}{r^2}(a-3M){\sin\theta\cos\theta\sin\phi}\biggr)\nn
\end{eqnarray}
\begin{eqnarray}
&{^{^{(1)}(O)}}m_{\theta} &= O\biggl(-\frac{2Ma\cos\phi(1-\sin^4\theta)}{r^2}\biggr)\nonumber\\
&{^{^{(1)}(O)}}m_{\phi}& =0\nonumber\\
&{^{^{(1)}(O)}}\un{m}_{\theta} &= O\biggr(-\frac{2Ma\cos\phi(1-\sin^4\theta)}{r^2}\biggr)\nonumber\\
&{^{^{(1)}(O)}}\un{m}_{\phi} &= 0\\
&{^{^{(1)}(O)}}n &= O\biggl(\frac{2a\sin\theta\cos\theta}{r^2}[a\sin\phi+M\cos\phi]\biggr) \nonumber\\
&{^{^{(1)}(O)}}\un{n} &=
O\biggl(\frac{2a\sin\theta\cos\theta}{r^2}[a\sin\phi-M\cos\phi]\biggr)\nn.
\end{eqnarray}
\end{prop}

\NI The proof is a straightforward computation, starting from the
definition of the null components of ${^{(O^{(i)})}}\pi$, and
using the results obtained for the connection coefficients of Kerr
spacetime. We
report only one of them in the appendix, see section 5.2.
\begin{cor}\label{omega'}
\NI As far as the $O$ components are concerned, the following inequalities hold in $\mathcal{K}$, with $c$ a suitable constant:
\begin{eqnarray}
|r^3\nabb({^{(O)}}i,{^{(O)}}j,{^{(O)}}m,{^{(O)}}\un{m},{^{(O)}}n,{^{(O)}}\un{n})|&\leq&
c\nn\\
|r^3\dddd_4({^{(O)}}i,{^{(O)}}j,{^{(O)}}m,{^{(O)}}\un{m},{^{(O)}}n,{^{(O)}}\un{n})|&\leq&
c\\
|r^3\dddd_3({^{(O)}}i,{^{(O)}}j,{^{(O)}}m,{^{(O)}}\un{m},{^{(O)}}n,{^{(O)}}\un{n})|&\leq&
c\nn.
\end{eqnarray}
\end{cor}

\NI Moreover we have the following proposition\begin{prop}
The first derivatives of the components of ${^{(O)}}\pi_{\mu\nu}$ satisfy the following $L_p$ estimates on any leave $S(u,\ub)\subset\mathcal{K}$, with $p\in[2,4]$:
\begin{eqnarray}
||r^{3-\frac{2}{p}}\nabb({^{(O)}}i,{^{(O)}}j,{^{(O)}}m,{^{(O)}}\un{m},{^{(O)}}n,{^{(O)}}\un{n})||_{p,S}\!&\leq&\!
c\nn\\
||r^{3-\frac{2}{p}}\dddd_4({^{(O)}}i,{^{(O)}}j,{^{(O)}}m,{^{(O)}}\un{m},{^{(O)}}n,{^{(O)}}\un{n})||_{p,S}\!&\leq&\!
c\\
||r^{3-\frac{2}{p}}\dddd_3({^{(O)}}i,{^{(O)}}j,{^{(O)}}m,{^{(O)}}\un{m},{^{(O)}}n,{^{(O)}}\un{n})||_{p,S}\!&\leq&\!
c\nn \ .
\end{eqnarray}
\end{prop}
\begin{proof}
\NI This result and those relative to the derivatives of the null components of the other deformation tensors are obtained
observing that:\\
\NI i) When $\dddd_4$ acts on a function $f$, it improves its asymptotic behavior by a $r^{-1}$ factor.\\
\NI ii) In general, when $\dddd_3$ acts on $f$, it operates substantially
as $\frac{\p}{\p u}$ and it brings a factor  $u^{-1}$, but if $f$
doesn't depend on $t$, but only on $r$, then the derivative with
respect to $e_3$ produces again a factor of the form
$r^{-1}$.\\
\NI iii) The tangential derivative $\nabb$ on $f$,  if $f$ depends on
$\theta$ or $\phi$, gives a factor $r^{-1}$; if  $f$  doesn't
depend on the angular variables, then one gets  a factor
$O(r^{-2}) $. (This follows easily looking at the explicit
expression of the vector fields $e_\theta,e_\phi$ (see
(\ref{e_a}))).
\end{proof}
\NI Recalling the explicit expressions of the null components of
${^{(X)}}\pih_{\mu\nu}$, when $X=T,S,K_0$ given in \cite{Ch-Kl1}, (pg.172-176), we obtain the following estimates for their asymptotic decays.
\begin{prop}\label{t}
\NI Recalling the decays of the connection coefficients, we obtain the following asymptotic behavior for the components of the
deformation tensor of the vector field $T=\frac 1 2 (e_3+e_4)$:
\begin{eqnarray}
&{^{(T)}}i_{ab} &=0 \nonumber\\
&{^{(T)}}j &= 0\nonumber\\
&{^{(T)}}m_{\theta} &= O\biggl(\frac
{1}{r^3}(a^2\sin\theta\cos\theta+MP+\frac P 2 \p_{\theta}P)\biggr)\\
&{^{(T)}}m_{\phi} & =O\biggl(-\frac{6Ma\sin\theta}{r^3}\biggr)\nonumber\\
&{^{(T)}}\un{m}_{\theta} &= -{^{(T)}}m_{\theta}\nonumber\\
&{^{(T)}}\un{m}_\phi &= -{^{(T)}}m_{\phi}\nonumber\\
&{^{(T)}}n &= O\biggl(\frac{2M}{r^2}\biggr)\nonumber\\
&{^{(T)}}\un{n} &=O\biggl(-\frac{2M}{r^2}\biggr) \nonumber.
\end{eqnarray}
\NI Their derivatives satisfy the following bounds:
\begin{eqnarray}
&&||r^{4-\frac 2 p}\nabb{^{(T)}}m ||_{p,S}\leq c\nonumber\\
&&||r^{4-\frac 2 p}\nabb{^{(T)}}\un{m} ||_{p,S}\leq c\nonumber\\
&&||r^{4-\frac 2 p}\nabb{^{(T)}}n ||_{p,S}\leq c\\
&&||r^{4-\frac 2 p}\nabb{^{(T)}}\un{n} ||_{p,S}\leq c\nonumber.
\end{eqnarray}
\begin{eqnarray}
&&||r^{4-\frac 2 p}\dddd_4{^{(T)}}m ||_{p,S}\leq c\nonumber\\
&&||r^{4-\frac 2 p}\dddd_4{^{(T)}}\un{m} ||_{p,S}\leq c\nonumber\\
&&||r^{3-\frac 2 p}\dddd_4{^{(T)}}n ||_{p,S}\leq c\\
&&||r^{3-\frac 2 p}\dddd_4{^{(T)}}\un{n} ||_{p,S}\leq c\nonumber.
\end{eqnarray}
\begin{eqnarray}
&&||r^{4-\frac 2 p}\dddd_3{^{(T)}}m ||_{p,S}\leq c\nonumber\\
&&||r^{4-\frac 2 p}\dddd_3{^{(T)}}\un{m} ||_{p,S}\leq c\nonumber\\
&&||r^{3-\frac 2 p}\dddd_3{^{(T)}}n ||_{p,S}\leq c\\
&&||r^{3-\frac 2 p}\dddd_3{^{(T)}}\un{n} ||_{p,S}\leq c\nonumber,
\end{eqnarray}
\NI with $p\in[2,4]$ for any $S\subset\mathcal{K}$.
\end{prop}
\begin{prop}\label{S}
\NI The null components of the
deformation tensor of the vector field  ${^{(S)}}\hat{\pi}_{\mu\nu}$ decay asymptotically in the following way:
\begin{eqnarray}
&{^{(S)}} i_{\theta\theta} &= 2M\frac{\log r}{r}+\frac 5 2 \frac{\p_{\theta}P-P\cot\theta-M}{r}\nonumber\\
&{^{(S)}} i_{\phi\phi} &=\frac {2M\log
r}{r}+\frac{3(P\cot\theta-\p_\theta P)-7M}
{2r}\\
&{^{(S)}} j &= 4M\frac{\log r}{r}-\frac{5M}{r}\nonumber\\
&{^{(S)}} m_{\theta} &= \frac u 2 \frac{2(a^2\sin\theta\cos\theta+MP)+P\p_\theta P}{r^3}\nonumber\\
&{^{(S)}} m_{\phi} &= -6u\frac{Ma\sin\theta}{r^3}\nonumber
\end{eqnarray}
\begin{eqnarray}
&{^{(S)}} \un{m}_{\theta} &= -\frac {\un{u}} {2} \frac{2(a^2\sin\theta\cos\theta+MP)+P\p_\theta P}{r^3}\nonumber\\
&=& -\frac{2(a^2\sin\theta\cos\theta+MP)+P\p_\theta P}{r^2}\nonumber\\
&{^{(S)}} \un{m}_{\phi} &= -\frac {\un{u}} {2} \frac{2(a^2\sin\theta\cos\theta+MP)+P\p_\theta P}{r^3}\\
&=& -\frac{2(a^2\sin\theta\cos\theta+MP)+P\p_\theta P}{r^2}\nonumber\\
&{^{(S)}} n &= 2\frac{M u}{r^2}\nonumber\\
&{^{(S)}} \un{n} &= -2\frac{M\un{u}}{r^2}=-4\frac M r\nonumber,
\end{eqnarray}
\NI where in the last equality we used $\ \frac{\tau_+}{r^2} =O(r^{-1})$ .
Moreover, for their first derivatives, the following $L_p$
estimates hold for any $p\in[2,4]$ and for any
$S\subset\mathcal{K}$:
\begin{eqnarray}
&&||r^{3-\frac 2 p}\nabb{^{(S)}}i ||_{p,S}\leq c\nonumber\\
&&||r^{3-\frac 2 p}\nabb{^{(S)}}j ||_{p,S}\leq c\nonumber\\
&&||\frac{r^{4-\frac 2 p}}{\tau_-}\nabla{^{(S)}}m ||_{p,S}\leq c\nonumber\\
&&||r^{3-\frac 2 p}\nabb{^{(S)}}\un{m} ||_{p,S}\leq c\nonumber\\
&&||r^{4-\frac 2 p}\frac{1}{\tau_-}\nabb{^{(S)}}n ||_{p,S}\leq c\\
&&||r^{3-\frac 2 p}\nabb{^{(S)}}\un{n} ||_{p,S}\leq c\nonumber.
\end{eqnarray}
\begin{eqnarray}
&&||r^{2-\frac 2 p}\dddd_4{^{(S)}}i ||_{p,S}\leq c\nonumber\\
&&||r^{2-\frac 2 p}\dddd_4{^{(S)}}j ||_{p,S}\leq c\nonumber\\
&&||r^{4-\frac 2 p}\frac{1}{\tau_-}\dddd_4{^{(S)}}m ||_{p,S}\leq c\nonumber\\
&&||r^{3-\frac 2 p}\dddd_4{^{(S)}}\un{m} ||_{p,S}\leq c\nonumber\\
&&||r^{3-\frac 2 p}\frac{1}{\tau_-}\dddd_4{^{(S)}}n ||_{p,S}\leq c\\
&&||r^{2-\frac 2 p}\dddd_4{^{(S)}}\un{n} ||_{p,S}\leq c\nonumber.
\end{eqnarray}
\begin{eqnarray}
&&||r^{2-\frac 2 p}\dddd_3{^{(S)}}i ||_{p,S}\leq c\nonumber\\
&&||r^{2-\frac 2 p}\dddd_3{^{(S)}}j ||_{p,S}\leq c\nonumber\\
&&||r^{3-\frac 2 p}\dddd_3{^{(S)}}m ||_{p,S}\leq c\nonumber\\
&&||r^{3-\frac 2 p}\dddd_3{^{(S)}}\un{m} ||_{p,S}\leq c\nonumber\\
&&||r^{2-\frac 2 p}\dddd_3{^{(S)}}n ||_{p,S}\leq c\\
&&||r^{2-\frac 2 p}\dddd_3{^{(S)}}\un{n} ||_{p,S}\leq c\nonumber.
\end{eqnarray}
\end{prop}
\NI Finally we  compute the components of the $\bar{K}$ deformation tensor:
\begin{prop}
For any $S\subset\mathcal{K}$, the following estimates hold
\begin{eqnarray}
&{^{(\bar{K})}}i_{\theta\theta}&= \frac{4Mt\log r}{r}\nonumber\\
&{^{(\bar{K})}}i_{\phi\phi} &=\frac{4Mt\log r}{r}\nonumber\\
&{^{(\bar{K})}}j &=8Mt\frac{\log r}{r}\nonumber\\
&{^{(\bar{K})}}m_{\theta} &=\tau_-^2\frac 1 2 \frac{2(a^2\sin\theta\cos\theta+MP)+P\p_\theta P}{r^3}\nonumber\\
&{^{(\bar{K})}}m_{\phi} &=\tau_-^2(-6\frac{Ma\sin\theta}{r^3})\nonumber\\
&{^{(\bar{K})}}\un{m}_{\theta} &=\tau_+^2(-\frac 1 2 \frac{2(a^2\sin\theta\cos\theta+MP)+P\p_\theta P}{r^3})\nonumber\\
&=& -2\frac{2(a^2\sin\theta\cos\theta+MP)+P\p_\theta P}{r}\\
&{^{(\bar{K})}}\un{m}_{\phi} &=\tau_+^2(-\frac 1 2 \frac{2(a^2\sin\theta\cos\theta+MP)+P\p_\theta P}{r^3})\nonumber\\
&=& -2\frac{2(a^2\sin\theta\cos\theta+MP)+P\p\theta P}{r}\nonumber\\
&{^{(\bar{K})}}n &= 2\frac{M\tau_-^2}{r^2}\nonumber\\
&{^{(\bar{K})}}\un{n} &= -2\frac{M\tau_+^2}{r^2} = -8M\nonumber.
\end{eqnarray}
\NI Moreover, for every $p\in[2,4]$ the following inequalities hold:
\begin{eqnarray}
&&||\frac{r^{3-\frac 2 p}}{t}\nabb{^{(\bar{K})}}i ||_{p,S}\leq c\nonumber\\
&&||\frac{r^{3-\frac 2 p}}{t}\nabb{^{(\bar{K})}}j ||_{p,S}\leq c\nonumber\\
&&||r^{4-\frac 2 p}\frac{1}{\tau_-^2}\nabb{^{(\bar{K})}}m ||_{p,S}\leq c\nonumber\\
&&||r^{2-\frac 2 p}\nabb{^{(\bar{K})}}\un{m} ||_{p,S}\leq c\nonumber\\
&&||r^{4-\frac 2 p}\frac{1}{\tau_-^2}\nabb{^{(\bar{K})}}n ||_{p,S}\leq c\\
&&||r^{4-\frac 2 p}\nabb{^{(\bar{K})}}\un{n} ||_{p,S}\leq
c\nonumber.
\end{eqnarray}
\end{prop}

\section {Estimate of the Error Term}
\NI This section is crucial in the proof of the theorem (\ref{final version}), as from the boundedness of the modified energy norms defined in (\ref{norms1}),(\ref{norms2}), we can know the asymptotic behavior of a prescribed Weyl tensor $\tW$. The calculations used to prove these norms are bounded are quite complicared, since many quantities are involved. Hereafter we  report the boundedness proof of some of them into details and we state the propositions about the remaining terms .
  
\NI In order to prove the boundedness of the $\cal{Q}$ norms, we need to control a quantity  we call the error term $\mathcal{E}$ , defined in the following
way:
\begin{eqnarray}
\mathcal{E}(u,\un{u})&\equiv & (\QQ+\QQb)(u,\ub)-\QQ_{\Si_0\cap
V(u,\ub)}.
\end{eqnarray}
\begin{remark}
\NI From now on the Weyl tensor field $\tW$ from which we have defined the $\QQ$ norms will be the tensor field $\lie_{T_0}W$. As $T_0$ is a Killing vector field, then $\tW$ is a Weyl field and it satisfies the massless spin 2 equations.
\end{remark}
\NI First we define the 1-form
$$
\tilde{P}_\mu =\ttau^{5+\e}P_\mu\	.
$$
\NI Recalling the definition of the 1-form $P$ related to the
Bel-Robinson tensor of a Weyl tensor $W$ (see (\ref{P})) and by Stokes
theorem it is easy to prove that:
\begin{eqnarray}
Div\tilde{P} &=& Div\tilde {Q}_{\b\ga\de}X^\b Y^\ga Z^\de+Div(\tau_-^{5+\epsilon})\tilde{Q}_{\b\ga\de}X^\b Y^\ga Z^\de\\
&+& \frac 1 2 \tilde{Q}_{\a\b\ga\de}\bigl({^{(X)}}\pi^{\a\b}Y^\ga Z^\de+{^{(Y)}}\pi^{\a\ga}X^\b Z^\de
+{^{(Z)}}\pi^{\a\de}X^\b Y^\ga
\bigr).\nonumber
\end{eqnarray}
\NI The term $Div(\tau_-^{5+\epsilon}\tilde{Q})_{\b\ga\de}X^\b Y^\ga Z^\de$ is not a
problem for the estimate of the error because it is negative,
then by Stokes  theorem, it follows that:

\begin{eqnarray}
&&\int_{\un{C}(\ub)\cap V(u,\ub)}\ttau^{5+\e}\tilde{Q}(W)(X,Y,Z,e_3)+\int_{C(u)\cap V(u,\ub)}
\ttau^{5+\e}\tilde{Q}(W)(X,Y,Z,e_4)\nonumber\\
&&-\int_{\Si_0\cap V(u,\ub)}\ttau^{5+\e}\tilde{Q}(W)(X,Y,Z,T)\nonumber\\
&&=\int_{V(u,\ub)}\ttau^{5+\e}[Div \tilde{Q}(W)_{\b\ga\de}X^\b Y^\ga Z^\de +\frac 1 2
\tilde{Q}^{\a\b\ga\de}(W)({^{(X)}}\pi_{\a\b}Y_\ga Z_\de\nonumber\\
&&\ \ +{^{(Y)}}\pi_{\a\b}Y_\ga Z_\de+{^{(Z)}}\pi_{\a\b}X_\ga Y_\de)]\nn\\
&&\ \ -(5+\e)\int_{V(u,\ub)}(\frac R {\sqrt{\De}})|\ttau|^{4+\e}\tilde{Q}(W)(X, Y,Z,e_4) .
\end{eqnarray}
\NI The last term, being negative, can be ignored and we have only to consider $\ttau^{5+\e}Div P$.\\

\NI Therefore we can decompose the error term into two parts, one of
it related to  the $\QQ_1$ norms and the other one associated to the
$\QQ_2$ norms:
\begin{eqnarray*}
\mathcal{E}(u,\un{u}) &\equiv &\EEb(u,\ub)+\EEbb(u,\ub)
\end{eqnarray*}
\NI where
\begin{eqnarray}\label{EEb}
\EEb(u,\ub)&=&\int_{V(u,\ub)} \ttau^{5+\e}Div \tilde{Q}(\lie_O
\tW)_{\b\ga\de}(\bar{K}^\b\bar{K}^\ga T^\de)\nonumber\\
& &+\frac 3 2 \int_{V(u,\ub)} \ttau^{5+\e}\tilde{Q}(\lie_T
\tW)_{\a\b\ga\de}({^{(\bar{K})}}\pi^{\a\b}\bar{K}^\ga\bar{K}^\de)\\
& &+ \int_{V(u,\ub)} \ttau^{5+\e}\tilde{Q}(\lie_O
\tW)_{\a\b\ga\de}({^{(\bar{K})}}\pi^{\a\b}\bar{K}^\ga T^\de)\nn,
\end{eqnarray}
\begin{eqnarray}
\EEbb(u,\ub)&=&\int_{V(u,\ub)}  \ttau^{5+\e}Div \tilde{Q}(\lie_O\lie_T
\tW)_{\b\ga\de}(\bar{K}^\b\bar{K}^\ga\bar{K}^\de)\nonumber\\
& &+\int_{V(u,\ub)}  \ttau^{5+\e}Div \tilde{Q}(\lie_O^2
\tW)_{\b\ga\de}(\bar{K}^\b\bar{K}^\ga T^\de)\nonumber\\
& &+\int_{V(u,\ub)}  \ttau^{5+\e}Div \tilde{Q}(\lie_S\lie_T
\tW)_{\b\ga\de}(\bar{K}^\b\bar{K}^\ga\bar{K}^\de)\nonumber\\
& &+\frac 3 2 \int_{V(u,\ub)} \ttau^{5+\e}\tilde{Q}(\lie_O\lie_T
\tW)_{\a\b\ga\de}({^{(\bar{K})}}\pi^{\a\b}\bar{K}^\ga\bar{K}^\de)\nonumber\\
& &+\frac 3 2 \int_{V(u,\ub)} \ttau^{5+\e}\tilde{Q}(\lie_S\lie_T
\tW)_{\a\b\ga\de}({^{(\bar{K})}}\pi^{\a\b}\bar{K}^\ga\bar{K}^\de)\\
& &+ \int_{V(u,\ub)} \ttau^{5+\e}\tilde{Q}(\lie_O^2
\tW)_{\a\b\ga\de}({^{(\bar{K})}}\pi^{\a\b}\bar{K}^\ga T^\de)\nonumber.
\end{eqnarray}
\NI This decomposition separates  the terms depending only on the first
derivatives of $\tW$, which appear in $\EEb$ from the terms that
involve second derivatives, included in $\EEbb$. We shall prove
there exists a constant $c_0$ such that
\begin{equation}\label{error}
\mathcal{E}(u,\ub)\leq \frac {c_0} {r_0} \QQ_{\cal{K}},
\end{equation}
\NI which implies
\begin{equation}\label{smallness}
\QQ_{\cal{K}}\leq \frac 1 {1-c_0/r_0}\QQ_{\Si_0\cap\cal{K}}\ .
\end{equation}
\NI Therefore, for $r\geq r_0$ sufficiently great, the proof
of the theorem follows. This last observation means that we have to
consider only the outer region with respect to the domain of dependance of $B(0,r)$, the ball of radius $r$, for a suitable $r$, contained in $\Si_0$.

\subsection{Preliminary estimates for the error }
\NI From now on, we shall follow the procedure used in chapter 6 of \cite{Kl-Ni1}\\
Let us consider:
$$
J(X,W)_{\b\ga\de} \equiv D^\a(\lie_X W)_{\a\b\ga\de}\	,
$$

\NI we define the null components of the Weyl current in the following way
\begin{eqnarray}
\Lambda(J) &=& \frac 1 4 J_{434},\qquad \un{\Lambda}(J) = \frac 1
4 J_{343},\qquad \Xi(J)_a =\frac 1 2 J_{44a}\nn\\
\un{\Xi}(J)_a &=& \frac 1 2 J_{33a},\qquad I(J)_a = \frac 1 2
J_{34a},\qquad \un{I}(J)_a = \frac 1 2 J_{43a}\nn\\
K(J) &=& \frac 1 4  \e^{ab}J_{4ab},\qquad \un{K}(J) = \frac 1 4
\e^{ab}J_{3ab}\\
\Th(J)_{ab} &=& J_{a4b}+J_{b4a}-(\de^{cd}J_{c4d})\de_{ab},\quad
\un{\Th}(J)_{ab} = J_{a3b}+J_{b3a}-(\de^{cd}J_{c3d})\de_{ab}\		.\nn
\end{eqnarray}
\NI Now, let us decompose $\div
Q(\lie_X \tW)$ along the null frame:
\begin{eqnarray}\label{D(T,W)1}
D(X,\tW)(\bar{K},\bar{K},\bar{K}) &=& \frac 1 8 \tau_+^6
D(X,\tW)_{444}+\frac 3 8 \tau_+^4\tau_-^2 D(X,\tW)_{344}\nn\\
& &\frac 3 8 \tau_+^2\tau_-^4 D(X,\tW)_{334}+\frac 1 8 \tau_-^6
D(X,\tW)_{333},\ \ \ \ \ 
\end{eqnarray}
\NI where
\begin{eqnarray}\label{D(T,W)2}
D(X,\tW)_{444} &=& 4\a(\lie_X \tW)\c\Th(T,\tW)-8\b(\lie_X
\tW)\c\Xi(X,\tW)\nn\\
D(X,\tW)_{443} &=& 8\ro(\lie_X \tW)\c\Lambda(X,\tW)+8\si(\lie_X
\tW)K(\lie_X \tW)\nn\\
& &+8\b(\lie_X \tW)\c I(T,\tW)\nn\\
D(X,\tW)_{334} &=& 8\ro(\lie_X \tW)\un{\Lambda}(X,\tW)-8\si(\lie_X
W)K(X,\tW)\\
& &-8\bb(\lie_X \tW)\c\un{I}(X,\tW)\nn\\
D(X,W)_{333} &=& 4\aa(\lie_X \tW)\c\un{\Theta}(X,\tW)+8\bb(\lie_X
\tW)\c\un{\Xi}(X,\tW)\nn
\end{eqnarray}
\NI Now,  we
decompose the null current $J(X,\tW)$ into three parts, see [12] , (6.1.6.) pag. 245 and sgg. for more details about this decomposition.
$$
J(X,\tW)=J^0(X,\tW)+J^1(X,\tW)+J^2(X,\tW)+J^3(X,\tW)
$$
\NI where $J^1(X,\tW)=h$ has been discussed in the introduction and
\begin{eqnarray}
J^1(X,\tW)_{\b\ga\de} &=& \frac 1 2 {^{(X)}}\pih^{\mu\nu}D_\nu
\tW_{\mu\b\ga\de}
\end{eqnarray}
\NI is such that it depends on ${^{(X)}}\pih$ and on the derivatives
of null components of Weyl tensor field up to the first order.
Then,
\begin{eqnarray}
J^2(X,\tW)_{\b\ga\de} &=& \frac 1 2
{^{(X)}}p_\la{\tW^\la}_{\b\ga\de}\nn\\
J^3(X,\tW)_{\b\ga\de} &=& \frac 1 2
({^{(X)}}q_{\a\b\la}{\tW^{\a\la}}_{\ga\de}+{^{(X)}}q_{\a\ga\la}{{\tW^\a\la}_{\b\ga}})
\end{eqnarray}
\NI where
\begin{eqnarray*}
&{^{(X)}p_\lambda}& =D^\alpha{^{(X)}}\hat{\pi}_{\alpha\lambda}\\
&{^{(X)}}q_{\alpha\beta\gamma}&=
D^\beta{^{(X)}\hat{\pi}}_{\alpha\gamma}-D^\gamma{^{(X)}\hat{\pi}}_{\alpha\beta}-\frac
 1 3({^{(X)}p_\gamma}g_{\alpha\beta}-{^{(X)}p_\beta}g_{\alpha\gamma}).
\end{eqnarray*}
\NI Let's write the explicit expressions for ${^{(X)}}p_\mu$;
they are the following
\begin{eqnarray}
{^{(X)}}p_3 &=& \divv{^{(X)}}\un{m}-\frac 1
2(\dddd_4{^{(X)}}\un{n}+\dddd_3{^{(X)}}j)+(2\un{\eta}+\eta-\zeta)\cdot
{^{(X)}}\un{m}\\
&-&\hat{\chi}\cdot{^{(X)}}i-\frac 1 2 \tr\chi(
\tr{^{(X)}}i+{^{(X)}}j)-\frac 1 2 \tr\un{\chi}{^{(X)}}n-(\dd_3\log
\Omega){^{(X)}}n\nonumber,
\end{eqnarray}
\begin{eqnarray}
{^{(X)}}p_4 &=& \divv{^{(X)}}m-\frac 1
2(\dddd_3{^{(X)}}n+\dddd_4{^{(X)}}j)+(\un{\eta}+2\eta+\zeta)\cdot
{^{(X)}}m\\
&-&\hat{\un{\chi}}\cdot{^{(X)}}i-\frac 1 2 \tr\un{\chi}(
\tr{^{(X)}}i+{^{(X)}}j)-\frac 1 2 \tr\chi{^{(X)}}\un{n}-(\dd_4\log
\Omega){^{(X)}}\un{n}\nonumber,
\end{eqnarray}
\begin{eqnarray}
{^{(X)}}\pp &=& \nabb_c{^{(X)}}i-\frac 1
2(\dddd_4{^{(X)}}\un{m}+\dddd_3{^{(X)}}m)+\frac 1 2
(\eta+\un{\eta}){^{(X)}}j\nonumber\\
&+&(\eta+\un{\eta})\cdot{^{(X)}}i-\frac 1 2
\hat{\chi}\cdot{^{(X)}}m-\frac 1 2
\hat{\un{\chi}}\cdot{^{(X)}}m-\frac 3 4
\tr\chi{^{(X)}}\un{m}-\frac
3 4 \tr\un{\chi}{^{(X)}}m\nonumber\\
&-&\frac 1 2 (\dd_4\log \Omega){^{(X)}}\un{m}
-\frac{1}{2}(\dd_3\log\Omega){^{(X)}}m.
\end{eqnarray}
\NI Once introduced this decomposition, we can decompose all null components of Weyl current in three parts. For their explicit expression, see \cite{Kl-Ni2} pg. 246-249.\\
In order to estimate the error term $\EEb$, we need to know their decays when
$X=T,O$.
\NI Since we know the asymptotic
behavior of connection coefficients of Kerr spacetime and that one
of ${^{(i)}} O $ deformation tensor components, we are
able to show the following asymptotic behaviors hold true:
\begin{prop}\label{normesuS}
\NI Based on proposition (\ref{omega}) and on corollary
(\ref{omega'}), the following estimates hold for any $S\subset
\mathcal{K}$ with $p\in [2,4]$:
\begin{equation}
||r^{3-\frac 2 p}({^{(O)}}p_3,{^{(O)}}p_4,{^{(O)}}\pp_a)||_{p,S}
\leq c.
\end{equation}
\end{prop}
\begin{prop}\label{Tp}
\NI Based on proposition (\ref{t}), the following estimates relative to ${^{(T)}}p_\lambda$ and relative to its derivatives
for any $S \subset \mathcal{K}$ with $p\in[2,4]$:
\begin{eqnarray*}
&||r^{3-\frac 2 p}{^{(T)}}p_3||_{p,S}&\leq c\\
&||r^{3-\frac 2 p}{^{(T)}}p_4||_{p,S}&\leq c\\
&||r^{3-\frac 2 p}{^{(T)}}\pp_a||_{p,S}&\leq c,
\end{eqnarray*}
\begin{eqnarray*}
&||r^{4-\frac 2 p}\nabb{^{(T)}}p_3||_{p,S}&\leq c\\
&||r^{4-\frac 2 p}\nabb{^{(T)}}p_4||_{p,S}&\leq c\\
&||r^{4-\frac 2 p}{^{(T)}}\pp_a||_{p,S}&\leq c.
\end{eqnarray*}
\end{prop}
\begin{prop}a straightforward  calculation imply that the following estimates hold, with $c$ a suitable constant:
\begin{eqnarray*}
&||r^{\frac 3 2 }{^{(T)}}p_3||_{L_2(C\cap \mathcal{K})}&\leq c\\
&||r^{\frac 5 2 }\nabb{^{(T)}}p_3||_{L_2(C\cap \mathcal{K})}&\leq c\\
&||r^{\frac 3 2 }\mathcal{L}_S{^{(T)}}\pp_a||_{L_2(C\cap
\mathcal{K})}&\leq c.
\end{eqnarray*}
\end{prop}
\NI Given a vector field $X$, we use the expressions for
any null components of the currents of $W$ relative to $X$ introduced in \cite{Kl-Ni2} (chapter 6) 
In order to estimate the error term, it
will be necessary to estimate their asymptotic behavior
when $X=T,O$. \\
\NI Let us state now the following propositions which prescribes their asymptotic decays
\begin{prop}\label{(O)q}
\NI Given the Weyl field $\tW$ propagating in the Kerr spacetime, the null
components of the part of the current $J^3(O,W)$ satisfy the
following estimates for any $S\subset \cal{K}$, with $p\in[2,4]$:
\begin{eqnarray*}
|r^{3-\frac 2 p}\Xi(O,\tW)|_{p,S}&\leq & c\\
|r^{3-\frac 2 p}\Theta(O,\tW)|_{p,S} & \leq & c\\
|r^{3-\frac 2 p}\Lambda(O,\tW)|_{p,S} &\leq & c\\
|r^{3-\frac 2 p}K(O,\tW)|_{p,S} &\leq & c\\
|r^{3-\frac 2 p}I(O,\tW)|_{p,S} &\leq & c,
\end{eqnarray*}
\NI and
\begin{eqnarray*}
|r^{3-\frac 2 p}\un{\Xi}(O,\tW)|_{p,S}&\leq & c\\
|r^{3-\frac 2 p}\un{\Theta}(O,\tW)|_{p,S} & \leq & c\\
|r^{3-\frac 2 p}\un{\Lambda}(O,\tW)|_{p,S} &\leq & c\\
|r^{3-\frac 2 p}\un{K}(O,\tW)|_{p,S} &\leq & c\\
|r^{3-\frac 2 p}\un{I}(O,\tW)|_{p,S} &\leq & c.
\end{eqnarray*}
\end{prop}
\NI The proof is also in this case a straightforward calculation.
\begin{prop}\label{(T)q}
Let $T=\frac 1 2 (e_3+e_4)$, then the null components of $J^3(T,\tW)$ satisfy the
following estimates in Kerr spacetime, for any $S\subset \cal{K}$, with $p\in[2,4]$, as $\tW$ is a Weyl field:
\begin{eqnarray*}
|r^{4-\frac 2 p}\Xi(T,\tW)|_{p,S}&\leq & c\\
|r^{4-\frac 2 p}\Theta(T,\tW)|_{p,S} & \leq & c\\
|r^{3-\frac 2 p}\Lambda(T,\tW)|_{p,S} &\leq & c\\
|r^{4-\frac 2 p}K(T,\tW)|_{p,S} &\leq & c\\
|r^{3-\frac 2 p}I(T,\tW)|_{p,S} &\leq & c,
\end{eqnarray*}
\NI and
\begin{eqnarray*}
|r^{4-\frac 2 p}\un{\Xi}(T,\tW)|_{p,S}&\leq & c\\
|r^{4-\frac 2 p}\un{\Theta}(T,\tW)|_{p,S} & \leq & c\\
|r^{3-\frac 2 p}\un{\Lambda}(T,\tW)|_{p,S} &\leq & c\\
|r^{4-\frac 2 p}\un{K}(T,\tW)|_{p,S} &\leq & c\\
|r^{3-\frac 2 p}\un{I}(T,\tW)|_{p,S} &\leq & c.
\end{eqnarray*}
\end{prop}
\begin{prop}
\NI The Lie coefficients of the vector field ${^{(2)}}O$ have the
following asymptotic behavior:
\begin{eqnarray*}
&& {^{^{(2)}(O)}}P_\theta = O(\sin\phi\cos\theta)\\
&& {^{^{(2)}(O)}}P_\phi = O(\cos\phi)\\
&& {^{^{(2)}(O)}}\un{P}_\theta = O(-\sin\phi\cos\theta)\\
&& {^{^{(2)}(O)}}\un{P}_\phi = O(-\cos\phi)\\
&& {^{^{(2)}(O)}}Q_a = O(-\sin\phi\cos\theta)\\
&& {^{^{(2)}(O)}}\un{Q}_a = O(-\sin\phi\cos\theta)\\
&& {^{^{(2)}(O)}}M = O(\frac c {r^2})\\
&& {^{^{(2)}(O)}}\un{M} = O(\frac c {r^2})\\
&& {^{^{(2)}(O)}}N = O(\frac c {r^2})\\
&& {^{^{(2)}(O)}}\un{N} = O(\frac c {r^2}).
\end{eqnarray*}
\NI Moreover, the following relations hold:
\begin{eqnarray*}
&& {^{^{(2)}(O)}}P_\theta+{^{^{(2)}(O)}}Q_\theta = O(\frac 1 r)\\
&& {^{^{(2)}(O)}}P_\phi+{^{^{(2)}(O)}}Q_\phi = O(\frac{\p
P}{\p\theta}\frac 1 r \cos\phi)\  .
\end{eqnarray*}
\NI Then at the highest order the null components of $\lie_{^{(2)}O} \tW$ behave as the projection onto $S(u,\ub)$ of the modified Lie
derivative with respect to ${^{(2)}}O$ of the corresponding null components of $\tW$.
\end{prop}
\begin{remark}
\NI The decays for the Lie coefficients of ${^{^{(1)}(O)}}\pi$ are
very similar, one has only to change $\sin\phi$ with $\cos\phi$
and $\sin\theta$ with $-\cos\theta$.
\end{remark}
\begin{prop}\label{(T)q}
\NI Let $T=\frac 1 2 (e_3+e_4)$, then the null components of $J^3(T,\tW)$ satisfy the
following estimates for any $S\subset \cal{K}$, with $p\in[2,4]$:
\begin{eqnarray*}
|r^{4-\frac 2 p}\Xi(T,\tW)|_{p,S}&\leq & c\\
|r^{4-\frac 2 p}\Theta(T,\tW)|_{p,S} & \leq & c\\
|r^{3-\frac 2 p}\Lambda(T,\tW)|_{p,S} &\leq & c\\
|r^{4-\frac 2 p}K(T,\tW)|_{p,S} &\leq & c\\
|r^{3-\frac 2 p}I(T,\tW)|_{p,S} &\leq & c,
\end{eqnarray*}
\NI and
\begin{eqnarray*}
|r^{4-\frac 2 p}\un{\Xi}(T,\tW)|_{p,S}&\leq & c\\
|r^{4-\frac 2 p}\un{\Theta}(T,\tW)|_{p,S} & \leq & c\\
|r^{3-\frac 2 p}\un{\Lambda}(T,\tW)|_{p,S} &\leq & c\\
|r^{4-\frac 2 p}\un{K}(T,\tW)|_{p,S} &\leq & c\\
|r^{3-\frac 2 p}\un{I}(T,\tW)|_{p,S} &\leq & c.
\end{eqnarray*}
\end{prop}
\NI The final step to relate the $L_2$ norms of the $\tW$ null
components to the $\QQ$ norms (which are expressed in terms of the
null components of a tensor field $\lie_X \tW$, for some suitable
$X$) is to express the covariant derivatives with respect to the
vector fields  $S$ or $T$ in terms of modified Lie derivatives of
the null components on the tensor $\tW$. These relations are the
following:
\begin{eqnarray}\label{e2}
\lieb_T\a_{ab}&=&
\dddd_T\a_{ab}+\de_{ab}\a\c(\chih+\chibh)+\bigl((\tr\chi+\tr\chib)+(\om+\un{\om})\bigr)\a_{ab}\nn\\
&=&\dddd_T\a_{ab},
\end{eqnarray}
\NI being $\chih+\chibh$, $tr\chi+\tr\chib$ and $\om+\un{\om}$ equal
to zero in the Kerr spacetime.

\NI We can now estimate the various termsof ${\mathcal{E}_1}$, let us estimate for example $Div Q(\hat{\mathcal{L}}_T \tilde{W})_{\beta\gamma\delta}
(\overline{K}^\beta,\overline{K}^\gamma,\overline{K}^\delta)$

\subsubsection{Estimate of $\int_{V(u,v)}\tau_-^{5+\e}\mbox{{ Div}}Q(\hat{\mathcal{L}}_T \tilde{W})_{\beta\gamma\delta}
(\overline{K}^\beta,\overline{K}^\gamma,\overline{K}^\delta)$}\label{section1}
Let's pose $T=X$ in (\ref{D(T,W)1}) and in (\ref{D(T,W)2}) and consider the second product term of $D(T,\tW)_{444}$:
\begin{equation}\label{e1}
4\b(\lie_T \tW)\c\Xi(T,\tW)\		.
\end{equation}
\NI Recalling the decomposition of the null components of Weyl currents, it follows
$$
\Xi(J(T,\tW))=\Xi(J^1(T,\tW))+\Xi(J^2(T,\tW))+\Xi(J^3(T,\tW))
$$
\NI where $\Xi(J^1)$ is a sum of quadratic expressions between
components of ${^{(T)}}\pih$ along the null frame and null
components of Weyl tensor or its first derivatives (plus lower
order terms), explicitly (recalling that ${^{(T)}}i={^{(T)}}j=0$):
\begin{eqnarray}
\Xi(J^1(T,\tW)) &=&
Qr[{^{(T)}}\un{m};\a_4]+Qr[{^{(T)}}m;\a_3]+Qr[{^{(T)}}m;\nabb\b]\nn\\
&+&Qr[{^{(T)}}n;\b_3]+\tr\chi\bigl\{Qr[{^{(T)}}\un{m};\a]+Qr[{^{(T)}}{m};(\ro,\si)]\bigr\}\nn\\
&+& 
\tr\un{\chi}\bigl\{Qr[{^{(T)}}m;\a]+Qr[{^{(T)}}n;\b]\bigr\}+l.o.t.\nn
\end{eqnarray}
\NI For the explicit expressions of
$\a_3,(\ro,\si)_{\{3,4\}},\b_{\{3,4\}}$ see definition (\ref{bianchi}),
while as far as the quantity  denoted  $\a_4$ is concerned,
(recalling that it does not exist an evolution equation for $\a$
along null outgoing hypersurface, as well as there is not the
evolution equation of $\aa$ along the incoming cones) it is obtained
expressing it in terms of $\a_3$ and $\dddd_T\a$:
\begin{equation}\label{alfa4}
\a_4 = 2\dddd_T\a+\a_3+(\frac 5 2 \tr\chi+\frac 1 2
\tr\un{\chi})\a
\end{equation}
\NI and analogously
$$
\aa_3=2\dddd_T\aa-\aa_4+(\frac 5 2 \tr\chib+\frac 1 2
\tr\chi)\aa.
$$
\NI Moreover,
\begin{eqnarray}
\Xi(J^2(T,\tW)) &=&
Qr[{^{(T)}}\pp;\aa]+Qr[{^{(T)}}p_3;\bb]\\
\Xi(J^3(T,\tW)) &=&
Qr[\a;(I,\un{I})({^{(T)}}q)]+Qr[\b;(K,\Lambda,\Th)({^{(T)}}q)]\nn\\
&+& Qr [(\ro,\si);\Xi({^{(T)}}q)]\label{J_3}.
\end{eqnarray}
\NI Then, recalling  (\ref{D(T,W)1})), it follows we have to estimate the following integrals:
\begin{eqnarray*}
&& \int_{V(u,\ub)}\tttau^6\ttau^{5+\e}D(T,\tW)_{444}, \quad
\int_{V(u,\ub)}\tttau^4\ttau^2\ttau^{5+\e}D(T,\tW)_{344}\\
&&\int_{V(u,\ub)}\tttau^2\ttau^4\ttau^{5+\e}D(T,\tW)_{334}, \quad
\int_{V(u,\ub)}\ttau^6\ttau^{5+\e}D(T,\tW)_{333}\ .
\end{eqnarray*}
\NI Let us control only the first integral, which has the highest
weight factor in $\tttau$. From equation (\ref{D(T,W)2}), we have
to control the following integrals:
\begin{eqnarray*}
&& \int_{V(u,\ub)}\tttau^6\ttau^{5+\e}\b(\lie_T \tW)\c\Xi(T,\tW)\\
&& \int_{V(u,\ub)}\tttau^6\ttau^{5+\e}\a(\lie_T \tW)\c\Theta(T,\tW).
\end{eqnarray*}
\NI In fact the following does hold:
\begin{prop}
In Kerr spacetime, the following inequalities hold
\begin{eqnarray*}
\bigl|\int_{V(u,\ub)}\tttau^6\ttau^{5+\e}\b(\lie_T \tW)\c\Xi(T,\tW)\bigr| &\leq &
\frac c {r_0}\QQ_{\cal{K}}\\
\bigl|\int_{V(u,\ub)}\tttau^6\ttau^{5+\e}\a(\lie_T \tW)\c\Theta(T,\tW)\bigr| &\leq
& \frac c {r_0}\QQ_{\cal{K}}.
\end{eqnarray*}
\end{prop}
\begin{proof}
\NI We discuss into details the first integral, the estimate of the
second one is similar. Using the coarea formula
$$
\int_{V(u,\ub)} F=\int_{u_0}^u du'\int_{C(u')\cap V(u,\ub)}F,
$$
\NI with $u_0(\ub)=u|_{\un{C}(\ub)\cap \Si_0}$, and the Schwartz
inequality, the integral we are considering is bounded by
\begin{eqnarray}
&&\bigl|\int_{V(u,\ub)}\tttau^6\ttau^{5+\e}\b(\lie_T \tW)\c\Xi(T,\tW)\bigr|\nn\\
&&\leq  c\int_{u_0}^u du'\bigl(\int_{C(u';[\ub_0,\ub])}
\ub'^6\ttau^{5+\e}|\b(\lie_T \tW)|^2\bigr)^{\frac 1 2 }\int_{C(u';[\ub_o,\ub])}\bigl(\ub'^6\ttau^{5+\e}|\Xi(T,\tW)^2|\bigr)^{\frac 1 2 } \nn\\
&&\leq  c\QQ_{\cal{K}}^{\frac 1 2 }\biggl[\sum_{i=1}^3 \int_{u_0}^u du'\c\bigl(\int_{C(u';[\ub_0,\ub])}
\ub'^6\ttau^{5+\e}|\Xi^{(i)}(T,\tW)|^2\bigr)^{\frac 1 2 }\biggr] .
\end{eqnarray}
\NI To complete the proof, we have to prove that the following inequalities
hold
\begin{eqnarray}
\bigl(\ub'^6\ttau^{5+\e}|\Xi^{(1)}(T,\tW)^2|\bigr)^{\frac 1 2 }&\leq & \frac 1
{|u'|^ 2}
\QQ_{\cal{K}}^{\frac 1 2 } \nn\\
\bigl(\ub'^6\ttau^{5+\e}|\Xi^{(2)}(T,\tW)^2|\bigr)^{\frac 1 2 }&\leq & \frac 1
{|u'|^2 }
\QQ_{\cal{K}}^{\frac 1 2 } \nn\\
\bigl(\ub'^6\ttau^{5+\e}|\Xi^{(3)}(T,\tW)^2|\bigr)^{\frac 1 2 }&\leq & \frac 1
{|u'|^2 } \QQ_{\cal{K}}^{\frac 1 2 }\ .\label{stimaLieTW}
\end{eqnarray}
\NI As far as the first integral is concerned, we have to estimate various terms (see the expression for $\Xi^{(1)}(T,\tW)$), which are
all estimated in the same way. As $^{^{(T)}}n$ is the ${^{(T)}}\pih$ component with the  slowest decay, let us control only the terms which involve it,
\begin{eqnarray*}
& &\int_{C(u';[\ub_0,\ub])}
\ub'^6\ttau^{5+\e}|{^{(T)}}n|^2|\b_3(\tW)|^2\\
& &\int_{C(u';[\ub_0,\ub])}
\ub'^6\ttau^{5+\e}|\tr\chib|^2|{^{(T)}}n|^2|\b(\tW)|^2.
\end{eqnarray*}
\NI The first integral can be estimated in
the following way:
\begin{eqnarray*}
\int_{C(u';[\ub_0,\ub])} \ub'^6\ttau^{5+\e}|{^{(T)}}n|^2|\b_3(\tW)|^2&\leq &
c \int_{C(u';[\ub_0,\ub])} \ub'^6\ttau^{5+\e}\frac 1 {r^4}\frac 1 {r^2}|\b(\tW)|^2\\
& \leq & c \frac 1 {u'^4}\int_{C(u';[\ub_0,\ub])}
\frac{\ub'^6}{r^2}\ttau^{5+\e}|\b(\mathcal{L}_O \tW)|^2
\end{eqnarray*}
\NI the first inequality following directly from the asymptotic
behavior of ${^{(T)}}n$ and the second inequality being true, due to the fundamental relation 
\[ \int_{S(u,\ub}|f|^2\leq\int_{S(u,\ub}|\lie_Of|^2 \] see \cite{Ch-Kl1} . Then:
\begin{equation}
\bigl(\ub'^6\ttau^{5+\e}|\Xi^{(1)}(T,\tW)|^2\bigr)^{\frac 1 2 }\leq  \frac 1
{|u'|^2 } \QQ_{\cal{K}}^{\frac 1 2 }
\end{equation}
\NI implying
\begin{eqnarray*}
\bigl|\int_{V(u,\ub)}\tttau^6\ttau^{5+\e}\b(\lie_T \tW)\c\Xi(T,\tW)\bigr| \leq
\frac c {u'} \QQ_{\cal{K}} \leq  \frac c {r_0}\QQ_{\cal{K}},
\end{eqnarray*}
\NI for $r_0$ sufficiently great.\\
\NI To control the second integral of (\ref{stimaLieTW}),  recalling
that
$$
\Xi^{(2)}(T,\tW)=Qr[{^{(T)}}\pp;\a]+ Qr[{^{(T)}}p_4;\b]
$$
\NI we have to estimate the integrals
\begin{eqnarray}
& &\int_{C(u';[\ub_0,\ub])} \ub'^6\ttau^{5+\e}|{^{(T)}}\pp|^2|\a(\tW)|^2\nn\\
& & \int_{C(u';[\ub_0,\ub])} \ub'^6\ttau^{5+\e}|{^{(T)}}p_4|^2|\b(\tW)|^2.
\end{eqnarray}
\NI Let us study the second integral as an example. It is controlled in the following way
\begin{eqnarray*}
\biggl(\int_{C(u';[\ub_0,\ub])}
\ub'^6\ttau^{5+\e}|{^{(T)}}p_4|^2|\b(\tW)|^2\biggr)^{\frac 1 2} &\leq &
c\biggl(\int_{C(u';[\ub_0,\ub])} \frac{\ub^6}{r^6}\ttau^{5+\e}|\b(\tW)|^2\biggr)^{\frac 1 2 }\\
 \leq &  \frac c {u'^2}& \biggl( \int_{C(u';[\ub_0,\ub])} \ub'^4\ttau^{5+\e}|\b(\lie_O \tW)|^2\biggr)^{\frac 1 2} \\
  \leq & \frac c {u'^2} & \QQ_{\cal{K}}^{\frac 1 2}.
 \end{eqnarray*}
\NI To control the third integral of (\ref{stimaLieTW}), let us recall the explicit expression for $\Xi^{(3)}(T,\tW)$:
 $$
 \Xi^{(3)}(T,\tW)= Qr[\a;(I,\un{I})({^{(T)}}q)]+Qr[\b;(K,\Lambda,\Theta)({^{(T)}}q)]+Qr[(\ro,\si);\Xi({^{(T)}}q)]
 $$
\NI Then we have to control the following integral terms:
\begin{eqnarray}
& &\int_{C(u';[\ub_0,\ub])}\tau_+^6\ttau^{5+\e}|(I({^{(T)}}q),\un{I}({^{(T)}}q))|^2|\a(\tW)|^2\nn\\
& &\int_{C(u';[\ub_0,\ub])}\tau_+^6\ttau^{5+\e}|(K({^{(T)}}q),\Lambda
({^{(T)}}q),\Theta({^{(T)}}q))|^2|\b(\tW)|^2\\
& &\int_{C(u';[\ub_0,\ub])}\tau_+^6\ttau^{5+\e}|\Xi({^{(T)}}q)|^2
|(\ro(\tW),\si(\tW))|^2\nn.
\end{eqnarray}
\NI The terms with the worst asymptotic behavior are $I({^{(T)}}q),\un{I}({^{(T)}}q),\Lambda({^{(T)}}q)$, which decay at
null infinity as $\frac c {r^3}$, therefore let us estimate the first integral, in particular let us show we can control
$$
\int_{C(u';[\ub_0,\ub])}\tau_+^6\ttau^{5+\e}|(I({^{(T)}}q),\un{I}({^{(T)}}q))|^2|\a(\tW)|^2.
$$
\NI This integral is bounded by:
$$
\frac c{u'^4}\int_{C(u';[\ub_0,\ub])}\frac{\tau_+^6}{r^6}\tau_+^4\ttau^{5+\e}|\a(\lie_O
\tW)|^2
$$
\NI and, therefore,
$$
\biggl(\int_{C(u';[\ub_0,\ub])}\tau_+^6\ttau^{5+\e}|(I({^{(T)}}q),\un{I}({^{(T)}}q))|^2|\a(\tW)|^2\biggr)^{\frac
1 2 }\leq \frac c {u'^2}\QQ_{\cal{K}}^{\frac 1 2 }.
$$
\end{proof}

\subsubsection{Estimate of $\int_{V(u,v)}\tau_-^{5+\e}Q(\hat{\mathcal{L}}_T \tilde{W})_{\alpha\beta\gamma\delta}
({^{(\bar{K})}}\pi^{\alpha\beta}\bar{K}^\gamma\bar{K}^\delta)$}

\begin{prop}
\NI In the Kerr spacetime the following inequality holds:
\begin{equation}
\int_{V(u,\ub)}\ttau^{5+\e}|Q(\lie_T
\tW)_{\a\b\ga\de}({^{(\bar{K})}}\pi^{\a\b}\bar{K}^\ga\bar{K}^\de)|\leq
\frac {c} {r_0}\QQ_{\cal{K}}
\end{equation}
\end{prop}

\NI The proof can be easily obtained by the estimates for ${^{(\bar{K})}}\pi^{\a\b}$ and we do not report it here.

\subsection{Estimate of $\int_{V(u,v)}\mbox{{\bf Div}} Q(\hat{\mathcal{L}}_O W)_{\beta\gamma\delta}
(\bar{K}^\beta\bar{K}^\gamma\bar{K}^\delta)$}

\NI The estimate for $\int_{V(u,v)}\mbox{{\bf Div}} Q(\hat{\mathcal{L}}_O W)_{\beta\gamma\delta}
(\bar{K}^\beta\bar{K}^\gamma\bar{K}^\delta)$ can be proved in a similar way as the estimate for  $\int_{V(u,v)}\mbox{{\bf Div}}(\hat{\mathcal{L}}_T \tilde{W})_{\beta\gamma\delta}
(\overline{K}^\beta,\overline{K}^\gamma,\overline{K}^\delta)$, by working with the right quantities depending on the rotation vectorfields ${^{(i)}}O$ instead of $T$, it is bounded by $c/r_0\QQ_{K}$ too.

\NI Now we prove the estimate for one of the terms of the other form:
\subsubsection{Estimate of $\int_{V(u,v)}\tau_-^{5+\e}Q(\hat{\mathcal{L}}_O
\tilde{W})_{\alpha\beta\gamma\delta}
({^{(\bar{K})}}\pi^{\alpha\beta}\bar{K}^\gamma T^\delta)$}

\begin{prop}
\NI In the Kerr spacetime the following inequality holds
\begin{equation}
\int_{V(u,\ub)}\ttau^{5+\e}Q(\lie_O
\tW)_{\a\b\ga\de}({^{(\bar{K})}}\pi^{\a\b}\bar{K}^\ga T^\de) \leq
\frac c {r_0^\frac 3 2}\QQ_{\cal{K}}.
\end{equation}
\end{prop}
\begin{proof}
\NI The proof is a straightforward calculation, the result is obtained
starting from the identity
\begin{eqnarray*}
{^{(\bar{K})}}\pi^{\a\b}Q(\lie_O \tW)_{\a\b\ga\de}\bar{K}^\ga T^\de
&=& {^{(\bar{K})}}\pi^{\a\b}\bigl[\tttau^2(Q(\lie_O
\tW)_{\a\b44}+Q(\lie_O \tW)_{\a\b43})\\
&+& \ttau^2(Q(\lie_O \tW)_{\a\b34}+Q(\lie_O \tW)_{\a\b33})\bigr],
\end{eqnarray*}
\NI by writing explicitly the various term of the integrand (see
\cite{Kl-Ni2}, (6.2.40)-(6.2.43)). Let us discuss in detail one of them, in particular let us check the boundedness of
$$
\int_{V(u,\ub)}\tau_+^2\ttau^{5+\e}|\ro(\lie_O \tW)||\a(\lie_O \tW)||{^{(\bar{K})}}i|\ .
$$
\NI Applying the Schwartz inequality, and recalling the decay of ${^{(\bar{K})}}i$, we obtain the following estimate
\begin{eqnarray*}
&&\int_{V(u,\ub)}\tau_+^2\ttau^{5+\e}|\ro(\lie_O \tW)||\a(\lie_O
\tW)||{^{(\bar{K})}}i| \\
&&\leq \biggl(\int_{V(u,\ub)}\ttau^{5+\e}|\ro(\lie_O
\tW)|^2|{^{(\bar{K})}}i|^2\biggr)^{\frac 1
2}\biggl(\int_{V(u,\ub)}\tau_+^4\ttau^{5+\e}|\a(\lie_O
\tW)|^2\biggr)^{\frac 1 2}\\
&\leq & c \biggl(\int_{u_0}^u du' \int_{C(u';[\ub_o,\ub])}|\frac
{\log r}{r}\ttau^{5+\e}\ro(\lie_O
\tW)|^2|\frac{r}{\log r}{^{(\bar{K})}}i|^2\biggr)^{\frac 1 2 }\\
&\c &\biggl(\int_{u_0}^u
du'\int_{C(u';[\ub_o,\ub])}\tau_+^4\ttau^{5+\e}|\a(\lie_O
\tW)|^2\biggr)^{\frac 1 2 }\\
&\leq & c\QQ_{\cal{K}}^{\frac 1 2 }\biggl(\sup_{\cal{K}}|\frac r
{\log
r}{^{(\bar{K})}}i|\biggr)\biggl(\sup_{\cal{K}}\int_{C(u';[\ub_0,\ub])}\tau_+^4\ttau^{5+\e}|\ro(\lie_O
\tW)|^2\biggr)\biggl(\int_{u_0}^u du' \frac {(\log
r)^2}{r^6}\biggr)^{\frac 1 2}\\
&\leq & c\QQ_{\cal{K}}\biggl(\int_{u_0}^{u}du' \frac1
{u'^4}\biggr)^{\frac 1 2}
\leq  \frac c {r_0^{\frac 3 2 }}\QQ_{\cal{K}}\ .
\end{eqnarray*}
\end{proof}
\NI We state the following:
\begin{prop}
\NI In the Kerr spacetime the following estimate holds
\begin{equation}
\int_{V(u,\ub)}\ttau^{5+\e}|Q(\lie_O
\tW)_{\a\b\ga\de}({^{(T)}}\pi^{\a\b}\bar{K}^\ga \bar{K}^\de)|\leq
\frac {c}{r_0}\QQ_{\cal{K}}.
\end{equation}
\end{prop}

\subsection{The error term \mbox{${\mathcal E}_2$}}

\NI In oorder to estimate the part of the norms involving two Lie derivatives of
the Weyl field $\tW$, we need some estimates about the behavior of
${^{(S)}}\pi$ (given in the proposition \ref{S}) and about the components of the currents $J(S,W)$. We give them in the
following
\begin{prop}
\NI From the explicit expression of ${^{(S)}}p_3$, ${^{(S)}}p_4$ and
${^{(S)}}\pp$, we obtain the following estimates for any $S\subset
\mathcal{K}$ with $p\in[2,4]$:
\begin{eqnarray*}
&||\frac{r^{2-\frac 2 p}}{\log r}{^{(S)}}p_3||_{p,S}&\leq c \\
&||\frac{r^{2-\frac 2 p}}{\log r}{^{(S)}}p_4||_{p,S}&\leq c \\
&||{r^{3-\frac 2 p}}{^{(S)}}\pp_a||_{p,S}&\leq c.
\end{eqnarray*}
\end{prop}
\begin{prop}\label{(S)q}
\NI The null components of the $S$ current $J^3$ satisfy the following
estimates for any $S\subset\mathcal{K}$ with $p\in[2,4]$:
\begin{eqnarray*}
&&|\frac {r^{2-\frac 2 p}}{\log r}\Lambda({^{(S)}}q)|_{p,S}\leq c\\
&&|\frac {r^{3-\frac 2 p}}{\log r}K({^{(S)}}q)|_{p,S}\leq c\\
&&|\frac {r^{4-\frac 2 p}}{\ttau}\Xi({^{(S)}}q)|_{p,S}\leq c\\
&&|r^{3-\frac 2 p}I({^{(S)}}q)|_{p,S}\leq c\\
&&|\frac{r^{2-\frac 2 p}}{\log r}\Theta({^{(S)}}q)|_{p,S}\leq c
\end{eqnarray*}
\NI and
\begin{eqnarray*}
&&|\frac {r^{2-\frac 2 p}}{\log r}\un{\Lambda}({^{(S)}}q)|_{p,S}\leq c\\
&&|\frac {r^{3-\frac 2 p} }{\log r}\un{K}({^{(S)}}q)|_{p,S}\leq c\\
&&|\frac {r^{4-\frac 2 p}}{\ttau}\un{\Xi}({^{(S)}}q)|_{p,S}\leq c\\
&&|r^{3-\frac 2 p}\un{I}({^{(S)}}q)|_{p,S}\leq c\\
&&|\frac{r^{2-\frac 2 p}}{\log r}\un{\Theta}({^{(S)}}q)|_{p,S}\leq
c.
\end{eqnarray*}
\end{prop}
\NI Moreover, we also need to know the relations between the modified Lie derivative of $W$ null
components done with respect to the
vector field $S$ and the null components of the tensor field $\lie_S W$,
i.e. we need to apply proposition 2.2.1 of \cite{Ch-Kl1} which relates $R(\lie_X \tW) = \lie_X R(\tW)$, $R$ the generic null Riemann component, when $X=S$. More precisely, we claim the following proposition holds true:
\begin{prop}
\NI Given the vector field $S=\frac 1 2(\ub e_4+ue_3)$ and a Weyl field $\tW$
satisfying the Bianchi equations, the following relations hold at the highest order:
\begin{eqnarray}\label{stimeS}
\a(\lie_S \tW) &=& \lieb_S \a(\tW)+c \a(\tW)\nn\\
\b(\lie_S \tW) &=& \lieb_S \b(\tW)\nn\\\
\ro(\lie_S \tW) &=& \lieb_S\ro(\tW)\nn\\
\si(\lie_S \tW) &=& \lieb_S\si(\tW)\\
\bb(\lie_S \tW) &=& \lieb_S\bb(\tW)\nn\\
\aa(\lie_S \tW) &=& \lieb_S\aa(\tW).\nn
\end{eqnarray}
\end{prop}
\NI At last, as far as the covariant
derivative respect to $S$ is concerned, we will need a further relation connecting it  to the $S$
Lie derivative, that is:
\begin{equation}\label{e3}
\dddd_S\a_{ab}=\lieb_S\a_{ab}+l.o.t.
\end{equation}
\NI Then we need some decays for the Lie derivatives of null components of the currents done with respect to $O$; in particular
the following propositions hold:
\begin{prop}\label{norme1suC}
\NI Based on proposition (\ref{omega}) and on corollary
(\ref{omega'}), the following estimates hold:
\begin{eqnarray*}
&&||r^{\frac 3 2 -\e
}\hat{\mathcal{L}}_O{^{(O)}}p_3||_{L_2(\un{C}(\ub)\cap
V(u,\un{u}))}\leq c\\
&&||r^{\frac 3 2-\e
}\hat{\mathcal{L}}_O{^{(O)}}p_4||_{L_2(\un{C}(\ub)\cap
V(u,\un{u}))}\leq c\\
&&||r^{\frac 3 2-\e
}\hat{\mathcal{L}}_O{^{(O)}}\pp_a||_{L_2(\un{C}(\ub)\cap
V(u,\un{u}))}\leq c.\\
\end{eqnarray*}
\end{prop}
\begin{prop}\label{norme2suC}
\NI The modified Lie derivative of the null components of $J^3(O,W)$
made with respect to the rotation vector fields $O$ satisfy the
following asymptotic estimates:
\begin{eqnarray*}
&&||r^{\frac 3 2 -\e }\lie_O\Th({^{(O)}}q)||_{L_2(\un{C}(\ub)\cap
V(u,\un{u}))}\leq c\\
&&||r^{\frac 3 2 -\e
}\lie_O\Lambda({^{(O)}}q)||_{L_2(\un{C}(\ub)\cap
V(u,\un{u}))}\leq c\\
&&||r^{\frac 3 2 -\e }\lie_O K({^{(O)}}q)||_{L_2(\un{C}(\ub)\cap
V(u,\un{u}))}\leq c\\
&&||r^{\frac 3 2 -\e }\lie_O I({^{(O)}}q)||_{L_2(\un{C}(\ub)\cap
V(u,\un{u}))}\leq c\\
&&||r^{\frac 3 2 -\e }\lie_O\Xi({^{(O)}}q)||_{L_2(\un{C}(\ub)\cap
V(u,\un{u}))}\leq c\\.
\end{eqnarray*}
\NI The underlined quantities satisfy the same inequalities as the previous ones.
\end{prop}
\noindent $\EEbb$ collects the error terms associated to the integrals $\QQ_2$ and $\QQb_2$, in particular it has the following form:
\begin{eqnarray*}
\EEbb(u,\ub)&=&\int_{V(u,\ub)} \ttau^{5+\e}\Div Q(\lie_O\lie_T
\tW)_{\b\ga\de}(\bar{K}^\b\bar{K}^\ga\bar{K}^\de)\\
& &+\int_{V(u,\ub)} \ttau^{5+\e}\Div Q(\lie_O^2
\tW)_{\b\ga\de}(\bar{K}^\b\bar{K}^\ga T^\de)\\
& &+\int_{V(u,\ub)} \ttau^{5+\e}\Div Q(\lie_S\lie_T
\tW)_{\b\ga\de}(\bar{K}^\b\bar{K}^\ga\bar{K}^\de)\\
& &+\frac 3 2 \int_{V(u,\ub)}\ttau^{5+\e}Q(\lie_O\lie_T
\tW)_{\a\b\ga\de}({^{(\bar{K})}}\pi^{\a\b}\bar{K}^\ga\bar{K}^\de)\\
& &+\frac 3 2 \int_{V(u,\ub)}\ttau^{5+\e}Q(\lie_S\lie_T
\tW)_{\a\b\ga\de}({^{(\bar{K})}}\pi^{\a\b}\bar{K}^\ga\bar{K}^\de)\\
& &+ \int_{V(u,\ub)}\ttau^{5+\e}Q(\lie_O^2
\tW)_{\a\b\ga\de}({^{(\bar{K})}}\pi^{\a\b}\bar{K}^\ga T^\de)\\
& &+\frac 1 2 \int_{V(u,\ub)}\ttau^{5+\e}Q(\lie_O^2
\tW)_{\a\b\ga\de}({^{(T)}}\pi^{\a\b}\bar{K}^\ga\bar{K}^\de).
\end{eqnarray*}
\begin{remark}
\NI As far as the study of $\EEbb$ is concerned, we note that it is made by many terms, but most of them can be treated as
the corresponding ones studied in the previous section.
\end{remark}
\NI First of all, given $X,Y$ two vector fields on $T\mathcal{M}$, let us define the following quantity:
\begin{equation}\label{J1derivate}
J(X,Y;\tW)=J^0(X,Y;\tW)+\frac 1
2\bigl(J^1(X,Y;\tW)+J^2(X,Y;\tW)+J^3(X,Y;\tW)\bigr),
\end{equation}
\NI where
\begin{eqnarray}\label{J2derivate}
J^0(X,Y;\tW) &=& \lie_X J(Y;\tW)\nn\\
J^1(X,Y;\tW) &=& J^1(X;\lie_Y \tW)\nn\\
J^2(X,Y;\tW) &=& J^2(X;\lie_Y \tW)\\
J^3(X,Y;\tW) &=& J^3(X;\lie_Y \tW).\nn
\end{eqnarray}
\NI Its null components
$\Th(X,Y;\tW),...,\un{\Xi}(X,Y;\tW)$ have the following structure:
$$
F(X,Y;\tW)=F^0(X,Y;\tW)+\frac 1
2\bigl(F^1(X,Y;\tW)+F^2(X,Y;\tW)+F^3(X,Y;\tW)\bigr)
$$
\NI and
\begin{equation}
F^0(X,Y;\tW)=\frac 1 2 \bigl[\lie_X F^1(Y;\tW)+\lie_X F^2(Y;\tW)+\lie_X
F^3(Y,\tW)\bigr].
\end{equation}
\NI By a straightforward calculation, the following decomposition for $\Div Q(\lie_X,\lie_Y \tW)$ is true (see \cite{Kl-Ni2},
propositions 7.1.1, 7.1.2),
\begin{eqnarray}\label{divQ}
\Div Q(\lie_X\lie_Y \tW)_{\b\ga\de} &=& {{{(\lie_X\lie_Y
\tW)_\b}^\mu}_\de}^\nu J(X,Y;\tW)_{\mu\ga\nu}+{{{(\lie_X\lie_Y
\tW)_\b}^\mu}_\ga}^\nu\nn\\
&\c & J(X,Y;W)_{\mu\de\nu}+{^*}{{{(\lie_X\lie_Y \tW)_\b}^\mu}_\de}^\nu
{J(X,Y;\tW)^*}_{\mu\ga\nu}\nn\\
&+&{^*}{{{(\lie_X\lie_Y)_\b}^\mu}_\ga}^\nu {J(X,Y;\tW)^*}_{\mu\de\nu},
\end{eqnarray}
\NI where $J(X,Y;\tW)$ is defined by (\ref{J1derivate}),
(\ref{J2derivate}). These new quantities will appear in the estimate of the terms involving the
divergence, choosing $X,Y$ between $\{O,T,S\}$ suitably.\\
We show only the boundness of the term that involves $Q(\lie_S\lie_T W)$ and we state the others.

\subsubsection{Estimate of $\int_{V(u,v)} \tau_-^{5+\e}\mbox{{\bf
Div}} Q(\hat{\mathcal{L}}_S\hat{\mathcal{L}}_T
\tilde{W})_{\beta\gamma\delta}(\bar{K}^\beta\bar{K}^\gamma\bar{K}^\delta)$}

\NI This time we have to study the terms involving
$$
J(T,S;\tW)=J^0(T,S;\tW)+\frac 1 2
\bigl(J^1(T,S;\tW)+J^2(T,S;\tW)+J^3(T,S;\tW)\bigr).
$$
\NI Proceeding as in subsection \ref{section1}, we have to control
\begin{eqnarray*}
&& \int_{V(u,\ub)}\tttau^6\ttau^{5+\e} D(T,S;\tW)_{444},\quad
\int_{V(u,\ub)}\tttau^4\ttau^2\ttau^{5+\e} D(T,S;\tW)_{344}\\
&& \int_{V(u,\ub)}\tttau^2\ttau^4\ttau^{5+\e} D(T,S;\tW)_{334},\quad
\int_{V(u,\ub)}\ttau^6\ttau^{5+\e} D(T,S;\tW)_{333}.
\end{eqnarray*}
\NI Let us examine only the first one. Since the following expression holds:
\begin{eqnarray}\label{eq}
D(T,S;\tW)_{444}=4\a(\lie_S\lie_T \tW)\c\Theta(T,S;\tW)-8\b(\lie_S\lie_T \tW)\c\Xi(T,S;\tW)\ ,\nn
\end{eqnarray}
\NI where
$$
\Theta(T,S;\tW)=\Theta^0(T,S;\tW)+\frac 1
2\bigl(\Theta^1(T,S;\tW)+\Theta^2(T,S;\tW)+\Theta^3(T,S;\tW)\bigr)
$$
\NI (and analogously for $\Xi(T,S;\tW)$), we consider first the terms with $i$ = 1, 2, 3 of the first term of (\ref{eq})
$$
\int_{V(u,\ub)}\tttau^6\ttau^{5+\e}\a(\lie_S\lie_T \tW)\c \Th^i(T,S;\tW)\	.
$$
\NI Proceeding as in the case of the subsection \ref{section1}, we
have
\begin{eqnarray*}
&&\int_{V(u,\ub)}\tttau^6\ttau^{5+\e}\a(\lie_S\lie_T \tW)\c \Th^i(T,S;\tW)\\
&& \leq c
\int_{u_0}^udu'\biggl(\int_{C(u';[\ub_0,\ub])}\ub'^6\ttau^{5+\e}|\a(\lie_S\lie_T
\tW)|^2\biggr)^{\frac 1
2}\biggl(\int_{C(u';[\ub_0,\ub])}\ub'^6\ttau^{5+\e}|\Th^i(T,S;\tW)|^2\biggr)^{\frac
1 2}.
\end{eqnarray*}
\NI The first integral is estimated by
$$
\sup_{\cal{K}}\biggl(\int_{C(u';[\ub_0,\ub])}\ub'^6\ttau^{5+\e}|\a(\lie_S\lie_T
\tW)|^2\biggr)^{\frac 1 2}\leq c \QQ_2^{\frac 1 2 }.
$$
As far as the second integral is concerned, we compare it with the
integral
$$
\int_{u_0}^u
du'\biggl(\int_{C(u';[\ub_0,\ub])}\ub'^6\ttau^{5+\e}|\Th^i(T,\tW)|^2\biggr)^{\frac
1 2}\	,
$$ 
\NI which has been estimated in Subsection \ref{section1}, noting that the following analogies and differences hold:\\
\NI i) In the present case the deformation tensors and their derivatives refer to the vector field $S$ and not $T$ and the null components and their
derivatives are relative to $\lie_T \tW$ instead of $\tW$.\\
\NI ii) The estimates of null components on $\lie_T \tW$ are better than
those relative to the null components of $\tW$ by a factor $r$,
while the asymptotic behavior of ${^{(S)}}\un{m},{^{(S)}}\un{n}$ is worst of a factor $r$, and that one of ${^{(S)}}n$ is the same
(both respect to the corresponding ${^{(T)}}\pi$ null components),
it follows that these terms are under control.\\
\NI iii) We have only to check the boundedness of the terms involving
${^{(S)}}i$ and ${^{(S)}}j$, because of ${^{(T)}}i,{^{(T)}}j=0$.\\
\NI As far as we consider the part $\Th^1(S;\lie_T \tW)$, they are the following:
\begin{eqnarray*}
&& Qr[{^{(S)}}i;\nabb\b(\lie_T \tW)], \quad
Qr[{^{(S)}}j;(\a_3,\ro_4,\si_4)(\lie_T \tW)]\\
&& \tr\chi Qr[{^{(S)}}j;(\ro,\si)(\lie_T \tW)],\quad \tr\chib
Qr[{^{(S)}}j;\a(\lie_T \tW)]\ .
\end{eqnarray*}
\NI They behave all in the same way asymptotically, so, let us consider
only one of them, in particular the first:
\begin{eqnarray*}
&& \int_{u_0}^u
du'\biggl(\int_{C(u';[\ub_0,\ub])}\ub'^6\ttau^{5+\e}|{^{(S)}}i|^2|\nabb\b(\lie_T
\tW)|^2\biggr)^{\frac 1 2}\\
&& \leq c \int_{u_0}^u
du'\biggl(\int_{C(u';[\ub_0,\ub])}\tttau^4|\b(\lie_O\lie_T
\tW)|^2\frac {(\log r)^2}{r^2}\biggr)^{\frac 1 2}\\
&& \leq c\int_{u_0}^u
du'\biggl(\int_{C(u';[\ub_0,\ub])}\tttau^6|\b(\lie_O\lie_T
\tW)|^2\frac {(\log r)^2}{r^4}\frac {r^2}{(\log r)^2}\biggr)^{\frac
1 2}\\
&& \leq \frac{c}{r_0} \QQ_2^{\frac 1 2 }.
\end{eqnarray*}
As far as $\Th^2$ is concerned, let's recall its expression:
$$
\Th^2(S,T;\tW)= Qr[{^{(S)}}p_3;\a(\lie_T
\tW)]+Qr[{^{(S)}}\pp;\b(\lie_T \tW)]+Qr[{^{(S)}}p_4;(\ro,\si)(\lie_T
\tW)].
$$
\NI We note that the term ${^{(S)}}\pp$ behave as the corresponding
one relative to $T$, so the part related to the second term in the
above sum is controlled. As far as the first and the third terms
are concerned, since the asymptotic behavior of ${^{(S)}}p_4$ is
the same as ${^{(S)}}p_3$ (which is $O(\frac{\log r} {r^2}))$, we
consider as an example only the third term. Then the following estimate holds:
\begin{eqnarray*}
&&
\int_{u_0}^udu'\biggl(\int_{C(u';[\ub_0,\ub])}\ub'^6\ttau^{5+\e}|\a(\lie_S\lie_T
\tW)|^2\biggr)^{\frac 1 2}
\biggl(\int_{C(u';[\ub_0,\ub])}\ttau^{5+\e}|{^{(S)}}p_4|^2|\ro(\lie_T
\tW)|^2\biggr)^{\frac 1 2}\\
&& \leq c\QQ_{\cal{K}}^{\frac 1
2}\biggl(\int_{u_0}^udu'\int_{C(u';[\ub_0,\ub])}\frac{(\log
r)^2}{r^4}\frac {r^2}{(\log r)^2}\ttau^{5+\e}|\ro(\lie_T \tW)|^2\biggr)^{\frac 1
2} \leq c \QQ_{\cal{K}}.
\end{eqnarray*}
\NI Finally, to control the part of the current $J^3$, recaliing the asymptotic behavior of the null components of ${^{(S)}}q$ (see
proposition \ref{(S)q}), it follows that it is estimated in the same way as the previous one. This concludes the estimate of
$\int_{V(u,\underline{u})} \tau_-^{5+\e}\Div
Q(\hat{\mathcal{L}}_S\hat{\mathcal{L}}_T \tW)_{\beta\gamma\delta}
(\bar{K}^\beta\bar{K}^\gamma\bar{K}^\delta)$.\\

\NI As far as the other terms involving  $\Div(Q)$, it is shown they are controlled by the $\QQ_2$ norms too.

\subsubsection{Estimate of the remaining terms}

\NI As far as the other terms of $\EEbb$ are concerned,
it is enough to observe that they are treated as the corresponding terms present in $\EEb$, with
the obvious substitutions.
In fact:
$$
\int_{V(u,\ub)}\ttau^{5+\e}Q(\lie_O^2
\tW)_{\a\b\ga\de}({^{(\bar{K})}}\pi^{\a\b}\bar{K}^\ga T^\de)
$$
\NI is treated as
$$
\int_{V(u,\ub)}\ttau^{5+\e}Q(\lie_O
\tW)_{\a\b\ga\de}({^{(\bar{K})}}\pi^{\a\b}\bar{K}^\ga T^\de)
$$
\NI and it is estimated by $\QQ_{2}$ instead of $\QQ_1$.\\
The term
$$
\int_{V(u,\ub)}\ttau^{5+\e}Q(\lie_O^2
\tW)_{\a\b\ga\de}({^{(T)}}\pi^{\a\b}\bar{K}^\ga\bar{K}^\de)
$$
\NI is of the same form of
$$
\int_{V(u,\ub)}\ttau^{5+\e}Q(\lie_O
\tW)_{\a\b\ga\de}({^{(T)}}\pi^{\a\b}\bar{K}^\ga\bar{K}^\de)
$$
\NI by substituting $\QQ_1$ with $\QQ_2$.\\
The term
$$
\int_{V(u,\ub)}\ttau^{5+\e}Q(\lie_O\lie_T
\tW)_{\a\b\ga\de}({^{(\bar{K})}}\pi^{\a\b}\bar{K}^\ga\bar{K}^\de)
$$
\NI is estimated in the same way as
$$
\int_{V(u,\ub)}\ttau^{5+\e}Q(\lie_T
\tW)_{\a\b\ga\de}({^{(\bar{K})}}\pi^{\a\b}\bar{K}^\ga\bar{K}^\de).
$$
\NI The final result is the same with the obvious substitutions of the
quantities $\QQ_1$ with $\QQ_2$.\\
The estimate of the integral
$$
\int_{V(u,\ub)}\ttau^{5+\e}Q(\lie_S\lie_T
\tW)_{\a\b\ga\de}({^{(\bar{K})}}\pi^{\a\b}\bar{K}^\ga\bar{K}^\de)
$$
\NI is made exactly in the same way as the estimate of
$$
\int_{V(u,\ub)}\ttau^{5+\e}Q(\lie_T
\tW)_{\a\b\ga\de}({^{(\bar{K})}}\pi^{\a\b}\bar{K}^\ga\bar{K}^\de)
$$
\NI with the substitutions of $\QQ_1$ with $\QQ_2$.

\section{Conclusions and developments}
\NI With the estimates of the error term we can consider proved the estimates \ref{intes}. As remarked in the introduction they are not yet in accordance with the Peeling theorem.
the last step is to exploit the Bianchi equations to gain extra decays in $r$ from the decays in $u$.
As already said we will not perform this calculation as it is identical to the ones in \cite{Kl-Ni1} chapter 6.

\medskip

\NI As we said in the introduction the \cite{Ch-Kl:book}  approach consists in a decoupling between the structure equations for the connection coefficients $\cal {O}$,  where the null Riemann components $\cal{R}$ have to be considered assigned, and, viceversa, the Bianchi equations for the $\cal{R}$ components, where the  $\cal{O}$ coefficients have to be considered assigned. As already said, we have considered in this paper the second part 
of this decoupling. The next step is to consider the first part, namely, we have to calculate the decays of the $\cal{O}$ components by the structure equations where the $\cal{R}$ have the decays obtained in theorem
\ref{maintheorem}, let us cal them $\cal{R_1}$, and the the connection coefficients $\cal{O_1}$. In this way we obtain the map:
\[F:({\cal O}_0,{\cal R}_0)\rightarrow ({\cal O}_1,{\cal R}_1).\] 
Where $\cal{O_0,R_0}$ are the connection coefficients and the null Riemann components associated to the Kerr spacetime. Repeating this procedure we obtain a sequence of solutions:
\[F_i:({\cal O}_i,{\cal R}_i)\rightarrow( {\cal O}_{i+1},{\cal R}_{i+1}).\]
The crucial step to obtain a perturbation of the kerr spacetime is to demonstrate that  the map $F$ is a contraction that do exists a ball $B$, defined for suitable Banach spaces such that
\[F(B)\subset B\] to exploit the fixed point theorem.

\section{Appendix}

\subsection{Estimate of $\chi_{\theta\theta}$}

Let us calculate explicitly the connection coefficient
$\chi_{\theta\theta}=g(D_{e_\theta}e_4,e_\theta)$:
\begin{eqnarray*}
&& \chi_{\theta\theta}=\frac Q {\Si R}\bigl[\p_\theta\bigl(\frac
{\sqrt{\De}P}{\Si
R}\bigr)g(\p_\theta,e_\theta)+\p_\theta\bigl(\frac{\sqrt{\De
}Q}{\Si R}g(\p_r,e_\theta)\bigr]\\
&&+\Ga^{\theta}_{\ro\si}e_4^\ro e_\theta^\si
g(\p_\theta,e_\theta)-\frac {\De P}{\Si
R}\bigl[\p_r\bigl(\frac{\sqrt{\De}P}{\Si
R}\bigr)g(\p_\theta,e_\theta)\\
&& +\p_r\bigl(\frac{\sqrt{\De}Q}{\Si
R}\bigr)g(\p_r,e_\theta)\bigr]+\Ga^{r}_{\ro\si}e_4^\ro
e_\theta^\si g(\p_r,e_\theta),
\end{eqnarray*}
\NI being the only terms different from 0, as $\p_\theta$ and $\p_r$
the only coordinate vector fields not orthogonal to $e_\theta$.
Then:
\begin{eqnarray*}
&& \chi_{\theta\theta}=\frac {Q^2}{\Si R}O\bigl(\frac{\p_\theta
P}{r^2}\bigr)-\frac{\sqrt{\De}Q P}{\Si
R^3}\bigl[2a^2\sin\theta\cos\theta-\frac 1 2\p_\theta\lambda\bigr]
\frac 1 {r^2}\\
&& + \bigl[\frac r \Si \bigl(\frac{\sqrt{\De}Q^2}{\Si^2
R^2}-\frac{\De^{\frac 3 2 }P^2}{\Si^2
R^2}\bigr)-\frac{2a^2\sin\theta\cos\theta}{\Si}\frac{\sqrt{\De}P
Q}{\Si^2 R^2}\bigr]\frac Q R \\
&& \frac {\De P}{\Si R} O\bigl(\frac P {r^3}\bigr)\frac Q R +\frac
{M\De P^2 Q}{\Si^2 R^2
r^2}+\bigl[-\frac{a^2\sin\theta\cos\theta}{\Si}\\
&& \bigl(\frac{\sqrt{\De}Q^2}{\Si^2 R^2}-\frac{\De^{\frac 3 2
}P^2}{\Si^2 R^2}\bigr)\frac \De \Si r
\frac{\sqrt{\De}PQ}{R^2\Si^2}\\
&& +\frac {Mr^2}{\Si\De}\frac{\De^{\frac 3 2} Q P}{\Si^2
R^2}\bigr]\bigl(-\frac P R\bigr)\\
&& = \frac{\p_\theta
P}{r^2}+O\bigl(\frac{2P(a^2\sin\theta\cos\theta-\frac
{\p_\theta\lambda}{2})}{r^4}\bigr)+O\bigl(\frac 1 r
\bigr)+O\bigl(\frac{P^2}{r^3}\bigr)\\
&& +O\bigl(\frac{-2P
a^2\sin\theta\cos\theta}{r^4}\bigr)+O\bigl(-\frac{P^2}{r^3}
\bigr)+O\bigl(\frac{MP^2}{r^4}\bigr)\\
&&
+O\bigl(\frac{Pa^2\sin\theta\cos\theta}{r^4}\bigr)+O\bigl(\frac{a^2
P^3}{r^6}\bigr)+O\bigl(\frac{P^2}{r^3}\bigr)+O\bigl(-\frac{P^2
M}{r^4} \bigr),
\end{eqnarray*}
\NI where we have used the expressions in a power series centered
at the point $\frac 1 r $ when $r \rightarrow\infty$ for the
following quantities:
\begin{eqnarray*}
\frac {\sqrt{\De}} R &=& 1-\frac M r +\frac{a^2 M}{r^3}(\frac 1 2
-\sin^2\theta)\\
\frac Q \Si &=& 1+\frac{a^2(\sin^2\theta-\frac \la 2 )}{r^2}\\
\frac \De \Si &=& 1-\frac M r +\frac{a^2\sin^2\theta}{2r^2}\\
\frac Q {R^2} &=& 1-\frac{a^2\la}{2r^2}\\
\end{eqnarray*}
The term $O\bigl(\frac 1 r \bigr)$ derives from the highest order
of $\chi_{\theta\theta}$, that is $\frac {\sqrt{\De}Q^3 r}{\Si^3
R^3}$. Developing it as a power series, its first terms result
to be:
$$
\frac 1 r -\frac M {r^2}+\frac {3a^2}{r^3}\bigl(\sin^2\theta-\frac
\la 2 \bigr),
$$
\NI and so at the higher orders, the component $\chi_{\theta\theta}$
assumes the following form:
\begin{equation}
\chi_{\theta\theta}=\frac 1 r-\frac M {r^2}+\frac {P^2}{r^3}.
\end{equation}
\newpage
\subsection{Estimate of ${^{^{(2)}(O)}}n$}

\begin{eqnarray*}
{^{^{(2)}O}}n &=& 2g(D_{e_4}{^{(2)}O},e_4)= \frac
{4Mra}{\sqrt{\De}\Si
R}\p_\phi(-\cos\phi)g(\p_{\theta},e_4)+2\Ga^t_{\ro\si}{^{(2)}O}^\ro
e_4^\si g(\p_t,e_4)\\
&& +2\Ga^r_{\ro\si}{^{(2)}O}^\ro e_4^\si
g(\p_r,e_4)+2\Ga^\theta_{\ro\si}{^{(2)}O}^\ro e_4^\si
g(\p_\theta,e_4)\\
&& = \frac {M P a r\sin\phi}{\Si
R^2}-2\bigl[O\bigl(\frac{2a^3M\sin^3\theta\cos\theta}{r^3}\bigr)\bigl(-\frac{2Mra\cos\phi}{\sqrt{\De}{\Si
R}}\\
&& +\frac{\sqrt{\De}P\sin\phi\cos\theta}{\Si R\sin\theta}\bigr)
+O\bigl(\frac{2M
a^2\sin\theta\cos\theta}{r^3}\bigl)\cos\phi\frac{R}{\sqrt{\De}}
\bigr]\\
&&
+\frac{2Q}{\sqrt{\De}R}\bigl[O\bigl(-\frac{Ma\sin^2\theta}{r^2}\bigr)\sin\phi\cot\theta\frac
R
{\sqrt{\De}}+\frac{a^2\sin^2\theta\cos\theta}{\Si}\cos\phi\frac{\sqrt{\De}Q}{\Si
R}\\
&& +\frac {\De}{\Si}r\cos\phi\frac{\sqrt{\De}P}{\Si
R}+O\bigl(-r\sin^2\theta\bigr)\sin\phi\cot\theta\frac {2M r a
}{\sqrt{\De}\Si R}\bigr]+\frac {2\sqrt{\De}P} R\\
&& \c\bigl[\frac{2M
a\sin\theta\cos\theta\sin\phi}{r^3}\cot\theta\frac{R}{\sqrt{\De}}-\frac
r \Si \cos\phi\frac Q \Si \frac {\sqrt{\De}}
R\\
&&
+\frac{a^2\sin\theta\cos\theta}{\Si}\cos\phi\frac{\sqrt{\De}P}{\Si
R} - \frac {2M r a \sin\phi\cos^2\theta}{\sqrt{\De}\Si R}\bigr].\\
\end{eqnarray*}
The highest order terms come from:
$$
g(\p_r,e_4)\bigl(\Ga^r_{t\phi}e_4^t{^{(2)}O}^\phi+\Ga^r_{\theta
r}{^{(2)}O}^\theta e_4^r\bigr).
$$
\NI Therefore, controlling the explicit expression of $\Ga^r_{t\phi}$ and
$\Ga^r_{\theta r}$ at the highest order,it follows that
$$
{^{^{(2)}O}}n = O\bigl(\frac{2a\sin\theta\cos\theta}{r^2}(a\cos\phi-M\sin\phi)\bigr)\ .
$$


\end{document}